\documentclass[final]{article}
\pdfoutput=1
\usepackage[utf8]{inputenc}
\usepackage{jheppub}

\usepackage{amsmath,amsthm,amsfonts,amssymb,amscdx}
\usepackage{multirow}
\usepackage{longtable}
\usepackage{mathtools}
\usepackage{upgreek}

\usepackage{tikz}
\usetikzlibrary{calc}		% Calculations in tikz
\usetikzlibrary{math} % needed tikz library for setting variables
\usetikzlibrary{intersections} 

\DeclareMathAlphabet{\mathpzc}{OT1}{pzc}{m}{it}
\usepackage[mathscr]{eucal}
\newcommand{\sets}[1]{\mathscr{#1}}  % global font for sets

\newcommand{\figL}{[Left]}
\newcommand{\figM}{[Middle]}
\newcommand{\figR}{[Right]}

\newcommand{\A}{\mathcal{A}}

\newcommand{\F}{\sets{F}} % set of sigma flows
\newcommand{\G}{\sets{G}}  % set of tau flows
\newcommand{\I}{\mathcal{I}}  % scri
\newcommand{\M}{\mathcal{M}}  % bulk
\newcommand{\N}{\mathcal{N}}  % bdy

\newcommand{\R}{\mathbf{R}}

\newcommand{\T}{\mathcal{T}}

\newcommand{\W}{\mathcal{W}}
\newcommand{\be}{\begin{equation}}
\newcommand{\ee}{\end{equation}}
\DeclareMathOperator{\area}{area}
\DeclareMathOperator*{\infp}{inf\vphantom{p}}

\newcommand{\ip}[2]{|#1\mathclap{\hspace{6.75pt}\cdot}{\wedge}#2| }
\newcommand{\bkslice}{\sigma}  % bulk slice
\newcommand{\bdyslice}{\Sigma}  % boundary slice
\newcommand{\bksliceset}{\sets{S}}  % set of all bulk slices
\newcommand{\ts}{\tau} % time sheet
\newcommand{\tsset}{\sets{T}}  % set of all time sheets
\newcommand{\eowsurf}{\surf^0}  % EoW surf.
\newcommand{\surf}{\gamma}  % co-dim 2 achronal surface
\newcommand{\Splus}{S_+}  % minimax value
\newcommand{\Sminus}{S_-}  % maximin value
\newcommand{\Sc}{S_{\rm c}}  % convex relaxed maximin = minimax
\newcommand{\SHRT}{S_{_{\text{HRT}}}} 
\newcommand{\HRT}{{\surf_{_{\text{HRT}}}}}  % HRT surface
\newcommand{\hHRT}{{\hat{\surf}_{_{\text{HRT}}}}}  % unregulated HRT surface

\newcommand{\fdc}{\mathfrak{j}^+}  % set of future-dir causal covectors
\newcommand{\fdt}{\mathfrak{i}^+}  % set of future-dir timelike covectors
\newcommand{\pcv}{\mathpzc{p}}  % V-thread
\newcommand{\qcv}{\mathpzc{q}}  % U-thread
\newcommand{\Pset}{\sets{P}}  % set of V-threads
\newcommand{\Qset}{\sets{Q}}  % set of U-threads
\newcommand{\Kset}{\sets{K}}  
\newcommand{\Lset}{\sets{L}}  
\newcommand{\hor}{\mathcal{H}}
\newcommand{\ew}{\mathcal{W}}
\newcommand{\Vzero}{V_{\equiv}}
\newcommand{\Vone}{V_{_{\!/\!/}}}
\newcommand{\Uzero}{U_{_{|\!|\!|}}}
\newcommand{\Uone}{U_{_{\!/\!/}}}
\newcommand{\muth}{\upmu}  % measure on V-threads (originally denoted as $\mu$)
\newcommand{\nuth}{\upnu}  % measure on U-threads (originally denoted as $\nu$)

\newcommand{\zV}{\zeta}  % slope of phi level set 
\newcommand{\zU}{\tilde{\zeta}}  % inverse slope of psi level set 
\newcommand{\QV}{Q}  % Q for V-flow
\newcommand{\QU}{\tilde{Q}}  % Q for U-flow
\newcommand{\aV}{\alpha}  % alpha for V-flow
\newcommand{\aU}{\tilde{\alpha}}  % alpha for U-flow
\newcommand{\zVm}{\zV_{_{{\text{max}}}}}   % optimized slope of phi level set 
\newcommand{\zUm}{\zU_{_{{\text{min}}}}}  %  optimized inverse slope of psi level set 

\newtheorem{theorem}{Theorem}[section]
\newtheorem{lemma}[theorem]{Lemma}

\title{Covariant bit threads}

\author[1]{Matthew Headrick}
\author[2]{and Veronika E. Hubeny}

\affiliation[1]{Martin Fisher School of Physics, Brandeis University, Waltham MA, USA}
\affiliation[2]{Center for Quantum Mathematics and Physics (QMAP)\\ 
Department of Physics \& Astronomy, University of California, Davis CA, USA}

\abstract{
We derive several new reformulations of the Hubeny-Rangamani-Takayanagi covariant holographic entanglement entropy formula. These include: (1) a minimax formula, which involves finding a maximal-area achronal surface on a timelike hypersurface homologous to $D(A)$ (the boundary causal domain of the region $A$ whose entropy we are calculating) and minimizing over the hypersurface; (2) a max V-flow formula, in which we maximize the flux through $D(A)$ of a divergenceless bulk 1-form $V$ subject to an upper bound on its norm that is non-local in time; and (3) a min U-flow formula, in which we minimize the flux over a bulk Cauchy slice of a divergenceless timelike 1-form $U$ subject to a lower bound on its norm that is non-local in space. The two flow formulas define convex programs and are related to each other by Lagrange duality. For each program, the optimal configurations dynamically find the HRT surface and the entanglement wedges of $A$ and its complement. The V-flow formula is the covariant version of the Freedman-Headrick bit thread reformulation of the Ryu-Takayanagi formula. We also introduce a measure-theoretic concept of a ``thread distribution'', and explain how Riemannian flows, V-flows, and U-flows can be expressed in terms of thread distributions.
}

\emailAdd{headrick@brandeis.edu} \emailAdd{veronika@physics.ucdavis.edu}

\preprint{BRX-TH-6708}

\begin{document}

\maketitle

\flushbottom

%---------------------------------------------------

\setcounter{section}{-1}

\section{Executive summary}
\label{ss:execsum}

In this paper we derive a number of new formulas that are equivalent to the HRT covariant holographic entanglement entropy formula. These formulas can be grouped into three classes: minimax, max V-flow, and min U-flow.

We fix a boundary spatial region $A$, and define $B:=A^c$ as its complement on a boundary Cauchy slice. For simplicity, we assume that the full system $AB$ is in a pure state. These spatial regions induce a decomposition of the conformal boundary into the four spacetime regions $D(A)$, $D(B)$, $J^\pm(\partial A)$, where $\partial A=\partial B$ is the entangling surface. We also define $\hat\I^\pm$ as the future/past boundary of the bulk spacetime (which may be a singularity or at infinite time), and $\hat\I^0$ as the end-of-the-world brane (if there is one). The HRT surface $\HRT$ divides the bulk into four spacetime regions: the entanglement wedge $\W(A)$, the complementary entanglement wedge $\W(B)$, and the future and past $J^\pm(\HRT)$; these bulk regions meet the conformal boundary at $D(A)$, $D(B)$, $J^\pm(\partial A)$ respectively \cite{Wall:2012uf,Headrick:2014cta}.\footnote{\, In the main text we are careful with the issue of the UV regulator, so that all quantities are finite and the maximizations and minimizations are meaningful. Specifically, we apply the entanglement wedge cross-section regulator \cite{Dutta:2019gen}. 
In this summary we will ignore this issue.}

\paragraph{Minimax:}
\be\label{minimax0}
S(A) = \frac1{4G_{\rm N}}
\inf_\ts\sup_{\surf\subset\ts}\area(\surf)\,.
\ee
The infimum is over piecewise timelike or null hypersurfaces $\ts$, which we call time-sheets, that are homologous to $D(A)$ relative to $\hat\I^+\cup\hat\I^-\cup\hat\I^0\cup J^+(\partial A)\cup J^-(\partial A)$; in other words, there exists a bulk spacetime region interpolating between $\ts$ and a part of the boundary that includes all of $D(A)$ and none of 
$D(B)$. Note that this is a \emph{spacetime} homology condition, as opposed to the \emph{spatial} homology condition on a given Cauchy slice familiar from the RT, HRT, and maximin formulas. As a consequence of this homology condition, $\ts$ necessarily contains the entangling surface $\partial A$. The supremum in \eqref{minimax0} is over achronal codimension-2 surfaces $\surf$ contained in $\ts$ and containing the entangling surface $\partial A$ (thereby excluding $\surf$ from the bulk chronal future and past of $\partial A$).

The minimax surface is the HRT surface (or any of them, if there is more than one). The minimizing time-sheet, on which the minimax surface is maximal, is highly non-unique; examples include the entanglement horizon (boundary of the entanglement wedge) of $A$ and that of $B$.

\paragraph{Max V-flow:}
\be\label{Vflow0}
S(A) = \frac1{4G_{\rm N}}
\sup_V\int_{D(A)}*V\,.
\ee
$V$ is a 1-form in the bulk, and the objective is its flux through $D(A)$. It is subject to a divergencelessness condition ($d*V=0$), a no-flux condition on $\hat\I^+\cup\hat\I^-\cup\hat\I^0$, and a norm bound. The norm bound can be expressed in two equivalent ways; the first is non-local while the second is local but involves an auxiliary scalar field:
\begin{enumerate}
\item $V=0$ in the \emph{bulk} chronal future and past of $\partial A$, and, for every bulk timelike curve, $\int dt\,|V_\perp|\le1$, where $t$ is the proper time along the curve and $V_\perp$ is the projection of $V$ perpendicular to the curve. Equivalently, the flux of $V$ through any codimension-1 timelike ribbon of spatial area $a$ is bounded above by $a$.
\item There exists a function $\phi$ in the bulk that equals $\pm1/2$ on $\hat\I^\pm\cup J^\pm(\partial A)$, such that the 1-forms $d\phi\pm V$ are everywhere future-directed causal.
\end{enumerate}
We can equivalently trade the 1-form $V$ for a set of ``V-threads'', bulk curves connecting $D(A)$ and $D(B)$. The first norm bound above would be interpreted as the statement that the total number of threads crossing a window of area $a$ being carried by an observer, over the observer's lifetime, is bounded above by $a$.

The V-threads are the covariant version of the bit threads introduced in \cite{Freedman:2016zud}. They may be spread out in time, but when collimated onto a single Cauchy slice, they reduce to Riemannian threads. A crucial point is that, in covariantizing the threads, they remain 1-dimensional, rather than becoming extended into world-sheets, and their endpoints remain spacetime points in $D(A)$ and $D(B)$, rather than becoming world-lines.

The maximal V-flows (or V-thread configurations) are highly non-unique but are restricted to the entanglement wedges $\W(A)$ and $\W(B)$, squeezing from the former to the latter via the HRT surface. This is schematically illustrated in the left panel of figure \ref{fig:maxVUflowES}.

One important difference between the V-flows and the Riemannian flows (or bit threads) is the following. For Riemannian flows, the choice of boundary region $A$ entered only in the objective (which is the flux through $A$), but not in the definition of a flow. Here, however, the region $A$ enters in the norm bound. That being said, a uniform definition of a V-flow can be given for multiple boundary regions $A,B,\ldots$, provided they lie on a common boundary Cauchy slice, by replacing the entangling surface $\partial A$ entering in the norm bound (in either version) by the union of all of the entangling surfaces $\partial A\cup\partial B\cup\cdots$ (which is equivalent to imposing the norm bound for all the regions simultaneously). This uniform definition of a V-flow can be used, for example, to prove subadditivity or to compute mutual informations. On the other hand, for boundary regions \emph{not} lying on a common Cauchy slice, the corresponding V-flow definitions differ essentially, and no uniform definition can be given.

\begin{figure}[tbp]
\centering
\includegraphics[width=0.8\textwidth]{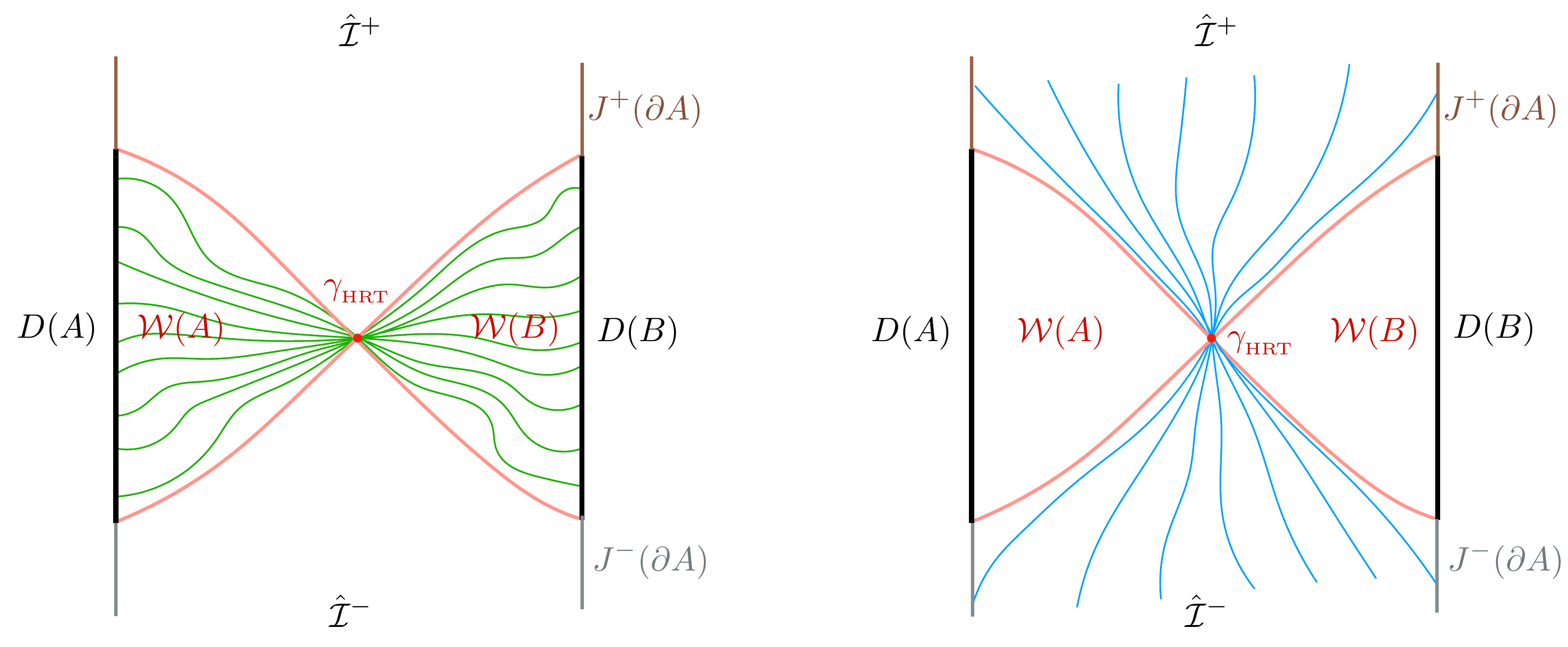}
\caption{\label{fig:maxVUflowES}
Cross section of generic maximal V-flow (left) and minimal U-flow (right). 
\figL: The V-flow (indicated by green curves) goes from $D(A)$ to $D(B)$, staying within their entanglement wedges $\W(A)$ and $\W(B)$, squeezing through the HRT surface $\HRT$ (red dot), and avoiding its past and future. 
\figR: The U-flow (indicated by blue curves) passes from $\hat\I^-\cup J^-(\partial A)$ to $\hat I^+\cup J^+(\partial A)$, squeezing through the HRT surface, and avoiding the entanglement wedges. (The future/past boundaries $\hat\I^\pm$ may be at infinity, as shown in the illustration, and/or singularities.)
}
\end{figure}

\paragraph{Min U-flow:}
\be\label{Uflow0}
S(A) = \frac1{4G_{\rm N}}
\inf_U\int_{\bkslice}*U\,.
\ee
$U$ is a future-directed causal 1-form in the bulk, and the objective is its flux through any bulk Cauchy slice $\bkslice$ containing $\partial A$. It is subject to a divergencelessness condition ($d*U=0$), a no-flux condition on $D(A)\cup D(B)\cup\hat\I^0$, and a norm bound. The norm bound, which is a \emph{lower} bound on the norm, in contrast to the upper bound constraining the V-flow, can be expressed in two equivalent ways; the first is non-local while the second is local but involves an auxiliary scalar field:
\begin{enumerate}
\item For every spacelike bulk curve connecting $D(A)$ to $D(B)$, $\int ds\,|U_\perp|\ge1$, where $s$ is the proper length along the curve and $U_\perp$ is the projection of $U$ perpendicular to the curve.
\item There exists a function $\psi$ in the bulk that equals $-1/2$ on $D(A)$ and $+1/2$ on $D(B)$, such that the 1-forms $U\pm d\psi$ are everywhere future-directed causal.
\end{enumerate}
We can equivalently trade the 1-form $U$ for a set of timelike ``U-threads'' beginning on the boundary region
$\hat\I^-\cup J^-(\partial A)$ and ending on $\hat\I^+\cup J^+(\partial A)$.

The minimal U-flows (or U-thread configurations) are highly non-unique, but are restricted to the bulk regions $J^-(\HRT)$ and $J^+(\HRT)$, squeezing from the former to the latter via the HRT surface and avoiding the entanglement wedges $\W(A)$ and $\W(B)$. This is schematically illustrated in the right panel of figure \ref{fig:maxVUflowES}.

\bigskip

Possible applications of these new formulations of the HRT formula are outside the scope of this paper, but since different ways of writing a given quantity are often useful for different purposes, it is generally advantageous to have as many such ways as possible. As one example, the formulas \eqref{Vflow0}, \eqref{Uflow0} define convex programs (which are actually Lagrange duals of each other), which may make them particularly amenable to numerical computation. Some of the new formulations may also be useful for proving general properties of holographic entanglement entropies, such as inequalities they obey. The new formulations may also have conceptual implications for our understanding of the relationship between geometry and entanglement in quantum gravity.

%---------------------------------------------------

\section{Introduction}

In this paper we develop a set of new, fully covariant geometrical prescriptions for holographic entanglement entropy (EE). These formulas are equivalent --- but not obviously so --- to the HRT and maximin formulas \cite{Hubeny:2007xt,Wall:2012uf}. We begin by providing, in this section, a self-contained introduction and summary of our results. We motivate the need for such reformulations in subsection \ref{ss:intro:motiv} and review the previously obtained prescriptions in \ref{ss:intro:review}. The reader familiar with the HRT, maximin, and Riemannian bit thread prescriptions is invited to skip to subsection \ref{ss:intro:expect}, which attempts to motivate intuitively what we might expect in covariantizing the bit threads and what further physical insight might be gained.  Subsection \ref{ss:intro:results} then describes the key aspects of the actual results, but still focusing on the conceptual rather than the technical side. In \ref{ss:intro:outline} we give a detailed outline of the rest of the paper.

\subsection{Motivation}
\label{ss:intro:motiv}

Holographic EE offers intriguing insights into the bulk geometry of holographic dualities.  Indeed, many now suspect that entanglement crucially  underlies the emergence of spacetime, stimulating the investigation of entanglement structure in holography.  An early hint at an interesting relation between spacetime geometry and entanglement came with the Ryu-Takayanagi (RT) prescription \cite{Ryu:2006bv,Ryu:2006ef}, promptly uplifted to a fully covariant formulation in general time-dependent context by Hubeny-Rangamani-Takayanagi (HRT) \cite{Hubeny:2007xt}: The EE $S(A)$ of a boundary region $A$ is given by the proper area of a smallest-area bulk codimension-2 extremal surface $\HRT(A)$ homologous to $A$.\footnote{\, 
	The homology condition can be rephrased as the existence of a \emph{homology region} whose boundary consists only of $A$ and  $\HRT(A)$ (the two meeting on the boundary at the entangling surface $\partial A$).
}
The fact that the HRT prescription relates a simple geometric construct, the extremal surface, to the EE is a priori highly non-trivial, since both the EE and the holographic mapping are individually rather intricate and complex.  

Although we now have compelling evidence for the HRT conjecture (for a review, see for example \cite{Rangamani:2016dms}), the prescription nevertheless retains mysterious features.\footnote{\, 
	Since these puzzles are compounded rather than dispelled in the quantum version of the holographic EE prescription \cite{Faulkner:2013ana,Engelhardt:2014gca}, in this paper we will restrict to the classical regime of $N\to \infty$, $\lambda \to \infty$.
}
On one hand, mapping a sharply-delimited CFT region to a sharply-defined bulk object, the extremal surface, whose location is likewise determined by entangling surface along with the bulk geometry) naively seems to be in tension with the usual holographic UV/IR correspondence: one might have expected the bulk construct to be more delocalized.
On the other hand, the HRT relation also has a peculiar global (as well as a topological) aspect: amongst all the extremal surfaces anchored on the entangling surface which are homologous to the entangled region, we are instructed to pick the one with least area. This allows for the extremal surface to jump to a different locus in the bulk under smooth deformations of the entangling surface or the CFT state, while at the phase transition itself we have a multiplicity of distinct but apparently admissible surfaces. This discontinuity is particularly perplexing in light of the bold conjecture 
\cite{Headrick:2014cta,Wall:2012uf}\footnote{\, 
Cf.\ \cite{Jafferis:2015del,Dong:2016eik,Faulkner:2017vdd,Cotler:2017erl,Chen:2019gbt} for recent evidence.
}
that the {\it entanglement wedge}\footnote{\, 
The entanglement wedge is defined as the bulk spacetime region spacelike-separated from the extremal surface $\HRT(A)$ and connected to the boundary region $A$ in question, or equivalently the bulk domain of dependence of the homology region.  (As usual, in asymptotically AdS spacetime we assume the requisite boundary conditions for bulk evolution, so the domain of dependence extends temporally along the boundary.  We will correspondingly take the generalized notion of global hyperbolicity and Cauchy surface.) 
} is the spacetime region which is most naturally ``dual to the reduced density matrix" $\rho_{A}$.  It suggests that, near such phase transitions, a tiny deformation of the reduced density matrix could suddenly allow us to encode a huge additional  spacetime region in the bulk.\footnote{\,  One extreme version of this, involving a large number of intervals in 3-d bulk, would change the entanglement wedge from covering `almost all' of the compactified Poincare disk to `almost none' of it. This feature readily generalizes to higher dimensions as well.
}

In light of these curious features, one is compelled to reexamine the meaning of holographic EE, or at an even more basic level, entanglement as such.  We will not attack this question directly here.  Instead, we want to obtain a more convenient characterization of holographic EE in terms of a distinct bulk construct that, while equivalent to HRT, would offer more suggestive hints as to its nature.  Indeed, one broadly expects that different formulations tend to demystify different features, so it is desirable to obtain as many distinct prescriptions for holographic EE as possible.  

Of course, in order for a given prescription to be even physically meaningful, it must be fully covariant:  it cannot rely on any choice of coordinates, foliation, or other extra baggage that is not part of the physics.
This requirement then automatically enables us to apply the prescription to general time-dependent settings.  Such explorations are interesting and useful in many contexts and indeed are presently being pursued with increasing vigor. More importantly, a hitherto underutilized aspect of covariant formulations is that they can naturally inspire conceptual advances, since a likely crucial but still mostly missing piece regarding the underpinnings of the holographic dictionary is the temporal aspect of the mapping.

\subsection{Previous prescriptions}
\label{ss:intro:review}

Let us briefly review previous reformulations of HRT.  First, as already explained in \cite{Hubeny:2007xt}, the codimension-2 extremal surface can be thought of as a surface with vanishing null expansions,\footnote{\, 
	Correspondingly, it admits 4 lightsheets \cite{Bousso:1999xy}, generated by the future/past directed in/out-going null normal congruences.  Equivalently, it is a surface of vanishing trace of the extrinsic curvature, and therefore vanishing expansion along any normal congruence.
}  which turns out to be a convenient characterization for using the Raychaudhuri equation to prove certain properties of the extremal surface (such as its consistency with CFT causality \cite{Headrick:2014cta}) under the usual physical assumptions, in particular the null energy condition (NEC).
In this situation, one can show that in the static context (or more generally on a surface of time-reflection symmetry), HRT reduces to the original RT prescription, involving a globally minimal surface on a preferred spatial slice.  Once restricted to Riemannian geometry, global minimality can be utilized to prove important properties of the holographic EE such as strong subadditivity (SSA) in an amazingly straightforward fashion \cite{Headrick:2007km,Headrick:2013zda}.  Unfortunately this convenience is not retained by the full Lorentzian context, which makes the corresponding properties rather more difficult to prove.  

To remedy this, Wall \cite{Wall:2012uf} reformulated the holographic EE via a {\it maximin} prescription, which entails taking an arbitrary bulk Cauchy slice passing through the entangling surface, finding the globally minimal area surface on it, and then maximizing this area over all possible slices; a codimension-2 bulk surface which realizes this maximin procedure is called a maximin surface, and any supporting slice on which it is globally minimal is called a maximin slice. Wall showed, under reasonable assumptions, first that the maximin surface coincides with the HRT surface, and second that its area obeys SSA in the general time-dependent context.
Moreover, the maximin construction specifies the homology constraint more naturally than HRT,\footnote{\, 
	As pointed out in \cite{Hubeny:2013gta}, to maintain causality, the HRT surface must remain spacelike-separated from the boundary region, which is ensured by requiring the homology region is achronal.  We call this the spacelike homology constraint, and in the maximin construction it is implemented automatically by extremizing only over minimal surfaces which lie on Cauchy slices containing the entangling surface.
} but its two-step minimization and maximization veils the nature of entanglement quantity even further, and is perhaps  conceptually less appealing (since it contains a vast amount of intermediate extra baggage and involves the breaking of a natural symmetry between the spatial and temporal directions).

So far, all reformulations involved a bulk codimension-2 surface, which leaves the conceptual meaning of holographic EE obscure and suffers from the associated puzzles such as a possible discontinuity in the location of the bulk surface; see \cite{Freedman:2016zud} for further discussion. However, in the static context, these puzzles were circumvented by a completely different prescription, which utilizes a construct dubbed \emph{bit thread}, a 1-dimensional object which can be thought of as a field line of a flow connecting the boundary region $A$ to its complement. In particular, the RT reformulation put forward by  Headrick-Freedman \cite{Freedman:2016zud} used the Riemannian max-flow min-cut (MFMC) theorem \cite{Federer74,MR700642,MR1088184,MR2685608} to express EE of a given region $A$ in terms of flows. The equivalence with RT was explained in greater detail in \cite{Headrick:2017ucz}, which further develops the tools we will use in the present work.   

The setting of \cite{Freedman:2016zud} is as for RT, namely the Riemannian geometry of a spatial slice $\bkslice$ which is a surface of time reflection symmetry  in the full Lorentzian geometry.  Define a {\it flow} to be any divergenceless vector field $v$ with unit-bounded norm: 
\begin{equation}
\nabla \cdot{} v = 0 \,, \qquad |v| \le 1 \ .
\label{eq:FHflowcond}
\end{equation}	
We can equivalently think of this vector field in terms of oriented flow lines (hence motivating the term threads\footnote{\, 
	Although this terminology suggests a discrete structure, this is merely employed as a conceptual crutch; the bound is in units of Planck area (or more accurately $4 \ell_P^{d-1}$ in $d+1$-dimensional spacetime), so we are typically dealing with a macroscopic number of threads within any region of interest.
}) which cannot end in the bulk (due to the divergencelessness condition) and have bounded transverse density (due to the norm bound).  For any boundary region $A$, the MFMC theorem states that the maximum flux of such  a flow from $A$ equals the minimal area achievable by any surface (or ``cut'') $m$ homologous to $A$:
\begin{equation}
\max_v \int_A v = \min_{\gamma\sim A} \, {\rm area}(\gamma) \ .
\label{eq:FHmfmc}
\end{equation}	
Intuitively, any flow from $A$ is clearly bounded by the minimal-area bottleneck $\gamma_{\rm min}$ the flow has to pass through, and the main content is that a maximizing (or optimal) flow achieves this bound. The EE $S(A)$ is then given by the flux of any such optimal flow (or maximal number of threads) from $A$:
\begin{equation}
S(A)= \max_v \int_A v \ .
\label{eq:FHS}
\end{equation}	
Note that although an optimal flow $v$ is far from unique, the bottleneck $\gamma_{\rm min}$ generically is unique, and corresponds to the RT minimal surface.  At this locus, the flow $v$ saturates the norm bound and is normal to $\gamma_{\rm min}$.  The homology constraint is implemented automatically, with the flow lines generating the requisite homology region between the boundary region and the bottleneck.
Moreover,  the number of flow lines has the familiar UV divergence coming from the divergent  area of $m$, and in fact the flow picture enables us to compare these divergent quantities more easily.  For pure states, we immediately see that $S(A) = S(A^c)$, implemented by the same flow configuration (with flipped directionality).  

Despite the flow prescription \eqref{eq:FHS} being equivalent to the RT prescription, it has a number of technical and conceptual advantages.  For example, showing certain properties such as subadditivity and SSA \cite{Freedman:2016zud} is even more immediate than for RT.
In fact, the utility of the difference in respective proof methods goes well beyond the mere confirmation of a previously-established result.  For example, it elucidated the difference between the universally-true SSA property and the holographically-true monogamy of mutual information (MMI) property \cite{Hubeny:2018bri,Cui:2018dyq}, whereas the  surface-based method proves these two properties equivalently.  In the cooperative flow construction of \cite{Hubeny:2018bri} this distinction was interpreted as the MMI being more intrinsically tied to bulk locality than SSA.

Conceptually, the bit thread picture is evocative of a bipartite nature of the entanglement structure.  One might think of each flow line as joining an EPR pair which straddles the entangling surface $\partial A$.  However, it is important to note that the flows depend not just on the state itself, but also on the entangling surface.  In other words, changing the region of interest generically changes the flow, unless the regions are nested (in which case one can find flows that simultaneously maximize both).  

Given the utility of bit threads, the obvious goal is to generalize them to the Lorentzian context, which allows for time dependence.  
The equivalence \cite{Freedman:2016zud,Headrick:2017ucz} between RT and bit thread formulations naturally suggests applying the same techniques (convex relaxation and Lagrange duality) to HRT.  While that is indeed the route we will take in this paper,  before embarking,  it will be instructive to pause to see what we might naively expect, and its pitfalls.

\subsection{Naive expectation}
\label{ss:intro:expect}

\begin{figure}[tbp]
\centering
\includegraphics[width=0.85\textwidth]{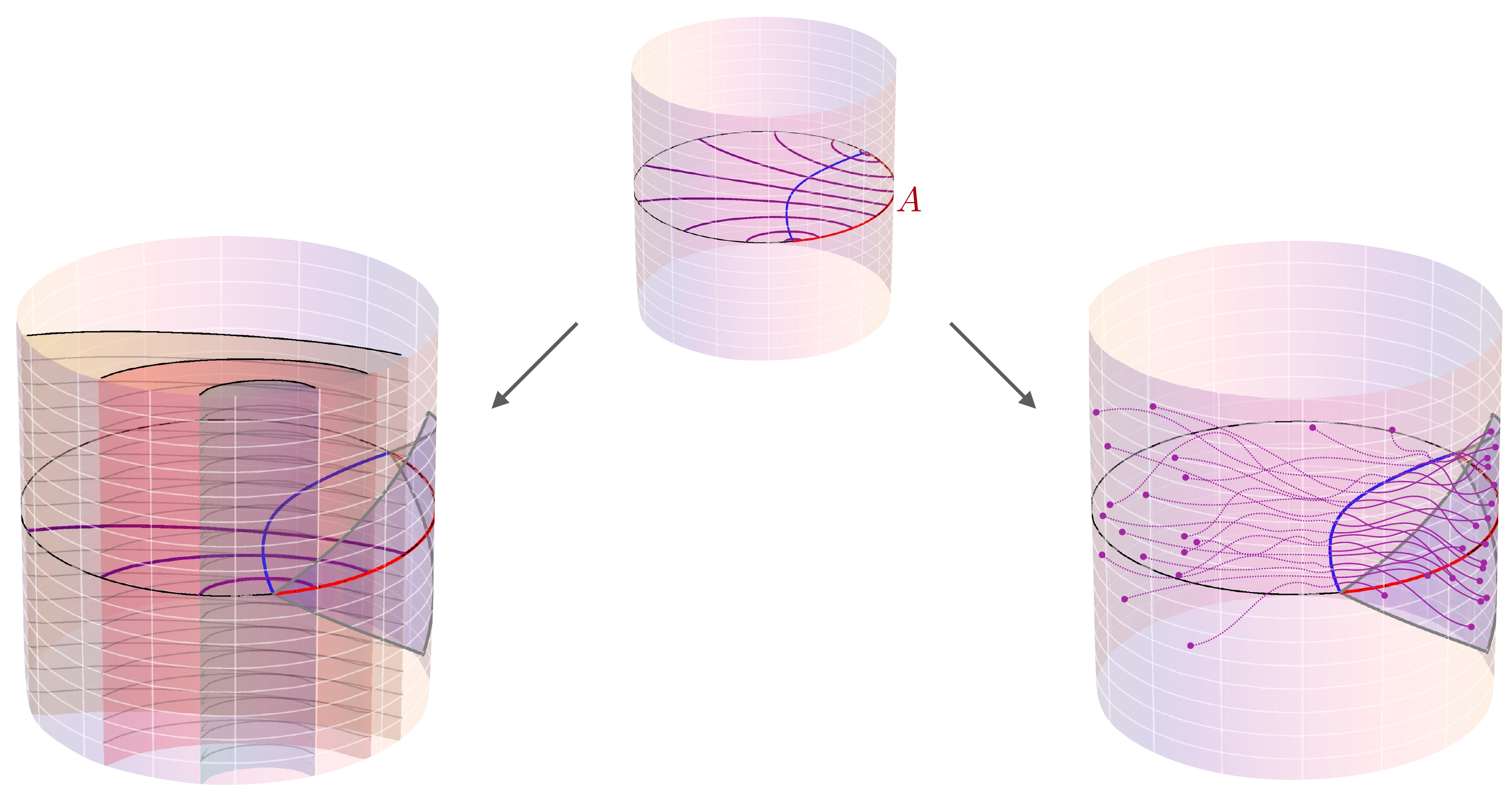}
\caption{\label{fig:expectations}
Schematic sketch of two natural covariantizations of Riemannian bit threads \figM, indicated by purple curves on constant time slice of Lorentzian global AdS spacetime.  
\figL: Threads extend in time to form 1+1 dimensional worldsheets.  
\figR: Threads remain 1-d objects but are no longer localized on a single time slice. The blue curve in each figure is the HRT surface. (For ease of illustration, the number of covariantized threads which we display varies between the three pictures, but of course in the actual constructions they would all agree.)
}
\end{figure}

If we view a Riemannian bit thread as a 1-dimensional string-like object that somehow embodies entanglement between the boundary regions at the string's endpoints, it then seems most natural to expect that when uplifted to the Lorentzian context, this string will extend in time, spanning a $(1+1)$-dimensional worldsheet, as schematically illustrated in the left panel of figure \ref{fig:expectations}. If the original endpoints characterized a Bell pair, the intersection of the worldsheet with the boundary would now correspond to the worldlines of these individual entangled particles.  And since their entanglement cannot be created or destroyed acausally, we might furthermore presume such worldlines, and so plausibly the interpolating worldsheet, to be generically timelike (or at least causal).  Intersecting the worldsheet by constant-time slices would then recover the bit threads as ``snapshots'' at the given boundary time, effectively tracking the entanglement dynamics.

This worldsheet picture is further bolstered by the observation that in the Lorentzian context, the HRT prescription entails a spacetime codimension-2 surface (uplifted from a spatial codimension-1 surface on a fixed-time slice), so that the uplift of the threads, being in some sense the dual objects, should be correspondingly bulk 2-dimensional sheets, locally extending in the orthogonal directions from the HRT surface. We might then try to recast the EE $S(A)$ as the maximal number of such worldsheets which can pass through the boundary domain of dependence $D(A)$.

Interestingly, this naive expectation does \emph{not} seem to be realized. One immediate challenge has to do with a suitable generalization of the norm bound in \eqref{eq:FHflowcond}.  Since in the Riemannian context this condition restricted the threads from getting too close together, we might expect that, similarly, the worldsheets should have suitably bounded density.  It is also natural that the divergencelessness condition in \eqref{eq:FHflowcond} translates to the restriction that the worldsheets cannot simply terminate in the bulk.  But these two conditions together appear to impose too global a constraint on the worldsheets, which in particular can violate causality.  For example, if the bulk spacetime happens to contract in the future, this would seem to teleologically expel the worldsheets in the present.\footnote{\, 
	We thank Juan Maldacena for originally raising this possibility.
}  

Moreover, if each thread uplifts to a physical  worldsheet, one might expect that this object is intrinsically 2-dimensional, analogously to a fundamental string worldsheet, with no residual information about its foliation by the `snapshot' threads. In other words, observers with relative boosts should experience the same  worldsheet, but naturally associate different thread-foliations to it.  This however makes the optimal thread congruence overconstrained, since temporally deforming the entangling surface deforms the position of the corresponding HRT surface, where the bit threads are required to be maximally packed and directed perpendicularly.\footnote{\, 
As a simple example, boosting the boundary region in opposite directions of a translational symmetry yields two distinct HRT surfaces that intersect in the middle but with different normals.}
Said differently, the collection of HRT surfaces anchored somewhere on the boundary `entangling tube' $\partial A(t)$, without preserving the $t$-foliation, spans a bulk codimension-0 region instead of a timelike codimension-1 hypersurface, analogously to the set of HRT surfaces anchored on a given boundary Cauchy slice generically spanning a codimension-0 bulk region instead of all lying on the same bulk Cauchy surface.

Another way of viewing the potential incompatibility of the threads forming worldsheets under time evolution is the following.  Considering for a moment the original Riemannian context, it was observed already by \cite{Freedman:2016zud} that a given thread configuration cannot optimize on two crossing regions simultaneously.  The intuitive reason is that the threads have to be maximally packed on the RT surface and traverse it perpendicularly, but crossing regions have intersecting RT surfaces, and therefore incompatible perpendicular directions at their intersection.  Coming back to the Lorentzian case, the physically pertinent geometrical object is not the entangling surface $\partial A$ as such, but the boundary domain of dependence $D(A)$.  It is then tempting to view the case of time-evolved regions $A(t_1)$ and $A(t_2)$ as likewise crossing, in the sense of their domains of dependence forming non-nested sets with non-empty intersection.  Even though the HRT surfaces as such do not intersect, one might nevertheless worry that it won't  be possible in general to find a single set of thread worldsheets whose $t_1$ and $t_2$ cross-sections would be simultaneously optimizing.  

The basic flaw with the naive expectation of bit thread worldsheets is to think of the EE as pertaining to a boundary space\emph{time}, in particular as admitting a canonical temporal extent.  On a static boundary  spacetime (for example, the Einstein static universe for an asymptotically globally AdS spacetimes), there is indeed a natural way of fixing a given boundary region in space and considering the evolution of its EE in time, in response to the evolution of the state.  However, in general, such a setting would be too limiting.  For example, there is no canonical way of spatially delimiting a region in time-evolving asymptotically locally AdS spacetimes, or for non-static observers.  EE pertains to a given boundary state and region at a single instant in time, and in these more general contexts there is no preferred extension from one time to another.  This is the reason that a codimension-2 surface does not dualize to a 2-dimensional object: we lose one dimension because of this instantaneous nature of EE.  

The conclusion (already emphasized by \cite{Freedman:2016zud}) is that the threads are not to be viewed as physical objects, and their endpoints are not naturally extended in time.  What could then be an interpretation of such a scenario?  A holographic CFT is necessarily strongly coupled, so dynamically a Bell pair should quickly decohere and the shared entanglement spread out.  Accordingly, we would expect entanglement to be generically delocalized, but nevertheless one might imagine that one could momentarily localize it, perhaps by an entanglement distillation process. Since we can think of this process as a spacetime event (with negligible time duration), we might associate the thread ends to such events.  The number of threads which can end in $D(A)$ would then be counting the number of distillation processes we could perform on the state specified at $A$, which would in turn characterize its entanglement.\footnote{\, 
	Although for general mixed states, the entanglement of distillation does not coincide with EE, it does so for pure states.  In the present context, we think of the full geometry as encoding a pure state, so one might hope that its entanglement could be viewed in this way.
} 

How does this bear on the relevant geometrical constructs in the bulk which should be associated with EE?
The above speculation suggests that on the boundary we retain string endpoints, and for the Lorentzian formulation to correctly reduce to the Riemannian one for the static cases, the most natural mathematical construct is then still a thread.\footnote{\, 
		In principle, due to the diverging conformal factor at the AdS boundary, another possible construct connecting two boundary points is an extended flux-tube-type object; but this does not immediately reduce to a Riemannian thread.  In fact, it will turn out that threads already allow requisite delocalization, while retaining relatively simple description.
	}  
But since we have a Lorentzian spacetime (with no preferred time slice in general), we would expect that these threads can meander in both space and time, as schematically illustrated in the right panel of figure \ref{fig:expectations}. The upgraded expectation is that the EE is captured by the maximal number of threads adhering to certain restrictions.  Causality requires that they end within $D(A)$. But what further constraints should we impose on them?  Can they be timelike somewhere? What bulk regions are they allowed to penetrate? How do they interface with each other?

\subsection{Preview of results}
\label{ss:intro:results}

To answer these questions we will dualize HRT. Our derivation effectively entails a double convex relaxation (in both spatial and temporal directions), which combines features of the previously studied MFMC theorem in the Riemannian setting, as well as the min flow-max cut theorem in the Lorentzian setting \cite{Headrick:2017ucz}. This allows us to construct a web of geometrically distinct prescriptions, or reformulations, in terms of flows, or equivalently in terms of threads (which we will slightly generalize from the original picture of integral curves of a nowhere-vanishing flow field).  Altogether we will present ten new formulas for computing the holographic EE in a general time-dependent holographic spacetime (in addition to the already-known HRT and maximin).

To develop the mathematical framework, it will be instructive to start in a more general context and only restrict to holography at a later point, which will simultaneously enable us to identify the special features implemented by holography.  Hence at the outset, we will retain only the ``kinematical'' aspects of the holographic context, but not specialize to keeping the ``dynamical'' ones until section 6.

We will see that in optimization problems, given a function of two variables $f(x,y)$, its maximin $\sup_x \, \inf_y \, f(x,y)$ is generically distinct from its minimax  $\inf_y \, \sup_x \, f(x,y)$ where we simply switch the order of the two extremizations.  However, the minimax universally provides an upper bound for the maximin.\footnote{\, 
		The reader who seeks a more intuitively obvious mneumotic is invited to observe that the shortest giant is still taller than the tallest dwarf.}
The Lagrange duality of convex optimization problems utilizes this structure: weak duality gives the bound while strong duality gives a sufficient condition to saturate it.  The two optimizations are over the original variables and over the Lagrange multipliers that implement the constraints, respectively; hence the dual problem is phrased in terms of the Lagrange multipliers instead of the original variables.

One natural class of situations where the maximin and the minimax values coincide is when there exists a ``global saddle point'' $(x_0,y_0)$ such that $f(x_0,y)$ is $y$-minimized at $y_0$ while $f(x,y_0)$ is $x$-maximized at $x_0$. More broadly, the criterion for the minimax to equal the maximin is specified by the minimax theorem (originally developed in the context of game theory). In the continuum context, the crucial criterion is for the function to be  convex-concave in its respective arguments. Before devising such a function in our context, we first consider a more localized geometric prescription where all the action effectively takes place within a single hypersurface of a specified class, and then  optimize over such hypersurfaces.  The maximin prescription for holographic EE \cite{Wall:2012uf} is but one example; here we can think of the action as taking place within a single Cauchy slice, and within this slice we can use Riemannian max flow-min cut theorem to convert it to slice flow, so that upon maximizing over all slices we arrive at an alternative,  ``maximax'' prescription.  But instead of Cauchy slices, one could equally start with a different class of hypersurfaces.  Since, roughly-speaking, an HRT surface area increases under spatial deformations and decreases under temporal deformations,\footnote{\, 
	This is just a heuristic to build intuition; the separator is not precisely null, and in fact one can typically find spacelike Cauchy slices (which are not maximin slices) along which the HRT surface is not the minimal area surface.  This is possible whenever the expansion of the null normal congruence from the extremal surface becomes negative (which is generically the case).
}  one could first find the maximal-area surface within a timelike hypersurface (which we will dub ``time-sheet'') and then minimize this area over all time-sheets.  This is our ``minimax'' prescription.

To unify both maximin and minimax into a common phrasing, we can view the extremal co-dimension-2 surface in question as the intersection of two codimension-1 hypersurfaces, namely a spatial slice and a time-sheet.  The difference in the two prescriptions then boils down to merely the order of extremizations, and hence becomes a subject of the minimax theorem. But since the respective sets of hypersurfaces (and intersections thereof) are not convex sets, the two quantities thus identified (which we'll still refer to as minimax and maximin), need not coincide, and in fact it is easy to construct an example of spacetime where they differ.  We will see that allowing partial relaxation brings the quantities closer together, and the minimax theorem indicates that if we can embed the problem into a fully convex-relaxed one, then the corresponding convex maximin and convex minimax would indeed coincide. Quite remarkably, in the actual holographic context (characterized by the correct dynamics), we will see in section \ref{sec:solutions} that maximin and minimax in fact do coincide even without any convex relaxation, due to the existence of a global saddle point, as already heralded by  the HRT prescription.
 
However, in the more general context (retaining the kinematics but not the dynamics), to achieve the equivalence criterion of the minimax theorem, we need to convex-relax these hypersurfaces. As in \cite{Headrick:2017ucz}, we can view a hypersurface as a level set of a scalar field, and impose conditions on the scalar field so as to comprise a convex set.  The slice is then convex-relaxed (``smeared'') to a continuous collection of level sets, weighed by the gradient norm.  The novel feature here compared to the implementation in \cite{Headrick:2017ucz} is that we do this not with  just a single scalar field but with two scalar fields simultaneously; in particular, we'll associate a field $\phi$ with a temporal smearing of a spatial slice, and a field $\psi$ with spatially-smearing a time-sheet.

The objective function to be optimized, generalizing the area of a codimension-2 surface, is constructed from a certain scalar pairing between the respective gradient 1-forms, such that it has the requisite convexity properties. The minimax theorem then states that we can exchange the order of the extremizations without changing the optimal value. The formulation in terms of these two scalar fields therefore provides a convenient starting point for subsequent reformulations.  However, while mathematically central to our story and appealingly treating the spatial and temporal directions on a similar footing, it is not the most intuitively suggestive formulation. Instead, it turns out to be conceptually more convenient to Lagrange-dualize on one or both scalar fields, recasting the formulation in terms of flows.  We will dub the primarily-spatial flows (dual to $\psi$) V-flows and the temporal flows (dual to $\phi$) U-flows.  In the next few paragraphs, we will preview the respective prescriptions in greater detail.
 
\paragraph{V-flows:}  The convex program of minimizing over the scalar field $\psi$ (which smears out the time-sheets) at fixed $\phi$ can be dualized to obtain a concave program, which entails maximizing the flux of a divergenceless 1-form $V$ subject to a certain norm bound. This norm bound is analogous to the familiar one $|v|\le1$ in the Riemannian setting, but because of the Lorentzian geometry, it now has a richer structure, and in particular is non-local. Specifically, it is implemented by requiring the 1-forms $d\phi\pm V$ to be everywhere future-directed causal.

We can re-cast this norm bound in a way that actually does not invoke $\phi$ at all, instead posing a more global norm condition, namely, we upper-bound the integral over an arbitrary timelike curve $\qcv$, parameterized by proper time $t$, of the norm of the perpendicular projection of $V$:
\begin{equation}\label{eq:Vflownormbnd}
	\int_\qcv dt\,|V_\perp|\le 1 \ .
\end{equation}	
$S(A)$ is then calculated by maximizing the flux of a divergenceless flow $V$ from $D(A)$ subject to the requisite boundary conditions and the above norm bound.  Using time-sheets, we can interpret \eqref{eq:Vflownormbnd} physically as the requirement that any observer carrying a unit-area window cannot capture more than a unit amount of total flux of $V$ over their entire lifetime.  To recover the original bit thread formulation in the Riemannian context \cite{Freedman:2016zud}, we can specialize to the case where all the flux is localized on a single slice, in which case the integrated bound \eqref{eq:Vflownormbnd} collapses to the simple norm bound of \eqref{eq:FHflowcond}.\footnote{\, 
	The norm bound \eqref{eq:Vflownormbnd} also ratifies our previous naive expectation that there cannot exist a compatible flow that simultaneously maximizes on temporally-separated regions with crossing domains of dependence: even if the respective HRT surfaces do not intersect, they are necessarily timelike-separated somewhere, which allows us to construct a curve $\qcv$ for which this norm bound is violated.  
	}
One might think that in the more general Lorentzian case this is now infinitely more complicated since we have (continuously) infinitely many observer worldlines to check, but in fact the initial formulation in terms of $\phi$ can be viewed as providing a single ``certificate'' that guarantees the bound for every worldline.

Although the original scalar field $\psi$ was introduced as a tool to smear out a time-sheet, and correspondingly the dual $V$-flow 1-form has flow lines which emanate from $D(A)$ and go into $D(A^c)$, these flow lines need not actually remain spacelike everywhere; they can have timelike (or null) pieces along the way, subject to \eqref{eq:Vflownormbnd}.

\paragraph{U-flows:}
The U-flow case follows a very similar story as the one for V-flows, but now pertains to dualizing on $\phi$ instead, starting from the minimax formulation.  In particular, keeping $\psi$ fixed, we can dualize the concave program of maximizing over $\phi$ (i.e.\ smeared Cauchy slices) to the convex program of minimizing the flux of a divergenceless 1-form $U$ at future infinity, again subject to a norm bound, namely that the 1-forms $U\pm d\psi$ must be everywhere future-directed causal.\footnote{\, 
We can also obtain the U-flow by directly dualizing the V-flow program; see the diagram \eqref{diagram} for a summary of these relations.}

Analogously to the V-flow case, we can write the norm bound as a global condition in terms of an arbitrary spacelike curve $\pcv$ passing between the domain of dependence of the region $A$ and that of its complement:
\begin{equation}\label{eq:Uflownormbnd}
	\int_\pcv ds\,|U_\perp|\ge1 \,,
\end{equation}	
where $U_\perp$ is the projection of $U$ perpendicular to $\pcv$.
Note the direct analogy with \eqref{eq:Vflownormbnd}; however, whereas in that case the constraint $\qcv$ had no information about the given region but the objective function (the boundary region on which the flux is evaluated) did, here it is the other way around: the objective function is region-agnostic while the region determines the constraint via requisite set of curves $\pcv$.

The difference of viewpoint between the two sides provides us with different toolkits that complement each other.  For example, in proving holographic entropy inequalities, one can optimize the smaller side in V-flow language to bound the larger side from below, or conversely one can optimize the larger side in the U-flow language and show that this provides an upper bound for the lower side. It is worth noting however that although the two prescriptions are closely analogous (seemingly amounting mainly to spatio-temporal flip accompanied by minimization-maximization flip), there are still crucial differences due to the boundary conditions, as mentioned above. An avatar of this feature appears already in the maximin and minimax formulation using hypersurfaces, where specifying a Cauchy slice does not fix the homology class of codimension-2 surfaces, while specifying a time-sheet does. 

\paragraph{Thread distributions:}

Once we have formulated the convex-relaxed problem in terms of flows, it follows immediately that we can re-cast them also in terms of threads.  In particular, for a flow field $V$, which is a 1-form, the flow lines (which can be thought of either as integral curves of the dual vector field, or in terms of the Hodge dual $*V$), can be viewed as threads joining the given boundary subsystem to its complement.\footnote{\, 
	One might wonder why we didn't try to construct threads already from $d\psi$ which is after all also a 1-form.  This is generically not possible, since $d\psi$ is not divergence-free, so that it is not extendible into a full thread; an extreme case being when $\psi$ is just a simple step function corresponding to a localized time-sheet.  
}
However, an alternative, and perhaps a more appealing, notion of a ``thread" is simply an unoriented curve in the spacetime, which is allowed to intersect other threads.  While this generalizes the notion of  integral curves of a smooth vector field (which cannot intersect by construction), one can reformulate the requisite norm bounds in terms of a restriction on the measure on the space of curves.  

In the Riemannian case, this measure is restricted to ensure that the thread density is $\le1$ everywhere, and the corresponding linear program maximizes the total measure, subject to this constraint, of the set of threads which join the given region and its complement.
We can also dualize this program, which amounts to minimizing a positive function $\lambda$ subject to it integrating to a value $\ge 1$ over any thread. The minimum is attained by a step function along the minimal surface, which again recovers the EE.  
We can readily generalize this setup to multiple regions and for example obtain a thread version of max multiflow theorem of \cite{Cui:2018dyq}.  Indeed, we can map any (multi)flow to a thread distribution and vice-versa (though not via an isomorphism due to the non-orientation of treads and non-crossing of flow lines), so that every statement pertaining to threads has an avatar in flows.

In the full Lorentzian context, we can similarly translate V-flows and U-flows into V-threads and U-threads, respectively.  However, the since the two flows depend on each other, the corresponding constraints on the threads are now non-local: instead of bounding thread density at every spacetime point, we need to bound thread density integrated over each dual thread.  In this way, the density bounds for the V- and U-threads enforce each other.  The EE is obtained by maximizing the number of V-threads, or equivalently (in the dual picture) minimizing the number of U-threads, subject to these constraints.  The V-threads join the boundary domain of dependence of the given region with that of its complement, but are not restricted to remain spacelike in the bulk. The U-threads, on the other hand, are necessarily causal and join the past boundary of the spacetime with the future one.  Their knowledge of the given region comes through the norm bound, which amounts to the U-threads forming a sufficient barrier separating the requisite domains of dependence.  The summary of all of the prescriptions for the convex-relaxed value is in the diagram \eqref{diagram2}.

\paragraph{Optimized flow/thread configurations in holography:}

The restrictions \eqref{eq:Vflownormbnd} and \eqref{eq:Uflownormbnd}, 
on the V- and U-flows, or the corresponding rephrasing in terms of the threads, still allows these geometrical objects to permeate the full spacetime without singling out any more localized regions. However, for \emph{optimized} flows in the holographic context, this allowed set collapses to become more colimated, as sketched in figure \ref{fig:maxVUflow}. In particular, the optimized flows must all pass through the HRT surface!  The V-threads are confined to the entanglement wedge of $A$ and that of its complement, while the U-threads are confined to the future and past of the HRT surface.  These four regions naturally partition the bulk spacetime \cite{Headrick:2014cta}, and the V-threads intersect the U-threads only along the HRT surface, thereby naturally counting its area as the EE.   Notice that this description implements just the right amount of localization, without imposing any non-geometrical features, and retaining maximal democracy between spatial and temporal directions.  The localization comes about collectively: each thread by itself does not exhibit any special points along its length --- rather, its physical significance lies in what regions it connects and how it interfaces with the other threads.

\subsection{Outline of the paper}
\label{ss:intro:outline}

Having previewed our results, in the remainder of the paper we develop them in full technical detail. We start by explaining our setup and assumptions in section \ref{sec:background}.  Since our primary focus is on geometrical prescriptions for calculating EE in holography, we first clarify in subsection \ref{sec:EWCS} how we regulate EE to obtain a meaningful quantity. In particular, in most of our derivations we work in ``regulated spacetime'' consisting of the entanglement wedge of the union of a given boundary region $A$ and its disjoint complement $B$ (with the separating entangling surface slightly thickened so that we can use the entanglement wedge cross section to regulate the EE).  This is merely a matter of presentational convenience; as implied by CFT causality (and confirmed explicitly in section \ref{sec:embedding}), the details of the spacetime in the past or future of the HRT surface do not influence the holographic EE.

To appreciate the special features of holography, it will furthermore be instructive to work in a more general setting --- namely a globally hyperbolic spacetime whose conformal boundary includes timelike components --- which we specify in subsection \ref{sec:setup}.
This allows us to identify the key codimension-1 constructs, schematically illustrated in figure \ref{fig:setup}, consisting of distinct components of the spacetime boundary, as well as slices and time-sheets, which will pave the way for defining  the maximin and minimax constructs (in section \ref{sec:relaxation}).  
In order to develop the framework to recast these prescriptions in terms of flows (in section \ref{sec:flows}), it will be convenient to use covectors and their Hodge duals instead of the more familiar vector fields; subsection \ref{sec:1forms} reviews this formalism and constructs a scalar concave-convex ``wedgedot" pairing of two covectors that will play a crucial role throughout the paper.

Section \ref{sec:relaxation} then proceeds to explain the essential points from minimax theory and convex relaxation.  To anchor the reader, we first specify the geometric maximin and minimax constructs and variations thereon in subsection \ref{sec:maximin}.  To understand the relation between them,  we step back to review minimax theory (in its original game theory context) in subsection \ref{sec:minimaxtheory}.  In subsection \ref{sec:convexrelaxation} we apply the theory to our geometrical context and perform convex relaxation to define a new quantity. We use the suggestive notation $S_-$ for maximin, $S_+$ for minimax, and $S_c$ for the convex-relaxed quantity.\footnote{\, 
	The choice of the letter $S$ and acronym ``EE'' in our summary of results were in anticipation of the holographic context; however strictly speaking in the broader context our results apply to, there is no dual theory in which to formulate an EE. The preceding paragraphs pertaining to V-flows, U-flows, and threads all give prescriptions for $S_c$, whereas the hypersurface-localized prescriptions give $S_-$ in slice-localized and $S_+$ in time-sheet-localized contexts.
}
In general, the minimax theorem only ensures that $S_- \le S_c \le S_+$, while in the more physically relevant context of holography (discussed in section \ref{sec:solutions}), all three quantities in fact do coincide, $S_- = S_c = S_+$, and provide alternate prescriptions for the holographic EE.  

Having formulated $S_c$ in terms of a convex-concave pairing of the scalar fields $\phi$ and $\psi$ (whose level sets implement convex-relaxation of slices and time-sheets), we finally get to the core of the paper in section \ref{sec:flows}, which reformulates these in terms of flows by applying Lagrange duality.  
Subsection \ref{sec:Vflows} dualizes on $\psi$ at fixed $\phi$ to obtain the V-flow program, while subsection \ref{sec:Uflows} dualizes on $\phi$ at fixed $\psi$ to obtain the U-flow program.  
To avoid fragmenting the narrative overmuch, in both subsections we relegate the actual dualizations, as well as proofs of lemmas needed for the reformulations of the norm bounds etc., to subsection \ref{sec:flowproofs} which serves as a mini-appendix to section \ref{sec:flows}. In order to hone intuition for the effect of convex relaxation, in both subsections we devote a subsubsection to an explicit toy example of a spacetime wherein $S_- \ne S_+$ (introduced in subsection \ref{sec:convexrelaxation}).  Instead of a full convex relaxation, we perform a particularly simple partial relaxation, to demonstrate how it diminishes the gap between maximin and minimax values, bringing them closer to $S_c$; we refine this with a further (but still partial) convex relaxation in appendix \ref{app:piecewiselin}, which generalizes the previous calculations, now presenting the V-flow and U-flow cases in parallel in a self-contained manner. In subsection \ref{sec:subadditivity} we prove subadditivity of $S_c$ in these two formulations, which illustrates that despite the close parallel between the V-flow and U-flow programs, there is a non-trivial difference in how they implement various features of $S_c$. Finally, the heart of the mathematical framework resides in subsection \ref{sec:flowproofs} where we prove the various statements asserted earlier.  We start by presenting six covector-pair lemmas, followed by five explicit dualizations, and culminating in proving the equivalence of the norm bounds appearing in the respective flow programs which involves a beautiful generalization of Hamilton-Jacobi theory for non-differentiable Lagrangians, this curious feature arising due to signature-dependence in the Lorentzian context.

Section \ref{sec:threads} substantiates (and surpasses) the title and original motivation of the paper; in addition to covariantizing the Riemannian bit thread formulation of \cite{Freedman:2016zud} captured by V-threads,  it also provides an alternate covariant prescription in terms of U-threads, mimicking the V-flow and U-flow formulations of section \ref{sec:flows}.  To generalize the notion of threads viewed as flow lines, subsection \ref{sec:RMFMC} develops the framework of thread distributions in the Riemannian context from scratch, applying the technology of convex optimization to measures on sets of curves and proving the analog of the MFMC theorem.  In subsection \ref{sec:VUthreads} we return to the Lorentzian context, and formulate the V-thread and U-thread prescriptions for computing the holographic EE.  

In section \ref{sec:solutions}, we finally specialize to the holographic context.  We first explain in subsection \ref{sec:HRTsurface} why the non-convex maximin and minimax values coincide with the convex-relaxed one by identifying the HRT surface as a global saddle point.  This relies on certain physical assumptions about the bulk spacetime and its boundary, and demonstrates the equivalence between maximin and HRT prescriptions already shown in \cite{Wall:2012uf}, as well as that between minimax and HRT.   In subsection \ref{sec:maxVflows} we consider the optimized flows and show that they pass through the HRT surface.  While hitherto most of our constructions pertained to a single region $A$ on the boundary, in subsection \ref{sec:multiple} we indicate how to generalize the discussion to multiple regions, with a suitably enlarged regulated spacetime.

In section \ref{sec:embedding}, we show how we can enlarge our constructions pertaining to the regulated spacetime, by embedding the latter in the full spacetime, without changing any of the results. We also explain what happens to the constructions in the limit that the regulator is removed.

We conclude in section \ref{sec:discussion} with a discussion and possible applications and future directions. 

As the paper is fairly notationally heavy, we also include for the reader's convenience in appendix \ref{sec:notation} a table of notation.

We will assume that the reader is familiar with the content of the paper \cite{Headrick:2017ucz}, in particular the concepts of convex relaxation and Lagrange duality and their application to geometrical theorems such as the Riemannian max flow-min cut theorem.

%---------------------------------------------------

\section{Background}
\label{sec:background}

In this section, we explain the basic setup, assumptions, and notation we will use in the rest of the paper. We start in subsection \ref{sec:EWCS} by explaining how we deal with ultraviolet divergences in holographic EEs. In subsection \ref{sec:setup}, we then detail the assumptions and notation we use for the spacetime we'll be working in. Finally, in subsection \ref{sec:1forms} we explain important notation for 1-forms that we will use throughout
the paper. We remind the reader that appendix \ref{sec:notation} contains a table summarizing the notation we use in the paper.

\subsection{EWCS regulator}
\label{sec:EWCS}

The purpose of this paper is to give several formulas, equivalent to the original HRT formula \cite{Hubeny:2007xt}, for the EE $S(A)$ of a spatial region $A$ in a holographic field theory. This task only really makes sense if $S(A)$ is a finite quantity, in other words if any infrared and ultraviolet divergences have been regulated in some way. The choice of IR regulator will not have any bearing on our work, so we will not deal with it explicitly, but will simply assume that any IR divergences have been regulated somehow. On the other hand, among the various kinds of UV regulators that have been employed for holographic EEs, there is one that will be particularly natural for our purposes, involving the so-called entanglement wedge cross section (EWCS) \cite{Takayanagi:2017knl,Nguyen:2017yqw,Dutta:2019gen}, which we will now review.\footnote{\, For systematic discussions of UV regulators for holographic EE, see \cite{Sorce:2019zce,Grado-White:2020wlb,paper1}.}

To set the stage, it is useful to first recall the case where the entropy is naturally UV-finite and does not need to be regulated. This happens when the conformal boundary has multiple connected components --- i.e.\ the spacetime is a multiboundary wormhole --- and we are computing the entropy of an entire connected component or a set of them. To fix some notation, we denote
the bulk by $\M$ and its conformal boundary by $\N$. $\M$ also has past and future boundaries $\I^-$, $\I^+$, which may be at infinite time and/or include singularities. Let $\{\N_i\}$ be the connected components of $\N$, and choose a Cauchy slice $\bdyslice_i$ for $\N_i$. Each $\bdyslice_i$ is boundaryless. So if we let $A$ be the union of a subset of $\{\bdyslice_i\}$, then the entangling surface $\partial A$ is empty and there is no UV divergence in $S(A)$.

On the other hand, if $\partial A$ is non-empty, then $S(A)$ contains a UV-divergent piece. To regulate this divergence, we essentially want to turn the bulk spacetime into something akin to a multiboundary wormhole; this is what the entanglement wedge and EWCS will do for us. We choose a boundary Cauchy slice $\bdyslice$ containing $A$, and a region $B\subset \bdyslice\setminus A$ that almost fills the complement but leaves a small buffer between $A$ and $B$ \cite{Dutta:2019gen}. The EWCS $S(A:B)$ is defined as the area of the minimal extremal surface in the entanglement wedge of $AB$ that is homologous to $A$ relative to the joint HRT surface $\hHRT(AB)$. (The reason for the hat on $\surf$ will become clear shortly.) From a field-theory viewpoint, the buffer between $A$ and $B$ eliminates the  short-wavelength modes shared between $A$ and its complement that cause the divergent EE. We should also emphasize, of course, that the EWCS is an interesting quantity in its own right, quite apart from its use as a regulator, with various conjectured field-theory interpretations \cite{Nguyen:2017yqw,Takayanagi:2017knl,Dutta:2019gen}. Although we have motivated the EWCS as a regulator, the results of this paper apply equally well to any EWCS calculation; they do not depend on the buffer between $A$ and $B$ being small in any sense.

\begin{figure}[tbp]
\centering
\includegraphics[width=0.3\textwidth]{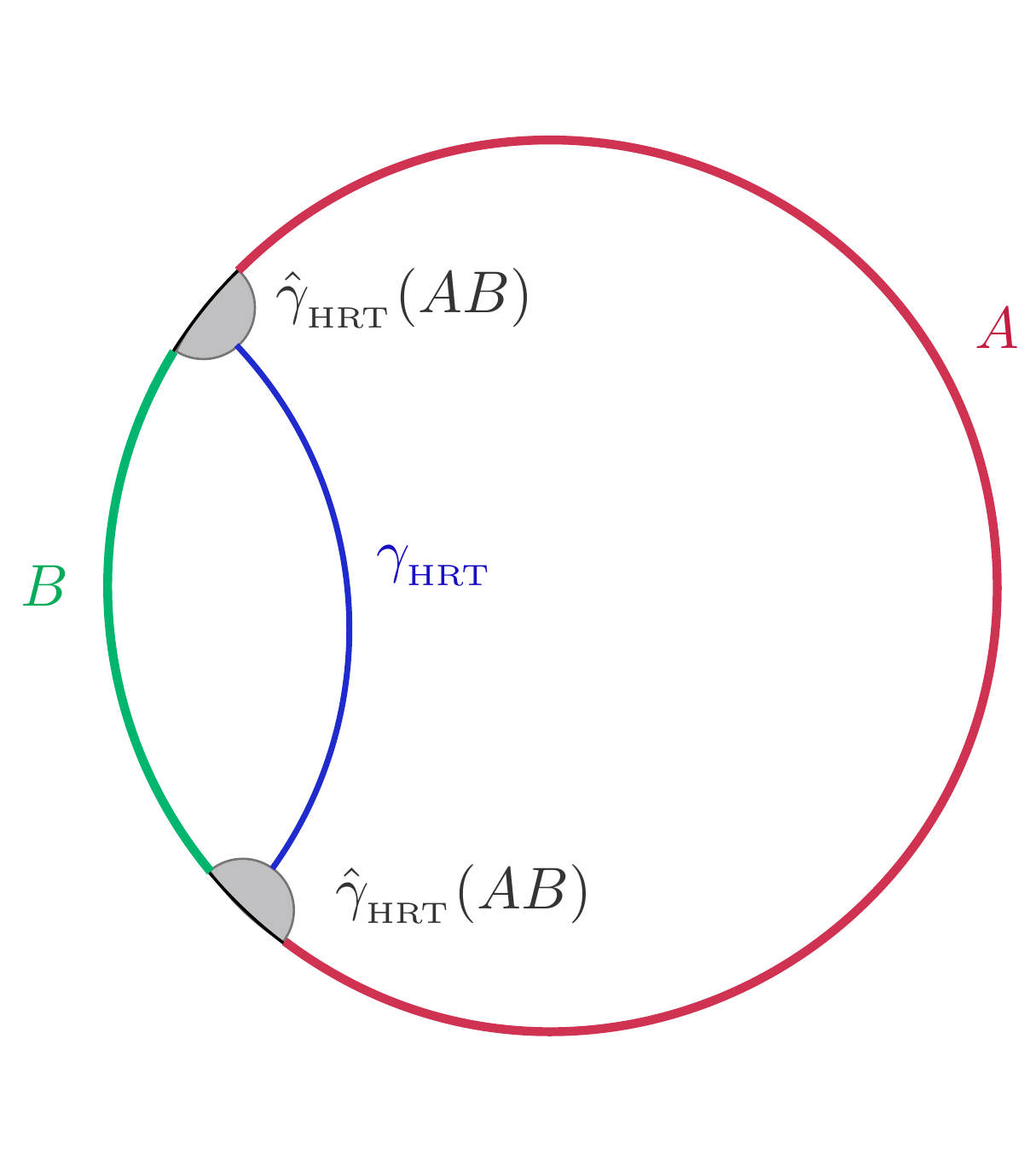}
\hspace{2cm}
\includegraphics[width=0.3\textwidth]{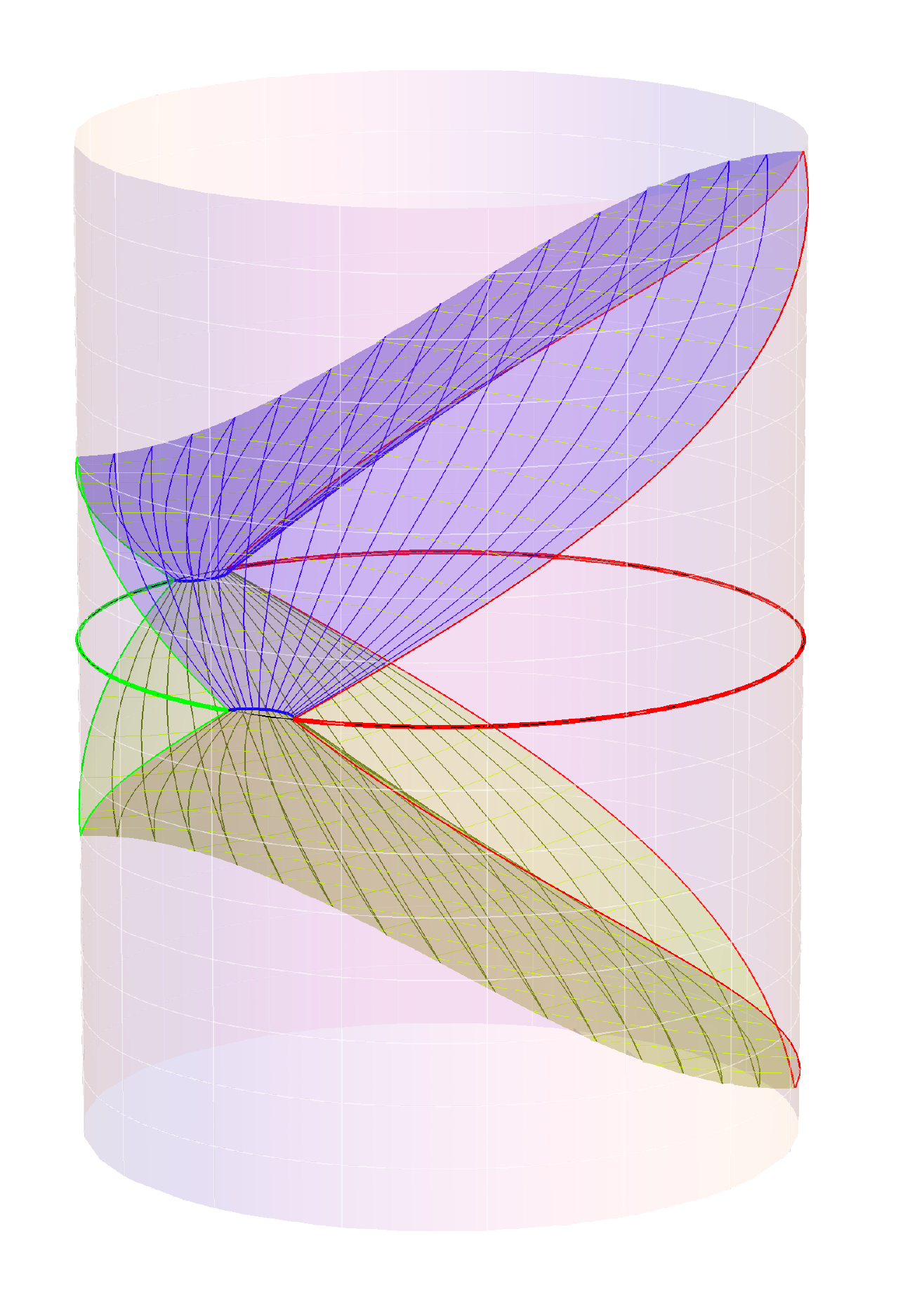}
\setlength{\unitlength}{1cm}
\put(-0.45,3.2){\color{red}{$A$}}
\definecolor{darkgreen}{rgb}{0.2,0.7,0.2}
\put(-4.75,3.2){\color{darkgreen}{$B$}}
\put(-2.,3.3){$\M$}
\put(-3,5){$\I^+$}
\put(-3,1.6){$\I^-$}
\caption{\label{fig:ewcsreg}
EWCS regulator.  
\figL: On the Poincare disk, the region outside the $AB$ entanglement wedge (shaded gray) is excised from the spacetime, and the regulated HRT surface for $A$, $\HRT:=\HRT(A)$ (blue curve), is anchored on $\hHRT(AB)$.  
\figR: In the Lorentzian AdS spacetime, the future and past of the removed regions is corresponding excised, so that the regulated spacetime $\M$ is bounded by the ingoing null congruences from $\hHRT(AB)$ shown.
}
\end{figure}

Now, the key point that makes the EWCS an appealing regulator from our viewpoint is that computing the EWCS is equivalent to simply applying the usual HRT formula \emph{within the spacetime defined by the $AB$ entanglement wedge}, with its future and past boundaries playing the role of $\I^\pm$.   This is illustrated in figure \ref{fig:ewcsreg}.
This entanglement wedge is thus essentially acting as a multiboundary wormhole; the only novelty is that the homology must be computed relative to the joint HRT surface $\hHRT(AB)$. (Thus here $\hHRT(AB)$ is playing a role similar to that of an end-of-the-world brane, for which we would also use relative homology.) This is the point of view we will take in most of this paper: the conformal boundary $\N$ will by definition consist \emph{only} of $D(A)$ and $D(B)$, and the bulk $\M$ \emph{only} of their joint entanglement wedge $\hat\W(AB)$. The rest of the original boundary and bulk spacetime will simply be discarded. However, in section \ref{sec:embedding}, we will return to the original spacetime (denoted $\hat\M$, with future/past boundaries $\hat\I^\pm$), and explain how our constructions can be extended to include it. In general, we will use hatted symbols to refer to constructs defined in the original, unregulated spacetime, and unhatted symbols to refer to constructs in the regulated spacetime.

Of course, we may be interested in calculating the entropy of more than one boundary region. For example, we may want to compute the mutual information $I(A:B):=S(A)+S(B)-S(AB)$ between disjoint regions $A$, $B$ of a common boundary Cauchy slice $\bdyslice$. If $A$, $B$ are separated (not touching), then the mutual information is finite and regulator-independent, but to compute it using the HRT formula we must first regulate the entropies appearing in the definition. To get a meaningful answer, it is important to use the same regulator for all three of those entropies. So we should choose a region $C\subset(AB)^c$ that leaves a buffer between $AB$ and $C$. We then proceed as above, letting $\M$ be the entanglement wedge $\hat\W(ABC)$ and deleting the rest of the original spacetime. The new conformal boundary is $\N=D(A)\cup D(B)\cup D(C)$, and the (regulated) mutual information is $S(A:BC)+S(B:AC)-S(AB:C)$. More generally, however many boundary regions $A,B,C,\ldots$ we are interested in, we assume that all entangling surfaces have been buffered and deleted from the boundary Cauchy slice, and the bulk replaced with the entanglement wedge of the remainder. We will call this entanglement wedge the ``regulated spacetime''.

The authors of \cite{Grado-White:2020wlb} studied a different kind of regulator for the HRT formula, in which the Cauchy slices are anchored on a fixed surface (spacelike codimension-2 submanifold) at finite distance, and the HRT surface is anchored on a fixed codimension-3 submanifold on that surface. Our formalism, described in the next subsection, can also accommodate this choice of regulator.

\subsection{Spacetime setup}
\label{sec:setup}

As explained in the previous subsection, we assume that the boundary regions we consider --- or more precisely their causal domains $D(A)$, $D(B)$, etc.\ --- are composed of entire connected components of the boundary $\N$. The bulk $\M$ is either a multiboundary wormhole or the regions' joint entanglement wedge within the original, unregulated spacetime. In the latter case, the future and past boundaries of the bulk are null hypersurfaces emanating from the joint HRT surface $\surf^0:=\hHRT(AB\cdots)$. In imposing the homology condition on the HRT (or cross-section) surface for any of the regions (or subset of them), we work in homology relative to $\surf^0$, which implies that the surface is allowed to end on $\surf^0$.

We would also like to allow for the possible presence of a timelike boundary at finite distance: an end-of-the-world brane that we call $\I^0$. This brane is assumed not to carry entropy, so its area does not count in computing the HRT entropy (or EWCS), and the HRT (or cross-section) surface is allowed to end on it. This is reflected mathematically in the fact that when we apply the homology condition we work in homology relative to $\I^0$.

$\I^0$ and $\surf^0$ play the same role: the HRT (or cross-section) surface is allowed to end on them, and we work in homology relative to them. The main difference is that $\I^0$ is codimension-1 whereas $\surf^0$ is codimension-2. Notationally, it will be a significant simplification not to have to deal separately with the two cases. For this reason, we subsume $\surf^0$ into $\I^0$, by imagining that $\surf^0$ has an infinitesimal extent in the time direction and is therefore formally codimension-1. Similarly, we may want to employ the regulator suggested in \cite{Grado-White:2020wlb}, in which the conformal boundary $\N$ is replaced by a (spacelike codimension-2) surface $n$ at finite distance, which is divided into regions $A,B,\ldots$. Everything we do in this paper goes through for that setup. Just like for $\surf^0$, in order to avoid dealing separately with that case, imagine that $n$ has an infinitesimal extent in the time direction, making it formally codimension-1. We emphasize that these maneuvers are purely for the purpose of decluttering our formulas. The diligent reader is invited to follow all the steps of our derivations with codimension-2 boundaries  $\surf^0$, $n$ in place of (or in addition to) $\I^0$, $\N$.

With that formality out of the way, we now lay out in a slightly more mathematically precise way our setting, assumptions, and terminology.  We fix a compact oriented manifold-with-boundary $\bar\M$ of dimension $D\ge3$.\footnote{\, With certain notational adjustments, our analysis applies to the $D=2$ case as well, e.g.\ in the setting of Jackiw-Teitelboim gravity \cite{Jackiw:1984je,Teitelboim:1983ux}. For this purpose, the dilaton field $\Phi$ plays the role of the ``area''. A codimension-2 ``surface'' is then a set of points $x_i$ and its effective area is $\sum_i\Phi(x_i)$; the effective flux of a 1-form $W$ over a hypersurface $H$ is $\int_H\Phi\,{*W}$ (compare to \eqref{fluxdef}); the divergenceless condition is $d(\Phi{*W})=0$ (compare to \eqref{divergenceless}); and so on. An easy way to understand the $D=2$ case is by ``dimensional oxidation'' to $D=3$: convert $\M$ into a 3-dimensional spacetime by fibering over it a circle of circumference $\Phi$, with all fields and geometric objects lifted to ones translationally invariant along the circle.} The interior $\M$ of $\bar\M$ has a Lorentzian metric $g$ and a time orientation. The causal structure of $(\M,g)$ extends to a causal structure on $\bar\M$. The boundary of $\bar\M$ is the union of the following compact codimension-1 submanifolds-with-boundary, which are disjoint except along codimension-2 common boundaries:
 \begin{itemize}
     \item $\N$ is timelike.
     \item $\I^-$ is spacelike and/or null and is the past boundary; every inextendible timelike curve in $\bar\M$ begins in $\I^-$.
          \item $\I^+$ is spacelike and/or null and is the future boundary; every inextendible timelike curve in $\bar\M$ ends in $\I^+$.
          \item $\I^0$ may be empty; otherwise it is timelike and the metric $g$ on $\M$ extends to the interior of $\I^0$ (in other words points in the interior of $\I^0$ are at finite distance from points in $\M$).
 \end{itemize}
 $\I^+$ and $\I^-$ are disjoint (though they may both emanate from the infinitesimally time-stretched surface $\surf^0$). We define
 \begin{equation}
  \I:=\I^+\cup\I^0\cup\I^-\,.   
 \end{equation}
 
\begin{figure}[tbp]
\centering
\includegraphics[width=0.5\textwidth]{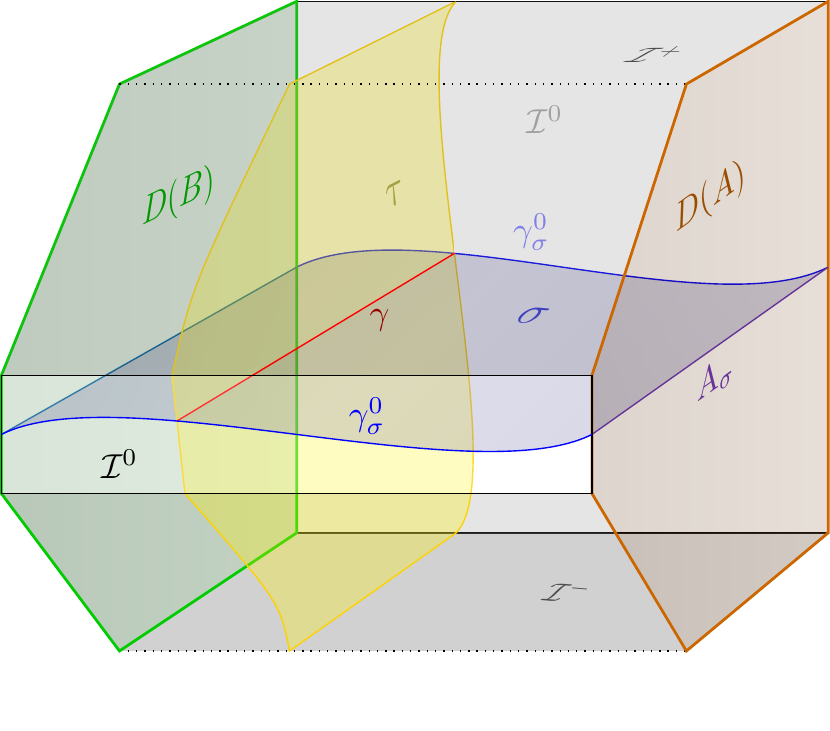}
\caption{\label{fig:setup}
Schematic summary diagram of spacetime setup. The boundary of the spacetime $\bar\M$ consists of the conformal boundary $\N=D(A)\cup D(B)$ (the right and left sides respectively); the future boundary $\I^+$ (top faces); the past boundary $\I^-$ (bottom faces); and possibly an end-of-the-world brane $\I^0$ (front and back faces). $\I^\pm$ are spacelike and/or null. $\I^0$ may be empty; otherwise it is timelike. If the bulk $\M$ is the entanglement wedge of $AB$ within a larger holographic spacetime, then $\I^0$ includes the HRT surface for $AB$, thickened slightly in the time direction to make it codimension-1.  The (blue) hypersurface $\bkslice$ is a (Cauchy) slice for $\bar\M$; its intersection with $\I^0$ is denoted $\surf^0_\bkslice$ and its intersection with $D(A)$ is denoted $A_\bkslice$. The (yellow) hypersurface $\ts$ is a time-sheet homologous to $D(A)$ (relative to $\I=\I^+\cup\I^-\cup\I^0$).
The intersection $\surf$ (red) of slice $\bkslice$ and timesheet $\ts$ is a codimension-2 surface.
}
\end{figure}

 A Cauchy slice for $\bar\M$ (which we will abbreviate as ``bulk slice'') is a hypersurface intersected exactly once by every inextendible timelike curve, and that does not intersect $\I^\pm$. 
 (Note that in our definition a slice is achronal but not necessarily acausal, i.e.\ it is allowed to have null regions.) We denote by $\bksliceset$ the set of all bulk slices. We assume that $\bar\M$ is globally hyperbolic in the sense that $\bksliceset$ is non-empty. Given $\bkslice\in\bksliceset$, we define the surface
 \begin{equation}
     \surf^0_\bkslice:=\I^0\cap\bkslice\,.
 \end{equation}
 We also assume that every inextendible timelike curve in $\N$ is inextendible in $\bar\M$. This implies that, for any $\bkslice\in\bksliceset$, $\bkslice\cap\N$ is a (Cauchy) slice for $\N$, and therefore that $\N$ is also globally hyperbolic. See figure \ref{fig:setup} for a summary of this setup.
 
We will make use of the following lemma:
 \begin{lemma}\label{lem:slicefromset}
Any closed achronal subset of $\bar\M$ that does not intersect $\I^\pm$ is contained in a bulk slice.
 \end{lemma}
 \begin{proof}
Let $s$ be a closed achronal subset of $\bar\M$ that does not intersect $\I^\pm$. Let $\bkslice$ be a bulk slice and define the following subset of $\bar\M$:
\be
R:=(J^-(\bkslice)\cup J^-(s))\setminus I^+(s)\,.
\ee
$R$ is a closed subset of $\bar\M$ that obeys $J^-(R)=R$. It is disjoint from $\I^+$ and contains an open neighborhood of $\I^-$. Therefore its future boundary $R\setminus I^-(R)$ is a bulk slice. Furthermore the future boundary contains $s$.
 \end{proof}

 Given that $\N$ is globally hyperbolic, each connected component $\N_i$ of $\N$ is globally hyperbolic as well. Fix a  slice $\bdyslice_i$ for each $\N_i$; then $\bdyslice:=\cup_i\bdyslice_i$ is a  slice for $\N$. A \emph{boundary region} $A$ is defined as the union of a set of the $\bdyslice_i$; its boundary causal domain $D(A)$ is the union of the corresponding $\N_i$.\footnote{\, Importantly, nothing we do will depend on the choice of boundary slices $\bdyslice_i$. In fact, there will be no dependence of anything on $A$ except through $D(A)$. $A$ thus essentially just serves as a label for $D(A)$ and associated objects and quantities.} Regions $A,B,\ldots$ are disjoint. With a few exceptions (namely in subsections \ref{sec:subadditivity} and \ref{sec:multiple}), we will consider only two regions $A$, $B$, so $AB=\bdyslice$. Given any $\bkslice\in \bksliceset$, its intersection with $D(A)$,
 \begin{equation}
     A_\bkslice:=D(A)\cap\bkslice\,,
 \end{equation}
 is a  slice for $D(A)$.
 
We will also work with timelike counterparts to slices that we call \emph{time-sheets}. More precisely, a time-sheet $\ts$ is a piecewise-timelike hypersurface in $\bar\M$ that does not intersect $\N$.\footnote{\, In subsection \eqref{sec:HRTsurface}, we will expand the definition of time-sheets to allow null pieces. For most of the paper, however, it is more convenient to restrict time-sheets to be everywhere timelike. The reason is that we will be very interested in the intersection of a given time-sheet with a given slice; if the time-sheet is timelike then we are guaranteed that they intersect transversely, with a codimension-2 intersection. Otherwise one has to consider separately the case where they coincide on a null hypersurface.} We define $\Gamma_\ts$ as the set of surfaces of the form $\ts\cap\bkslice$ for some $\bkslice\in\bksliceset$. We denote by $\tsset$ (or $\tsset_A$ when considering more than two boundary regions) the set of time-sheets homologous to $D(A)$ relative to $\I$. Note that such a time-sheet may have ``seams'' where two timelike pieces meet on their common future or past boundary.

In sections \ref{sec:relaxation}--\ref{sec:threads}, we do not impose any equations of motion or energy conditions on the metric on $\M$, or any boundary conditions (such asymptotically AdS ones) beyond those given above. Starting in section \ref{sec:solutions}, we will make further assumptions about the spacetime, related to the existence and properties of the HRT surface; these are laid out in subsection \ref{sec:HRTsurface}. Then in section \ref{sec:embedding} we return to the full, original spacetime.

\subsection{1-forms}
\label{sec:1forms}

Given that we work with a fixed metric on $\M$, a vector $W^\mu$ at a point $x\in\M$ can be expressed equivalently in terms of the covector or 1-form $W=W_\mu dx^\mu$ or the $(D-1)$-form $*W$. We will mainly employ the 1-form notation, supplemented by liberal use of the Hodge star.

\subsubsection{Pointwise notions} 
\label{sec:pointwise}

Fix a point $x\in\M$ and let $T^*:=T^*_x$ be the cotangent space at $x$. Given covectors $W,X\in T^*$, we define
\begin{equation}
    W\cdot X:=W_\mu X^\mu\,,\qquad
    W^2:=W\cdot W\,,\qquad
    |W|:=\sqrt{|W^2|}\,.
\end{equation}
We say $W$ is future-directed causal if it evaluates non-negatively on any future-directed causal vector (or equivalently if $W^2\le0$ and the time component $W_0$ is non-negative, or equivalently if the dual vector $W^\mu$ is \emph{past}-directed causal); similarly for  future-directed timelike. We define $\fdc\subset T^*$ as the set of future-directed causal covectors at $x$, and $\fdt\subset T^*$ for the future-directed timelike covectors.

\paragraph{Covector pairings:} We will often make use of two real pairings between covectors. The first is the norm of the 2-form $W\wedge X$:\footnote{\, We are employing the convention where the norm of the $p$-form $\omega$ is defined in terms of its components as $|\omega|:=\sqrt{|\omega_{\mu_1\cdots\mu_p}\omega^{\mu_1\cdots\mu_p}|/p!}$.}
\begin{equation}\label{WXperp}
|W\wedge X| = \begin{cases}
|W\cdot X|\,,\quad&W^2=0\text{ or }X^2=0 \\
|W||X_\perp|\,,\quad &W^2\neq0 \\
|W_\perp||X|\,,\quad &X^2\neq0 \,
\end{cases}
\end{equation}
where
\begin{equation}
 X_\perp:= X-\frac{X\cdot W}{W^2}W\,,\qquad
  W_\perp:= W-\frac{X\cdot W}{X^2}X
 \end{equation}
are the projections of $X$ perpendicular to $W$ and vice versa. The following lemma gives some important basic properties of this function. (Although the function $|W\wedge X|$ is symmetric, in the lemma we treat its arguments asymmetrically, for reasons that will become clear below.)
\begin{lemma}\label{lem:wedgeprops}
{\rm (a)} $|W\wedge X|$ is continuous and homogeneous in each argument. 
{\rm (b)} For fixed $W\in\fdc$, $|W\wedge X|$ is a convex function of $X$. 
{\rm (c)} For fixed spacelike or null $X$, $|W\wedge X|$ is concave in $W$ on $\fdc$.
\end{lemma}
\begin{proof} (a) Clear. (b) For $W$ null, $|W\wedge X|=|W\cdot X|$, which is clearly convex in $X$. For $W\in\fdt$, $|X_\perp|$ is convex, being the norm of $X_\perp$ within the spacelike hyperplane orthogonal to $W$. (c) For $X$ null, $|W\cdot X|$ is linear (and therefore concave) in $W$ on $\fdc$. For $X$ spacelike, $|W\wedge X|=|X||W_\perp|$, and $|W_\perp|$ is concave, being the norm on the future solid light-cone within the Lorentzian hyperplane orthogonal to $X$.\footnote{\, 
To see the convex-concave character heuristically, plot the norm function on the vertical axis and the space(time) directions on the horizontal axes.  Then  the norm function looks like an ``upright'' cone (which is convex) when pertaining to spacelike covector, but a ``sideways'' cone (which is concave) when pertaining to a timelike covector.  The constant-norm level sets are the conic sections which respectively give spheres and hyperboloids in the corresponding space(time) sections.
} 
\end{proof}

Second, we will need a function on $\fdc\times T^*$ that is concave in the first argument and convex in the second. $|W\wedge X|$ almost works, except that it is not concave in $W$ when $X$ is timelike. This flaw is remedied by taking the concave hull\footnote{\, The concave hull (or envelope) of a function is the smallest concave function greater than or equal to the given one. The proof that $\ip{W}{X}$ is the concave hull of $|W\wedge X|$ is given in the proof of lemma \ref{lem:wedgedotprops}.} with respect to $W$, yielding the following pairing, which we call ``wedgedot'' 
and denote $\ip{\ }{\ }$:\footnote{\, Note that we define only the scalar quantity $\ip{W}{X}$; unlike $W\wedge X$ and $W\cdot X$, $W\mathclap{\hspace{5.75pt}\cdot}{\wedge}X$ by itself has no meaning.}
\begin{equation}\label{ipdef}
 \ip WX:=\max\{|W\cdot X|,|W\wedge X|\} = \begin{cases}|W\cdot X|\,,\quad&X^2\le0 \\
|W\wedge X|\,,
\quad&X^2\ge0\end{cases} \,.
 \end{equation}
(See lemmas \ref{lem:supWU} and \ref{lem:supVX} in subsection \ref{sec:lemmaproofs} for two further ways of writing $\ip{W}{X}$.) While this function may appear from its definition to be symmetric in its two arguments, for our purposes it is crucial that its domain $\fdc\times T^*$ is \emph{not} symmetric. The following lemma spells out several properties of the wedgedot pairing that will play important roles in this paper.
\begin{lemma}\label{lem:wedgedotprops}
On its domain $\fdc\times T^*$, the function $\ip WX$ is continuous and homogeneous in each argument; concave in $W$ for fixed $X$; and convex in $X$ for fixed $W$.  It is also the unique concave-convex function on $\fdc\times T^*$ that equals $|W\wedge X|$ for $X^2\ge0$.
\end{lemma}
\begin{proof}
For fixed $W$, homogeneity, continuity, and convexity in $X$ follow from the fact that $|W\cdot X|$ and $|W\wedge X|$ both have those properties (see lemma \ref{lem:wedgeprops}), and those properties are inherited by the maximum. 

For fixed spacelike $X$, the continuity, homogeneity, and concavity in $W$ are given in lemma \ref{lem:wedgeprops}. For fixed future- or past-directed causal $X$, the function $|W\cdot X|$ is linear in $W$ (on $\fdc$) and therefore continuous, homogeneous, and concave.

For the uniqueness of the concave-convex extension of $|W\wedge X|$,
note first that $|W\wedge X|=|W\cdot X|$ when either $W$ or $X$ is null. When $X^2\le0$, $|W\cdot X|$ is a linear function of $W$ on its domain $\fdc$; any concave function of $W$ that equals $|W\wedge X|$ when $W$ is null is therefore greater than or equal to $|W\cdot X|$. (This fact makes $|W\cdot X|$ the concave hull of $|W\wedge X|$ with respect to $W$ for fixed timelike $X$. For fixed spacelike $X$, $|W\wedge X|$ is already concave with respect to $W$ and is therefore its own concave hull.) Similarly, any convex function of $X$ that equals $|W\wedge X|$ when $X$ is null is, for $X^2\le0$, less than or equal to $|W\cdot X|$.
\end{proof}

\begin{figure}[h!]  % - - - - - - - - - illustration of W+-X in j+
\definecolor{Wcol}{rgb}{0.7,0.2,0.9}  % color for W
\definecolor{Xcol}{rgb}{0.3,0.7,0.8}  % color for X
  \begin{center} 
  \begin{tikzpicture}        % fix W; allowed X
    \tikzmath{
    \Wx=0.7;
    \Wt=1.9;
    }
      \draw (-2,0) -- (2,0);
      \draw (0,-2) -- (0,2);
      \filldraw[Xcol,fill opacity=0.3] (\Wx,\Wt) -- (\Wt,\Wx) -- (-\Wx,-\Wt) -- (-\Wt,-\Wx) -- cycle;
      \draw[thick,Wcol,->] (0,0) -- (\Wx,\Wt) node[above] {$W$};
      \draw[Xcol] (-2,1.5) -- (-1.5,1.5) -- (-1.5,2) node[below left] {$X$};
  \end{tikzpicture}
  \hspace{1cm}
 \begin{tikzpicture}        % fix X spacelike; allowed W
    \tikzmath{
    \Xx=1.9;
    \Xt=0.7;
    \wt=3.5-\Xx-0.5;
    }
      \draw (-2,0) -- (2,0);
      \draw (0,-0.5) -- (0,3.5);
      \fill[Wcol,opacity=0.3] (\Xt,\Xx) -- (\Xt+\wt,\Xx+\wt) -- (\Xt-\wt,\Xx+\wt) -- cycle;
      \draw[Wcol] (\Xt-\wt,\Xx+\wt) -- (\Xt,\Xx) -- (\Xt+\wt,\Xx+\wt);
      \draw[thick,Xcol,->] (0,0) -- (\Xx,\Xt) node[above] {$X$};
      \draw[thick,dotted,Wcol,->] (0,0) -- (\Xt,\Xx) node[below right] {$W_0$};

      \draw[Wcol] (-2,3) -- (-1.5,3) -- (-1.5,3.5) node[below left] {$W$};
  \end{tikzpicture}
  \hspace{1cm}
 \begin{tikzpicture}        % fix X timelike; allowed W
    \tikzmath{
    \Xx=0.7;
    \Xt=1.9;
    \wt=3.5-\Xt-0.5;
    }
      \draw (-2,0) -- (2,0);
      \draw (0,-0.5) -- (0,3.5);
      \fill[Wcol,opacity=0.3] (\Xx,\Xt) -- (\Xx+\wt,\Xt+\wt) -- (\Xx-\wt,\Xt+\wt) -- cycle;
      \draw[Wcol] (\Xx-\wt,\Xt+\wt) -- (\Xx,\Xt) -- (\Xx+\wt,\Xt+\wt);
      \draw[thick,Xcol,->] (0,0) -- (\Xx,\Xt) node[below right] {$W_0=X$};
           \draw[thick,dotted,Wcol,->] (0,0) -- (\Xx,\Xt) node[below right] {$W_0$};
      \draw[Wcol] (-2,3) -- (-1.5,3) -- (-1.5,3.5) node[below left] {$W$};
  \end{tikzpicture}
     \caption{Illustration of the constraint \eqref{condition} in the 2-dimensional case.  
     \figL: The set of allowed $X$ for fixed  $W \in \fdt$ (purple arrow).
    \figM:  The set of allowed $W$ for fixed spacelike $X$ (cyan arrow); here $W_0\in\fdt$ is normal to $X$ and of equal magnitude.    
    \figR: The set of allowed $W$ for fixed $X\in \fdt$; here $W_0=X$.
    } 
         \label{fig:WXfdc}
  \end{center}
\end{figure}
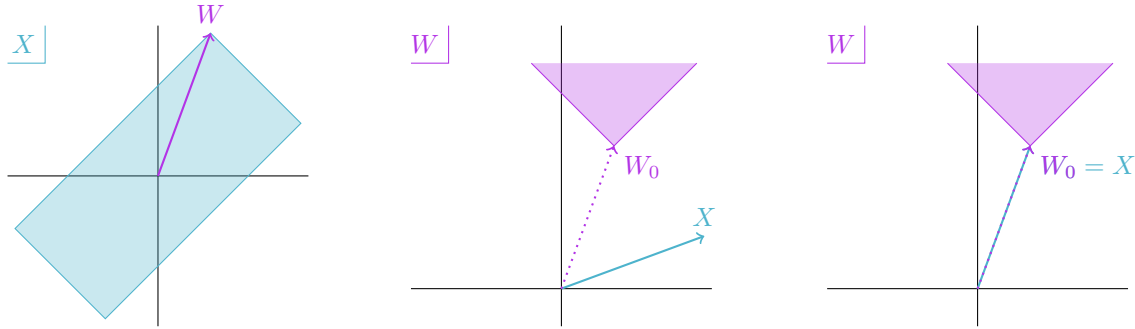  % - - - -

\paragraph{Covector pairs:} An important role throughout the paper will be played by the set of pairs of covectors $(W,X)$ obeying
\be\label{condition}
W\pm X\in\fdc\,.
\ee
Such pairs form a convex subset of $T^*\times T^*$. To get some intuition for this set, first fix $W$ (which is necessarily in $\fdc$). If $W$ is timelike, then \eqref{condition} is equivalent to $|X_\perp|+|X_\parallel|\le|W|$, where $X_\perp,X_\parallel$ are the projections of $X$ orthogonal and parallel respectively to $W$; or, to say it another way, $X$ is required to be in the ``causal diamond'' within $T^*$ with vertices $\pm W$. (See figure \ref{fig:WXfdc}, left panel.) For $W$ null, the diamond degenerates to a null line segment, and $X$ must be a convex combination of $\pm W$ (and therefore must also be null). On the other hand, if we fix $X$, then $W$ is restricted to the ``causal future'' (within $T^*$) of a point $W_0\in\fdc$ defined as follows: for spacelike $X$, $W_0$ is the (unique) covector in $\fdc$ obeying $X\cdot W_0=0$, $|W_0|=|X|$; for $\pm X\in\fdc$, $W_0=\pm X$. (See figure \ref{fig:WXfdc}, center and right panels.)

Further characterizations of the set of covector pairs obeying \eqref{condition} are given in subsection \ref{sec:lemmaproofs}. It is also shown there that this set is closely related to the wedgedot pairing defined above. Specifically, \eqref{condition} is equivalent to either of the following conditions:
\be
\forall\, Y\in\fdc\,,\quad\ip{Y}{X}\le-Y\cdot W
\ee
\be
W\in\fdc\text{ and }\forall \,Y\in T^*\,,\quad X\cdot Y\le\ip{W}{Y}\,.
\ee
Conversely, the wedgedot can be derived from the condition \eqref{condition}, in two different ways. For any $W\in\fdc$ and $X\in T^*$,
\be
\ip{W}{X}=\inf_{\substack{Y:\\Y\pm X\in\fdc}}(-W\cdot Y)
=\sup_{\substack{Y:\\W\pm Y\in\fdc}}Y\cdot X\,.
\ee
This is another way to see that $\ip{W}{X}$ is concave in $W$ and convex in $X$: $-W\cdot Y$ is concave in $W$ and the pointwise infimum of a set of concave functions is concave; similarly, $Y\cdot X$ is convex in $X$ and the pointwise supremum of a set of convex functions is convex.

\subsubsection{1-form fields}
 \label{sec:1formfields}

The vector field $W^\mu$ is divergenceless if and only if $*W$ is closed:
\begin{equation}\label{divergenceless}
    \nabla_\mu W^\mu=0\quad
    \Longleftrightarrow\quad
    d{*W}=0\,.
\end{equation}
We will sometimes use the shorthand ``$W$ is divergenceless'' (strictly speaking ``co-closed'' would be more correct). The flux of $W^\mu$ through a spacelike or timelike hypersurface $H$ can also be expressed in terms of $*W$,
\begin{equation}\label{fluxdef}
\int_H\sqrt{h}\,N\cdot W = \int_H*W\,,
\end{equation}
where $h$ is the determinant of the induced metric on $H$ and $N$ is a unit normal covector field, and similarly for the no-flux condition,
\begin{equation}\label{noflux}
N\cdot W|_H=0\quad
    \Longleftrightarrow\quad
    *W|_H=0\,,
\end{equation}
where $*W|_H$ is the pull-back of $*W$ onto $H$. The formulas in terms of $*W$ are slightly more general than the ones in terms of $N\cdot W$, since they apply even when $H$ is null. Even more generally, we will be using these formulas on the boundaries $\N$, $\I^\pm$ of $\bar\M$ where the metric, and therefore the Hodge star, are not defined; implicitly, one is using the limiting value of $*W$ as the boundary is approached. A more careful notation would instead use a $(D-1)$-form $\tilde W$ on $\bar\M$ as the underlying variable, so that in \eqref{divergenceless}, \eqref{fluxdef}, \eqref{noflux} we would have $d\tilde W$, $\int_H\tilde W$, $\tilde W|_H$ respectively; and then define the 1-form $W:=*\tilde W$ on $\M$ where the metric is defined.\footnote{\label{foot:reconstruction}\, 
If one is interested in applying the V-flow and U-flow formulas discussed in section \ref{sec:flows} in a setting where the metric is not fixed a priori, such as metric reconstruction or deriving the Einstein equation from the HRT formula, then it is probably better to use the $(D-1)$-forms $\tilde V=*V$, $\tilde U=*U$, rather than the 1-forms $V$, $U$ themselves, as the fundamental quantities defining a flow.
} One may think of $\tilde W$ as the collection of integral curves of $W^\mu$, namely a 1-dimensional oriented structure which can be integrated over a hypersurface to get a number, and in this sense $\tilde W$ naturally captures the original notion of the bit threads. Nonetheless, at the cost of a slight notational sloppiness we will stick to the more convenient 1-form notation.

For the boundary of $\bar\M$, we will use an orientation in which the normal covector $N$ is \emph{inward}-directed. This is slightly non-standard but simplifies the homology relations (e.g.\ it makes the entanglement horizon $\hor(A)$ homologous to $D(A)$). In particular, $N$ is past-directed on $\I^+$ and future-directed on $\I^-$. Hence the flux of a future-directed 1-form is positive through $\I^+$ and negative through $\I^-$. With this convention, Stokes' theorem takes the following form:
\begin{equation}\label{Stokes}
    \int_\M d{*W} +\int_\I*W+\int_\N*W=0\,.
\end{equation}
On a  slice $\bkslice$, the normal covector is chosen to be past-directed, so that (as on $\I^+$) a future-directed covector has a positive flux. In terms of homology, these orientations imply the following relations:
\begin{equation}
    \I^+\sim\bkslice\sim-\I^-\text{ rel }(\I^0\cup\N)
\end{equation}
(the spacetime homology region between $\bkslice$ and $\I^\pm$ being $I^\pm(\bkslice)$). 

We will use a similar notation for 1-forms on a given hypersurface, but we will denote them with lower-case letters $v,u$. The Hodge star and exterior derivative applied to these forms are always the ones defined on the hypersurface, not in the ambient space.

%---------------------------------------------------

\section{Games \& relaxation}
\label{sec:relaxation}

We start by recalling the maximin formula
\cite{Wall:2012uf}:
\begin{equation}\label{maximin}
\Sminus(A:B):= \sup_{\bkslice\in \bksliceset}\ \infp_{
\surf\in\Gamma_{\bkslice}
}
\area(\surf)\,,
\end{equation}
where $\Gamma_{\bkslice}$ is the set of surfaces $\surf$ in $\sigma$ homologous to $A_\bkslice$ relative to $\eowsurf_\bkslice:=\I^0\cap\bkslice$, that do not intersect $\N$. The reason for the subscript on $S_-$ will become clear shortly. In this section we will derive a number of variations on \eqref{maximin}. In subsection \ref{sec:maximin}, we will rewrite it in two ways: first, by using the Riemannian max flow-min cut theorem, in terms of a flow localized on a slice; and then in terms of the intersection of a slice and a time-sheet. By switching the order of the maximization and minimization, this will then allow us to obtain a ``minimax'' formula. The relation between such maximin-minimax pairs of formulas is the subject of minimax theory, which is closely related to game theory and which we briefly review in subsection \ref{sec:minimaxtheory}. This will lead us in subsection \ref{sec:convexrelaxation} to convex relax the two formulas, yielding a third one, whose value sits between them. This convex-concave formula will be the starting point for our derivation of flow formulas in section \ref{sec:flows}.

The maximin quantity is symmetric, $S_-(B:A)=S_-(A:B)$, and the same holds for all of the quantities we derive in this section; this is hopefully clear by inspection. Furthermore, since throughout this section we fix $A,B$, we will from now on simply write $S_-$ for $S_-(A:B)$.

Throughout this section, as well as sections \ref{sec:flows} and \ref{sec:threads}, we assume only the basic structure for the bulk and boundary spacetimes described in subsection \ref{sec:setup} (essentially global hyperbolicity), not any particular boundary conditions or energy conditions.

\subsection{Variations on maximin}
\label{sec:maximin}

Our first alternative formula for $\Sminus$ is obtained by applying the Riemannian max flow-min cut theorem to replace, within each slice $\bkslice$, the minimization over surfaces by a maximization over flows.\footnote{\, Strictly speaking, this situation does not quite fit the assumptions of the RMFMC theorem proved in \cite{Headrick:2017ucz}, which applies to compact Riemannian manifolds. While $\bkslice$ is compact, the metric on it does not extend to $\N$ and may include null pieces. The first issue can be dealt with by  removing a neighborhood of $\N\cap\bkslice$. The second issue can be dealt with either by considering $\bkslice$ as a limit of spacelike manifolds, or by treating the null locus following the treatment of null manifolds in \cite{Headrick:2017ucz} (but for max flow-min cut rather than min flow-max cut as in that paper). Specifically, on the null locus there is a unique $(D-2)$-form $\omega$ such that the area of any surface $m$ equals $|\int_\surf\omega|$. A flow is defined not by the 1-form $v$ but by the $(D-2)$-form $*v$, and the constraint $|v|\le1$ is replaced by $*v=\alpha\omega$, where $|\alpha|\le1$.\label{foot:omegadef}} 
A 1-form $v$ on $\bkslice$ is called a \emph{$\bkslice$-flow} if it has the following properties:
\begin{equation}\label{sigmaflowdef}
    |v|\le1\,,\qquad
d{*v} = 0\,,\qquad
    * v|_{\surf^0_\bkslice}=0
\end{equation}
(where $d$ is the exterior derivative on $\bkslice$, and $|\cdot|$ and $*$ are defined with respect to the induced metric). We call the set of $\bkslice$-flows $\F_\bkslice$. The RMFMC theorem (see \cite{Headrick:2017ucz} and references therein) states that
\begin{equation}\label{mfmc}
    \inf_{\surf\in\Gamma_\bkslice}
    \area(\surf)=\sup_{v\in\F_\bkslice}\ \int_{A_\bkslice}* v \,.
\end{equation}
$\F_\bkslice$ is a convex set, and the objective functional $\int_{A_\bkslice}* v$ is linear; hence the right-hand side of \eqref{mfmc} defines a convex program.\footnote{\, Recall that a convex program is defined as the problem of minimizing a convex function $f_0$ over a convex subset $X$ of an affine space, subject to constraints $f_i(x)\le0$, $g_i(x)=0$, where the $f_i$ are convex functions and the $g_i$ are affine functions on $X$. The constraints defining $X$ are \emph{implicit}, while the constraints $f_i(x)\le0$, $g_i(x)=0$ are \emph{explicit}. Thus, strictly speaking, the right-hand side of \eqref{mfmc} defines a convex program only after one has decided whether each of the constraints \eqref{sigmaflowdef} is implicit or explicit.} Using \eqref{mfmc}, we can write $\Sminus$ in terms of a ``maximax'' formula:
\begin{equation}\label{maximax}
\Sminus = \sup_{\bkslice\in \bksliceset}\ 
    \sup_{v\in\F_\bkslice}\ \int_{A_\bkslice}* v \,.
\end{equation}

Alternatively, in order to put the space and time variations in \eqref{maximin} on an equal footing, we can change the minimization domain so that it does not depend on the maximization variable $\bkslice$. This can be done by thinking of the surface $\surf$ as the intersection of the slice $\bkslice$ with a time-sheet $\tau\in\tsset$ (where $\tsset$ is the set of time-sheets homologous to $D(A)$ relative to $\I$), as justified by the following lemma:
\begin{lemma}\label{lem:surfsheet}
Fix $\bkslice\in\bksliceset$. For any $\ts\in\tsset$,  $\bkslice\cap\ts\in\Gamma_{\bkslice}$. Conversely, for any $\surf\in\Gamma_{\bkslice}$, there exists a $\ts\in\tsset$ such that $\bkslice\cap\ts=\surf$.
\end{lemma}
\begin{proof}
Given $\ts\in\tsset$, $\bkslice\cap\ts$ is homologous to $A_\bkslice$ on $\bkslice$ (relative to $\eowsurf_\bkslice$) via the intersection of $\bkslice$ with the spacetime homology region between $D(A)$ and $\ts$; and it does not intersect $\N$ (since $\tau$ does not intersect $\N$). Hence it is in $\Gamma_{\bkslice}$.

For the converse, let $t^\mu$ be a continuous future-directed timelike vector field on $\bar\M$ that is tangent to $\N$ and $\I^0$.\footnote{\, It is clear that such a vector field exists on a sufficiently small patch of $\bar\M$. The following construction shows that there is no global obstruction to its existence: start from an atlas of charts for $\bar\M$, define such a vector field on each chart, and average them on the overlaps using a partition of unity. Alternatively, the standard argument for the existence of a globally defined timelike vector field on a Lorentzian manifold, using an auxiliary Riemannian metric (see e.g.\ p.\ 39 of \cite{Hawking:1973uf}), can be upgraded in the presence of a boundary to ensure that the vector field is tangent to the timelike parts of the boundary.} The integral curves of $t^\mu$ pass through every point of $\bar\M$, and each curve starts on $\I^-$, ends on $\I^+$, and lies entirely in $\M$ (except its endpoints), in $\N$, or in $\I^0$. Given $\surf\in\Gamma_{\bkslice}$, let $r_\bkslice$ be the homology region on $\bkslice$ between $A_\bkslice$ and $\surf$. The surface $\surf$ can be extended into a time-sheet $\ts$ by following the integral curves of $t^\mu$ in both directions. $\ts$ is homologous to $D(A)$ (relative to $\I$) via the corresponding extension of $r_\bkslice$.
\end{proof}
With this lemma in hand, we can replace the minimization over $\surf$ in \eqref{maximin} with a minimization over $\ts$:
\begin{equation}\label{maximin2}
\Sminus= \sup_{\bkslice\in \bksliceset}\ \infp_{
\ts\in\tsset
}
\area(\bkslice\cap\ts)\,.
\end{equation}
The formula \eqref{maximin2} invites us to switch the order of the minimization and maximization. We therefore also define the following ``minimax'' quantity:\footnote{\, We use the symbol $S$ here (and below, where we define a third quantity $\Sc$) to emphasize the parallel to the maximin formula for holographic EE. However, in this setting, which is much more general than standard holographic spacetimes and where (as we will see) the three quantities $\Sminus$, $\Splus$, and $\Sc$ are not necessarily equal, we make no claim for any of them being an entropy.}$^,$\footnote{\, Swapping the order of the minimization and maximization in the maximin formula was also considered in \cite{Agon:2019qgh}.}
\begin{equation}\label{minimax}
\Splus:= 
\infp_{\ts\in\tsset}\ 
\sup_{\bkslice\in \bksliceset}
\area(\bkslice\cap\ts)\,.
\end{equation}
Recalling that $\Gamma_\ts$ is the set of surfaces $\gamma=\ts\cap\bkslice$ for some $\bkslice\in\bksliceset$, we can write this
\begin{equation}\label{minimax2}
\Splus= 
\infp_{\ts\in\tsset}\ 
\sup_{\surf\in\Gamma_{\ts}}
\area(\surf)\,,
\end{equation}
expressing the minimax quantity in a form analogous to \eqref{maximin}.

By lemma \ref{lem:slicefromset}, any achronal surface $\surf'$ (that is closed as a subset of $\bar\M$ and does not intersect $\I^\pm$) is contained in a slice $\bkslice$. If $\surf'$ is contained in the time-sheet $\ts$ then $\surf'\subseteq\surf:=\bkslice\cap\ts\in\Gamma_\ts$ and $\area(\surf')\le\area(\surf)$, so we can replace the maximization in \eqref{minimax2} with one over achronal surfaces:\footnote{\, The reader may wonder whether, by analogy with the maximax formula \eqref{maximax}, $\Splus$ can be written in terms of a ``minimin'' formula. Indeed, one may be tempted to apply the Lorentzian max cut-min flow theorem \cite{Headrick:2017ucz} to the supremum in \eqref{minimax2} in order to obtain a formula involving minimizing the flux of a timelike flow on a time-sheet (a timelike flow being a covector field $u$ obeying $u\in\fdt$, $|u|\ge1$, $d{*u}=0$, $*u|_{\I^0\cap\ts}=0$). This would work if the supremum in \eqref{minimax4} were over surfaces $\surf$ that are achronal \emph{within} the time-sheet $\ts$; whereas it is only over surfaces that are achronal in the ambient spacetime $\M$, a stronger condition. Nonetheless, in subsection \ref{sec:time-sheet-flows}, we will define a ``time-sheet-flow'', which is closely related to a Lorentzian flow living on a time-sheet and gives a sort of minimin formula for $\Splus$.}
\begin{equation}\label{minimax4}
\Splus= 
\infp_{\ts\in\tsset}\ 
\sup_{\substack{\surf\subset\ts\\ \text{achronal}}}
\area(\surf)\,.
\end{equation}

\subsection{Minimax theory}
\label{sec:minimaxtheory}

In order to better understand the relation between the two quantities $\Sminus$ and $\Splus$ defined in the previous subsection, we now make a short digression into minimax theory. This theory\footnote{\,  Minimax theory, which was born with J. von Neumann's seminal work on game theory, continues to be an active area of research in analysis, with applications to economics and many other fields. For an overview, see \cite{du2013minimax}.} addresses the following question: Given sets $X,Y$ and a function $f:X\times Y\to\R$, what can we say about the relation between the maximin value $\sup_{x\in X}\infp_{y\in Y}f(x,y)$ and the minimax value $\infp_{y\in Y}\sup_{x\in X}f(x,y)$? We start with three elementary general facts.
\begin{itemize}
\item 
First, the maximin and minimax values are not necessarily equal. A simple counterexample is given by setting $X=Y=\{1,2\}$ and $f(x,y)=(-1)^{x+y}$; then
\begin{equation}\label{evens-and-odds}
\sup_{x\in X}\infp_{y\in Y}f(x,y)=-1\,,\qquad
\infp_{y\in Y}\sup_{x\in X}f(x,y)=1 \,.
\end{equation}
\item 
The maximin and minimax nonetheless do obey a relation. Clearly, for any $x_0\in X$, $y_0\in Y$,
\begin{equation}\label{minimaxin}
\inf_{y\in Y}f(x_0,y) \le f(x_0,y_0)\le
\sup_{x\in X}f(x,y_0)\,.
\end{equation}
Maximizing the left-hand side over $x_0$ and minimizing the right-hand side over $y_0$ yields the \emph{min-max inequality}:
\begin{equation}\label{min-max}
\sup_{x\in X}\infp_{y\in Y}f(x,y) \le
\infp_{y\in Y}\sup_{x\in X}f(x,y)\,.
\end{equation}

The min-max inequality can be understood in simple game-theory terms. Let $X$, $Y$ represent the set of possible strategies for the two players respectively in a single-round, zero-sum game, with payout $f(x,y)$ for player $X$ and $-f(x,y)$ for player $Y$. Then the maximin is the best outcome for player $X$ if she plays first, while the minimax is her best outcome if she plays second (assuming player $Y$ in each case is choosing his best strategy). The inequality \eqref{min-max} expresses the fact that, in such a game, it is often better --- and never worse --- to play second, allowing one to use knowledge of the other player's move to one's advantage. Consider the example of \eqref{evens-and-odds}, which corresponds to the children's game of evens-and-odds; normally in this game the two players play simultaneously, for the simple reason that otherwise the second player would always be able to win.

In the setting of Lagrange duality, the min-max inequality is responsible for the weak duality property. With $f$ the Lagrangian function, the dual pair of programs consists of maximizing $\inf_{y\in Y}f(x,y)$ and minimizing $\sup_{x\in X}f(x,y)$. By the min-max inequality, the optimal value of the maximization problem is bounded above by the optimal value of the minimization problem.
\item 
Finally, if there exists a pair $(x_0,y_0)$ that saturates both inequalities in \eqref{minimaxin}, then both the minimax and maximin equal $f(x_0,y_0)$, so the min-max inequality is saturated:
\begin{equation}\label{globalsaddle}
\sup_{x\in X}\infp_{y\in Y}f(x,y) =f(x_0,y_0)=
\infp_{y\in Y}\sup_{x\in X}f(x,y)\,.
\end{equation}
Such a point is called a \emph{global saddle point}.
\end{itemize}

A minimax theorem is a theorem giving sufficient conditions on $X$, $Y$, and $f$ for the maximin to equal the minimax. The first such theorem was von Neumann's, which applies to a zero-sum game with \emph{mixed strategies}. A mixed strategy for player $X$ is a vector of non-negative numbers $x_i$ ($i=1,\ldots,n$) summing to 1, which can be thought of as a probability distribution over a set of $n$ pure strategies; thus $X$ is the unit simplex in $\R^n$. Similarly for player $Y$. Given a payout matrix $A_{ij}$ for the pure strategies, we assume that the payout function is the expectation value of $A_{ij}$ over the joint distribution $x_iy_j$:
\begin{equation}
f(x,y) = x_iA_{ij}y_j\,.
\end{equation}
According to von Neumann's theorem, for such a game, the minimax equals the maximin, hence the second-player advantage we saw before is erased. For example, in the evens-and-odds game, player $X$'s best strategy if she plays first is to choose equal weights for the two pure strategies, $x_1=x_2=1/2$, leading to a payout $f=0$ regardless of player $Y$'s strategy. The same holds for player $Y$, making it irrelevant who plays first. (This is a simple example of a Nash equilibrium.)

In the case just discussed, the mixed-strategy maximin (and minimax) value 0 sits between the pure-strategy maximin $-1$ and minimax 1. This holds generally:
\begin{equation}\label{mixed min-max}
\max_i\min_jA_{ij}\le\sup_{x\in X}\infp_{y\in Y}f(x,y) =
\infp_{y\in Y}\sup_{x\in X}f(x,y)
\le \min_j\max_iA_{ij}\,.
\end{equation}
Let us prove the first inequality. Consider a particular pure strategy $\hat\imath$ for player $X$. One possible mixed strategy $\hat x\in X$ is to put all the weight on $\hat\imath$: $\hat x_i:=\delta_{i\hat\imath}$. In that case \be
\inf_{y\in Y}f(\hat x,y) = \inf_{y\in Y}A_{\hat\imath j}y_j=\min_jA_{\hat\imath j}\,,
\ee
so $\sup_{x\in X}\inf_{y\in Y}f(x,y)\ge\min_jA_{\hat\imath j}$. This holds for any $\hat\imath$, implying the first inequality.

Von Neumann's theorem can be generalized to the case where $X$ and $Y$ are convex subsets of affine spaces, at least one of them is compact, and $f$ is concave-convex, i.e.\ concave in $x$ for fixed $y$ and convex in $y$ for fixed $x$. There are many results extending this theorem in various ways, including to the infinite-dimensional case, with additional technical assumptions; see for example \cite{pjm/1103040253} and the review \cite{Simons1995}. As is our custom throughout this paper, we will proceed rather naively, as far as the functional analysis is concerned, when applying ideas from convex optimization to function spaces.

\subsection{Convex-concave formula}
\label{sec:convexrelaxation}

\begin{figure}[tbp]
\centering
\includegraphics[width=0.35\textwidth]{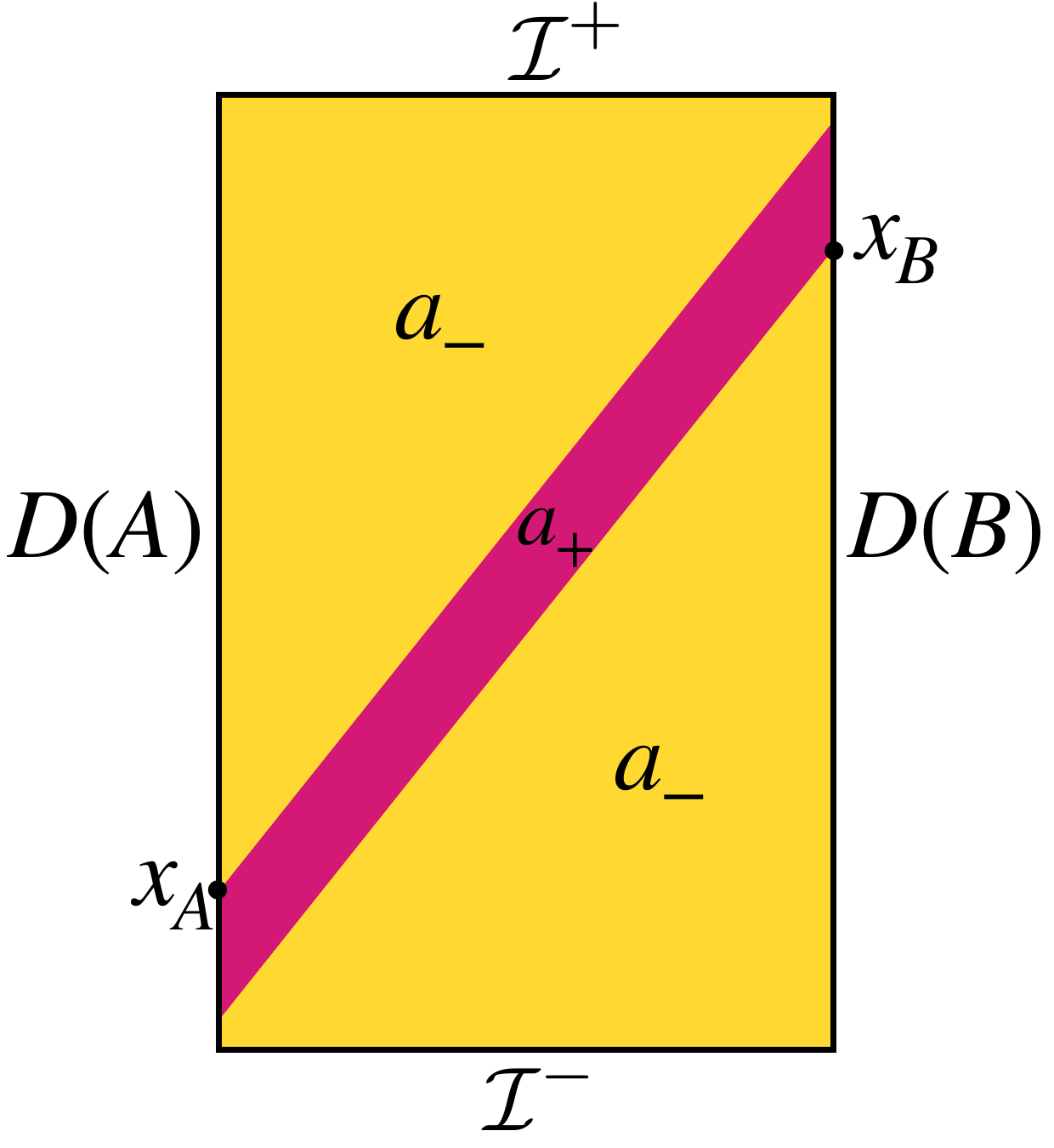}
\caption{\label{fig:gap}
Penrose diagram of a spacetime in which $\Sminus\neq \Splus$. The spacetime consists of a $(D-2)$-sphere fibered over a rectangular portion of 2-dimensional Minkowski space. The left and right sides of the rectangle are $D(A)$ and $D(B)$ respectively, while the top and bottom are $\I^\pm$ respectively. ($\I^0$ is empty in this example.) The area of the sphere varies on the rectangle, being equal to $a_-$ on the yellow portion and $a_+>a_-$ on the purple stripe. The edges of the purple stripe are moving at a constant speed $\beta<1$, and furthermore the point $x_B$ is timelike-related to $x_A$ ($x_B\in I^+(x_A)$) implying that no slice lies entirely within the purple stripe. Since every slice $\bkslice$ passes through the $a_-$ region, $\inf_{\surf\in\Gamma_\bkslice}\area(\surf)\le  a_-$; furthermore, this inequality is saturated by spherically-symmetric slices, so $\Sminus=a_-$. On the other hand, every time-sheet $\ts$ passes through the $a_+$ region, so $\sup_{\surf\in\Gamma_\ts}\area(\surf)\ge a_+$; furthermore, this inequality is saturated by any time-sheet that is spherically symmetric and does not have any seams, so $\Splus=a_+$. As expected, $\Sminus\le \Splus$. Note that this spacetime does not obey the null energy condition, as can be seen from the fact that the area of a spherical null congruence starting in the $a_-$ region initially has vanishing expansion, yet, if it enters the $a_+$ region, its area increases. (Observe however that, if we replaced the $a_-$ in the lower portion of the figure by an area $a_{++}>a_+$, the maximin and minimax values would coincide and both equal $a_{++}$.)
}
\end{figure}

We now return to the setting at hand, and the maximin and minimax formulas \eqref{maximin2}, \eqref{minimax} respectively. From the min-max inequality \eqref{min-max} we have
\begin{equation}
\Sminus \le \Splus
\end{equation}
(thereby explaining the subscripts). Are there spacetimes where the two quantities are  unequal? As we will see, this depends on the assumptions one makes about the spacetime. In the very general framework we use in this section, as set out in subsection \ref{sec:setup} (essentially just global hyperbolicity), one can indeed construct spacetimes with a gap between the maximin and minimax; an example is shown in figure \ref{fig:gap}. On the other hand, in section \ref{sec:solutions}, we will use Wall's results \cite{Wall:2012uf} to show that, under standard assumptions about holographic spacetimes (specifically, the null energy condition and AdS boundary conditions), the HRT surface is a global saddle point, and therefore its area equals both $\Sminus$ and $\Splus$.

For the next few sections, we will instead close the gap between the maximin and minimax by following von Neumann's method of allowing ``mixed strategies'', in other words by convex-relaxing the domains $\bksliceset$ of slices and $\tsset$ of time-sheets. The standard method for convex-relaxing such sets of hypersurfaces is to introduce a scalar function obeying certain boundary conditions, whose level sets represent a ``smeared'' hypersurface. (For an overview and details, see \cite{Headrick:2017ucz}.) Specifically, a convex combination of slices is represented by a function $\phi$ on $\bar\M$:\footnote{\, Note that the restriction of $\phi$ to lie in the interval $[-1/2,1/2]$ is automatic, given the other constraints on $\phi$ and the fact that every point in $\bar\M$ lies on a causal curve starting in $\I^+$ and ending in $\I^-$.}
\begin{equation}\label{phicond}
\bksliceset_{\rm c} := \bigg\{\phi:\bar\M\to\left[-\frac12,\frac12\right]
\quad \left| \quad \phi|_{\I^\pm}=\pm\frac12\,,\,d\phi\in \fdc\right\}.
\end{equation}
Given $\phi\in\bksliceset_{\rm c}$, for any $t\in(-1/2,1/2)$, the level set $\bkslice_t$, defined as the future boundary of the region where $\phi\le t$, is a slice. A single slice $\bkslice$ is represented by the step function with $\phi=\pm1/2$ on $I^\pm(\bkslice)$. (If $\bkslice$ touches $\I^\pm$, then imagine deforming it infinitesimally away from those boundaries.) Similarly, a convex combination of time-sheets is represented by a function $\psi$:
\begin{equation}\label{psicond}
\tilde\tsset_{\rm c}:=\bigg\{\psi:\bar\M\to\left[-\frac12,\frac12\right]\quad \left| \quad\psi|_{D(A)}=-\frac12\,,\,
\psi|_{D(B)}=\frac12\,,\,
d\psi\text{ spacelike or $0$}
\right\}.
\end{equation}
Given $\psi\in\tilde\tsset_{\rm c}$, for any $s\in(-1/2,1/2)$, the level set $\ts_s$ of $\psi$ is a time-sheet in $\tsset$. A single time-sheet $\ts\in\tsset$ would be represented by the step function with $\psi=-1/2$ on the spacetime homology region between $D(A)$ and $\ts$, and $\psi=1/2$ on the complement. Given $\phi\in\bksliceset_{\rm c}$, $\psi\in\tilde\tsset_{\rm c}$, the area of the intersection of the respective level sets $\bkslice_t$, $\ts_s$, averaged over $t$ and $s$, can, by a generalization of the coarea formula, be expressed as an integral over $\M$:
\begin{equation}\label{fdef1}
f[\phi,\psi]:=\int_{-1/2}^{1/2}dt\int_{-1/2}^{1/2}ds\,\area(\bkslice_t\cap\ts_s) = \int_{\M}\sqrt g\,|d\phi\wedge d\psi|\qquad(\phi\in\bksliceset_{\rm c},\psi\in\tilde\tsset_{\rm c})\,.
\end{equation}

So far, so good. However, our job of convex relaxing is not finished: while $\bksliceset_{\rm c}$ is a convex set, $\tilde\tsset_{\rm c}$ is not, since the spacelike or vanishing covectors do not form a convex subset of the cotangent space at a point on $\bar\M$. The convex hull of this subset is the entire cotangent space, leading us to simply drop the condition on $d\psi$, and define the relaxed time-sheet set as follows:\footnote{\, The restriction of $\psi$ to lie in the interval $[-1/2,1/2]$ in the definition of $\tsset_{\rm c}$ is actually optional. Dropping it simply leads to extra, superfluous level sets, and doesn't change the optimal value of any of the functionals we consider. See subsection \ref{sec:dualizations} for details.}
\begin{equation}\label{psicond2}
\tsset_{\rm c}:=\bigg\{\psi:\bar\M\to\left[-\frac12,\frac12\right]\quad \left| \quad\psi|_{D(A)}=-\frac12\,,\,
\psi|_{D(B)}=\frac12
\right\}.
\end{equation}
The level sets of $\psi\in\tsset_{\rm c}$ are still hypersurfaces homologous to $D(A)$, but they are no longer necessarily everywhere timelike. We must also correspondingly extend the definition of the objective functional $f$ to $\bksliceset_{\rm c}\times\tsset_{\rm c}$. Recall that we want the objective to be concave with respect to $\phi$ and convex with respect to $\psi$. Luckily, it is possible to extend the definition \eqref{fdef1} to this larger domain while satisfying this condition. In fact, as shown in lemma \ref{lem:wedgedotprops}, the extension is unique, and given by the wedgedot pairing:
\begin{equation}\label{fdef2}
f[\phi,\psi]
:= \int_{\M}\sqrt g\,\ip{d\phi}{d\psi}\qquad
(\phi\in\bksliceset_{\rm c},\psi\in\tsset_{\rm c})\,.
\end{equation}
where $\ip WX:=\max\{|W\wedge X|,|W\cdot X|\}$. (See lemma \ref{lem:wedgedotprops} for more properties of this function, which will play a starring role throughout this paper.) This amounts to adjusting the area functional in regions where the $\psi$ level set $\ts_s$ is not timelike:
\begin{equation}\label{fdef3}
f[\phi,\psi]=\int_{-1/2}^{1/2}dt\int_{-1/2}^{1/2}ds\,\area'(\bkslice_t,\ts_s)\,,
\end{equation}
where
\begin{equation}
\area'(\bkslice,\ts):=\int_{\bkslice\cap\ts}\sqrt h\times
\begin{cases}
1\,&\text{($\ts$ timelike or null at $\bkslice$)} \\
\coth\chi\,&\text{($\ts$ spacelike at $\bkslice$)}
\end{cases}\,,
\end{equation}
$h$ is the determinant of the metric on $\bkslice\cap\ts$, and $\chi$ is the dihedral boost angle between $\bkslice$ and $\ts$ (in other words, the dot product of their future-directed unit normals equals $-\cosh\chi$).

Since $f$ is concave-convex, we expect its maximin and minimax values to agree,\footnote{\, Rigorously proving the equality of the minimax and maximin in this setting is outside the scope of this paper. However, we can make a few very crude remarks in this direction. Minimax theorems valid for infinite-dimensional spaces, such as Sion's theorem \cite{pjm/1103040253}, typically require, in addition to the objective function being convex-concave (or some generalization thereof), a continuity assumption on the objective as well as a compactness assumption on the domain of at least one of the variables. Since the integrand $\ip{d\phi}{d\psi}$ in the definition of the functional $f$ is a continuous function of $d\phi$ and $d\psi$, we would expect $f$ to be continuous with suitable mathematically precise definitions of the function spaces $\bksliceset_{\rm c}$, $\tsset_{\rm c}$ and topologies thereon. Furthermore, we would expect those function spaces to be compact, since the functions' domain $\bar\M$ and range $[-1/2,1/2]$ are both compact.} and we call this value $\Sc$ (or $\Sc(A:B)$ if we need to make explicit the dependence on the boundary regions):
\begin{equation}\label{Sccdef}
\Sc:=
\sup_{\phi\in\bksliceset_{\rm c}}\ \infp_{\psi\in\tsset_{\rm c}}\ f[\phi,\psi]
=\infp_{\psi\in\tsset_{\rm c}}\ \sup_{\phi\in\bksliceset_{\rm c}}\ f[\phi,\psi]\,.
\end{equation}
Furthermore, this value must sit between the non-convex maximin and minimax values, the analogue of \eqref{mixed min-max} in the game-theory setting:
\begin{equation}\label{min-max between}
\Sminus\le \Sc\le \Splus\,.
\end{equation}
We can prove \eqref{min-max between} by the same method. For the first inequality, choose a slice $\hat\bkslice$, and let $\hat\phi$ be the corresponding step function, $\hat\phi:=\pm1/2$ on $I^\pm(\hat\bkslice)$. Then, for any $\psi\in\tsset_{\rm c}$,
\begin{equation}
f[\hat\phi,\psi]=
\int_{-1/2}^{1/2} ds\,\area'(\hat\bkslice,\ts_s)\ge
\int_{-1/2}^{1/2} ds\,\area(\hat\bkslice\cap\ts_s)\ge
\inf_{\surf\in\Gamma_{\hat\bkslice}}\area(\surf)\,.
\end{equation}
In fact, these inequalities are tight, since, starting from the minimal surface $\hat\surf$ on $\hat\bkslice$, by lemma \ref{lem:surfsheet}, one can construct a time-sheet $\hat\ts\in\tsset$ such that $\hat\bkslice\cap\hat\ts=\hat\surf$, and from there set $\psi$ equal to the corresponding step function. Hence
\begin{equation}
\inf_{\psi\in\tsset_{\rm c}} f[\hat\phi,\psi]=
\inf_{\surf\in\Gamma_{\hat\sigma}}\area(\surf)\,,
\end{equation}
which, by the definition of $\Sc$, implies
\begin{equation}
\Sc\ge
\inf_{\surf\in\Gamma_{\hat\sigma}}\area(\surf)\,.
\end{equation}
Maximizing the right-hand side over $\hat\sigma\in\bksliceset$ gives the first inequality in \eqref{min-max between}. The second inequality is proven the same way, using the second formula for $\Sc$ in \eqref{Sccdef}.

For a spacetime with a gap between $\Sminus$ and $\Splus$, such as the one shown in figure \ref{fig:gap}, $\Sc$ must obviously differ from at least one of them. So we see that, in this case, convex relaxation is not an innocent operation: smearing the slices and time-sheets can change the optimal value of the objective. This is in contrast to the simpler Riemannian min cut and Lorentzian max cut programs, where the convex relaxation does not change the optimal value of the objective. In the next section, we will see in detail how this works for the spacetime of figure \ref{fig:gap}. But first we will rewrite the minimax and maximin formulas for $\Sc$ as pure minimization and maximization programs, respectively.

%---------------------------------------------------

\section{Flows}
\label{sec:flows}

In this section, we will derive, by Lagrange dualization starting from the convex-concave formula \eqref{Sccdef}, two new formulas for the quantity $\Sc$. The first, in subsection \ref{sec:Vflows}, involves dualizing on $\psi$ for fixed $\phi$ to obtain what we call the \emph{V-flow} program, while the second, in subsection \ref{sec:Uflows}, involves dualizing on $\phi$ for fixed $\psi$ to obtain the \emph{U-flow} program. In subsection \ref{sec:subadditivity}, we consider the case where the boundary has been divided into more than two regions, and show that $\Sc$ obeys the subadditivity inequality. To avoid interrupting the narrative, most of the derivations and proofs are relegated to subsection \ref{sec:flowproofs}.

As in the previous section, we assume the spacetime setup laid out in subsection \ref{sec:setup}. Except in subsection \ref{sec:subadditivity}, we fix regions $A$, $B$ covering a boundary slice, so we continue to suppress the dependence of them, e.g.\ writing $\Sc$ rather than $\Sc(A:B)$.

\subsection{V-flows}
\label{sec:Vflows}

In the first formula for $\Sc$ in \eqref{Sccdef}, the minimization over $\psi$, for fixed $\phi$, defines a convex program. It is therefore natural to apply Lagrange duality to this program to obtain a concave (maximization) program. This dualization, which is carried out in subsection \ref{sec:dualizations}, is similar to the dualization in the Riemannian setting of the relaxed min cut program to obtain the max flow program (see \cite{Headrick:2017ucz}). However, this case is slightly more complicated because the integrand $\ip{d\phi}{d\psi}$ in the objective functional \eqref{fdef2} is more complicated than the relaxed min cut objective, which is simply $|d\psi|$. Nonetheless, the main features of the Riemannian max flow program remain: we are maximizing the flux of a divergenceless 1-form (which we call $V$) subject to a norm bound and a no-flux boundary condition on the relative-homology boundary (in this case $\I$). However, the norm bound on $V$, which in the Riemannian case is simply $|v|\le1$, is more complicated here, and involves the 1-form $d\phi$. All in all, we have\footnote{\, It is possible to dualize the program \eqref{maxVflow}--\eqref{Vnormbound1} on $\phi$, yielding a maximin formula in terms of two 1-forms $U$ and $V$. However, this formula does not tell us anything new. See subsection \ref{sec:dualizations} for details.}
\begin{equation}\label{maxVflow}
\Sc = \sup_{V\in\F}\int_{D(A)}*V\,,
\end{equation}
where $\F$ is the set of V-flows, and a V-flow is a 1-form $V$ obeying
\begin{equation}\label{Vflowdef1}
    d{*V}=0\,,\qquad
    *V|_{\I}=0
    \ee
    \be
    \label{Vnormbound1}
\exists\,\phi\in \bksliceset_{\rm c}
\text{ s.t. } d\phi\pm V\in \fdc   \,.
\ee
We could have defined a V-flow as a pair $(V,\phi)$ obeying the constraints \eqref{Vflowdef1}, \eqref{Vnormbound1}. However, we have chosen to treat $\phi$ instead as an auxiliary variable whose job it is to enforce the norm bound. Indeed, as we will show in subsection \ref{sec:Vflowbounds}, it is possible to rewrite \eqref{Vnormbound1} without reference to $\phi$.

To understand the norm bound \eqref{Vnormbound1}, it is important to understand the pointwise condition ${d\phi\pm V\in \fdc}$. Subsection \ref{sec:pointwise} (see around figure \ref{fig:WXfdc}) discussed in detail properties of the set of covector pairs $(W,X)$ obeying the condition $W\pm X\in\fdc$; we repeat the most important points here in terms of $d\phi$ and $V$. First, this condition implies that $d\phi\in\fdc$ (as we already have from the definition \eqref{phicond} of $\Sc$).
 For timelike $d\phi$, this condition is equivalent to $|V_\perp|+|V_\parallel|\le|d\phi|$, where $V_\perp,V_\parallel$ are the projections of $V$ orthogonal and parallel respectively to $d\phi$; or, to say it another way, $V$ is required to be in the ``causal diamond'' in the cotangent space with vertices $\pm d\phi$. Note that this constraint allows $V$ to be timelike, spacelike, or null. For $d\phi$ null, the diamond degenerates to a null line segment, and $V$ must be a convex combination of $\pm d\phi$. In general, the larger $d\phi$ is at a given point, the larger $V$ may be, and therefore the more flux can pass through the point. However, given the constraints $d\phi\in \fdc$, $\phi|_{\I^\pm}=\pm1/2$ in the definition of $\bksliceset_{\rm c}$, $d\phi$ cannot be arbitrarily large everywhere in $\M$, which is ultimately what limits the total flux of $V$ that can pass through the spacetime from $D(A)$ to $D(B)$.

The symmetry of $\Sc$, $\Sc(B:A)=\Sc(A:B)$, can be understood easily in the language of V-flows. First note that the constraints \eqref{Vflowdef1}, \eqref{Vnormbound1} defining the V-flows do not refer specifically to the region $A$; the set $\F$ of V-flows is thus the same for the complementary region $B$ (or, as we will discuss in subsection \ref{sec:subadditivity}, any other boundary region for the same regulated spacetime). Second, those constraints are invariant under $V\to -V$, so if $V\in\F$ then $-V\in\F$. Finally, the constraints \eqref{Vflowdef1} guarantee that the objectives match: $\int_{D(A)}*V=\int_{D(B)}*(-V)$. Hence, for any max V-flow $V$ for $A$, $-V$ is a max V-flow for $B$ with the same value of the objective, and vice versa, so the two programs must have equal maxima.

\subsubsection{Slice-flows}
\label{sec:slice-flows}

Recall from \eqref{sigmaflowdef} that, given an (everywhere spacelike) slice $\bkslice$, a $\bkslice$-flow is a Riemannian flow with respect to the induced metric on $\bkslice$. A particular kind of V-flow is a \emph{slice-flow}, obtained from a $\bkslice$-flow $v$ on some slice $\bkslice$ by setting
\be\label{sliceflow}
\phi|_{I^\pm(\sigma)}=\pm\frac12\,,\qquad V=\delta(x^0)v_a dx^a\,,
\ee
where $\{x^0,x^a\}$ is a set of Gaussian normal coordinates about $\bkslice$. Conversely, any V-flow such that $\phi= \pm1/2$ on $I^\pm(\sigma)$ for some slice $\bkslice$ necessarily takes the form $V=\delta(x^0)v_a dx^a$, where $v$ is a $\bkslice$-flow. In fact, the full set $\F$ of V-flows can be obtained by convex-relaxing the set of slice-flows.\footnote{\, \label{ft:pointwise}More precisly, $\F$ is the \emph{pointwise} convex hull of the set of slice-flows. A slice-flow obeys, at each point in $\M$, the conditions $|V|\le|d\phi|$, $V\cdot d\phi=0$. The convex hull (in $T^*\times T^*$) of the set of covector pairs $(W,V)$ obeying $|V|\le|W|$, $W\cdot V=0$ is the set obeying $W\pm V\in \fdc$; see lemma \ref{lem:convexhull} in subsection \ref{sec:lemmaproofs}.} Since for any slice-flow $V$, $\int_{D(A)}*V=\int_{A_\bkslice}*v$, we have
\be\label{sliceflowS-}
\Sminus=\sup_{V\text{ slice-flow}}\int_{D(A)}*V\,.
\ee
As discussed at the end of section \ref{sec:relaxation}, since $\Sminus$ may be strictly less than $\Sc$, this is a case where convex relaxation can lead to a higher maximum.

\subsubsection{Worked example}

As an instructive example of the effect of convexity, consider the spacetime of figure \ref{fig:gap}. The maximum flux of any slice-flow is $\Sminus=a_-$. However, with a smeared slice we can do better. We will give a simple (although still not optimal) construction.

To be concrete, let the metric on the spacetime be
\be
ds^2 = -(dx^0)^2+(dx^1)^2+r(x^0,x^1)^2\, d\Omega_{D-2}^2\,,\qquad
-\frac T2\le x^0\le \frac T2\,,\qquad
-\frac L2\le x^1\le \frac L2\,,
\ee
with $T,L$ the dimensions of the rectangle, and let $T'$ be the time duration (at fixed $x^1$) of the stripe where the area of the sphere is $a_+$ (the purple stripe in figure \ref{fig:gap}), and $\beta$ is the speed at which its edges are moving. For the stripe to fit fully within the rectangle, we have
\be\label{Lbetaconstraint0}
T>\frac{L}{\beta}+T'\,.
\ee
Furthermore, the relation $x_B\in I^+(x_A)$ implies
\be\label{Lbetaconstraint}
L<\frac L\beta-T'\,.
\ee

\begin{figure}[tbp]
\centering
\includegraphics[width=0.35\textwidth]{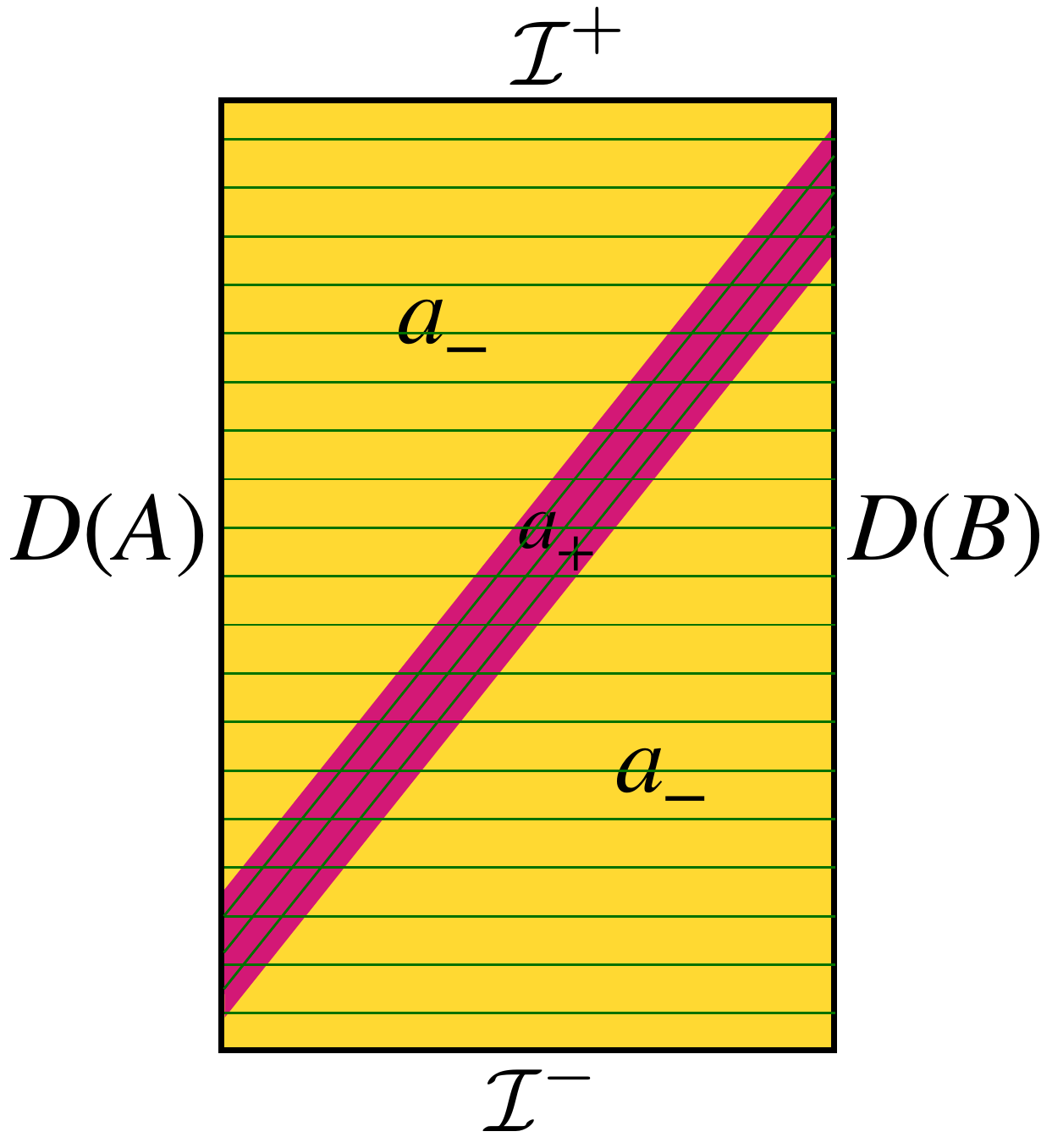}
\caption{\label{fig:gapV}
The V-flow $\hat V=\Vzero+\Vone$ on the spacetime of figure \ref{fig:gap}. The horizontal lines represent the component $\Vzero$ \eqref{V0def}, while the diagonal lines in the stripe represent the component $\Vone$ \eqref{V1def}.
}
\end{figure}

Convexity (or concavity) of an optimization problem guarantees that the optimal configuration can be assumed to share any symmetries of the problem. In what follows, we will use this fact to restrict our attention to spherically symmetric configurations.

The function $\hat\phi=x^0/T$ represents a slice that is extended in the $x^1$ and sphere directions and smeared uniformly in the $x^0$ direction. The following 1-form obeys the constraints \eqref{Vflowdef1}, \eqref{Vnormbound1} and has total flux $a_-$:
\be\label{V0def}
\Vzero =\begin{cases}
\frac1T \, dx^1\quad&(a_-\text{ region}) \\
\frac{a_-}{Ta_+} \, dx^1\quad&(a_+\text{ region})
\end{cases}\,.
\ee
Notice that this configuration saturates the constraint $d\phi\pm V\in \fdc$ in the $a_-$ region, but not in the $a_+$ region. This gives us the opportunity to send additional flux through the $a_+$ stripe. Specifically, the following ansatz, with $\alpha$ an arbitrary constant, represents a divergenceless 1-form that is uniform within the stripe:
\be\label{V1def}
\Vone=\begin{cases}
0\quad&(a_-\text{ region}) \\
\alpha\, (-dx^0+\beta dx^1)\quad&(a_+\text{ region})
\end{cases}\,.
\ee
$\Vone$ is divergenceless because the dual vector $\Vone^\mu$ is parallel to those edges. The flux of $\Vone$ is $\alpha\beta T'a_+$. In order to maximize the total flux of $\hat V=\Vzero+\Vone$, we should maximize $\alpha$, subject to the constraints $d\hat\phi\pm \hat V\in \fdc$. In fact, $d\hat\phi-\hat V\in \fdc$ for any positive $\alpha$, so the only constraint we have to worry about is $d\hat\phi+\hat V\in \fdc$, which requires
\be\label{alphaconstraint}
\frac{a_-}{Ta_+}+\alpha\beta\le\frac1T-\alpha\,.
\ee
The maximum value of $\alpha$ subject to \eqref{alphaconstraint} is
\be
\alpha =\left(1-\frac{a_-}{a_+}\right) \frac1T\frac1{\beta+1}\,.
\ee
All in all, the total flux for this configuration, illustrated in figure \ref{fig:gapV}, is
\be\label{flux}
\int_{D(A)}*\hat V =
a_-+\frac{T'}T\frac\beta{\beta+1}(a_+-a_-)\,.
\ee
In particular, the flux lines of $\hat{V}$ are horizontal outside the stripe but refract to have positive slope through the stripe, with the size of the effect depending on the parameters: they remain spacelike when $\frac{a_-}{a_+} > \frac{1-\beta}{2} $ but actually become timelike for $\frac{a_-}{a_+} < \frac{1-\beta}{2}$.

Further study shows that this simple V-flow can be improved in various ways (cf.\ appendix \ref{app:piecewiselin}) and in fact we don't know what the maximal flux is for this spacetime. The important point, however, is that convex relaxation has allowed us to increase the flux beyond the non-convex value $a_-$. Note also that, since $T'<T$, the flux of $\hat V$ is less than $a_+$, as must be the case since $\Sc\le \Splus=a_+$.

\subsubsection{Bounds}
\label{sec:Vflowbounds}

To gain more intuition about V-flows, we will now derive several inequalities involving $V$ from the constraints \eqref{Vflowdef1}, \eqref{Vnormbound1}. These inequalities have the advantage that they do not involve $\phi$; however, they are non-local.

\paragraph{World-line bound:} First, in subsection \ref{sec:lemmaproofs} (see lemma \ref{lem:Yineq1}), we show that, at any point in $\M$ and for any covector $T\in \fdc$, the condition $d\phi\pm V\in \fdc$ implies
\begin{equation}\label{Tbound}
\ip{T}{V}\le -T\cdot d\phi\,.
\end{equation}
Let $\qcv$ be an inextendible causal curve parametrized by $t$, and $\dot x$ the covector dual to the tangent vector $dx^\mu/dt$ of $\qcv$. (Recall that, with the tangent vector future-directed, the dual covector is past-directed.) Then, setting $T=-\dot x$ in \eqref{Tbound}, integrating over $t$, and using the fact that $\phi|_{\I^\pm}=\pm1/2$ so $\int_\qcv  dt(\dot x\cdot d\phi)=1$, we have
\begin{equation}\label{Vbound1}
\int_\qcv dt\,\ip{(-\dot x)}{V}\le1\,.
\end{equation}
(Note that, since the pairing $\ip{\ }{\ }$ is homogeneous in the first argument, the left-hand side of \eqref{Vbound1} is independent of the parametrization. Note also that, since the integrand is non-negative and any extendible causal curve is part of an inextendible one, the bound \eqref{Vbound1} applies also to extendible causal curves.)

If we specialize to the case where $\qcv$ is everywhere timelike and $t$ is a proper-time parameter (so $|\dot x|=1$), then, from \eqref{WXperp} and \eqref{ipdef},
\begin{equation}\label{Vperpbound0}
|V_\perp|\le\ip{(-\dot x)}{V}\,,
\end{equation}
(where $V_\perp$ is the projection of $V$ orthogonal to $\dot x$), so \eqref{Vbound1} implies
\begin{equation}\label{Vbound2}
\int_\qcv dt\,|V_\perp|\le1\,.
\end{equation}
If furthermore $V$ is everywhere spacelike or null, then \eqref{Vperpbound0} is saturated, so \eqref{Vbound2} is actually equivalent to \eqref{Vbound1}. The situation is slightly more complicated if $V$ is timelike in some region; then \eqref{Vbound1} is equivalent to \eqref{Vbound2} applied to an almost-null timelike curve $\qcv'$ that zig-zags close to $\qcv$.\footnote{\, This is essentially a consequence of the fact that $\ip{(-\dot x)}{V}$ is linear in $\dot x$ when both arguments are timelike, and approaches $|V_\perp|$ when $\dot x$ becomes null. More explicitly, consider a parameter interval $\delta t$ on $\qcv$, short enough that the metric, $V$, and $\dot x$ can be considered constant. Let $\qcv'$ connect the endpoints of this interval by way of two timelike segments with tangent covectors $\dot x'_\pm$, each of proper time $\delta t'/2$, so that $\delta x'_++\delta x'_-=\dot x\delta t$, where $\delta x'_\pm=\dot x'_\pm\delta t/2$. Then, on $\qcv'$, using \eqref{WXperp}, $|V_\perp|=|\dot x'_\pm\wedge V|$, so the contribution of these segments to the left-hand side of \eqref{Vbound2} is $|\delta x'_+\wedge V|+|\delta x'_-\wedge V|$. In the limit that the segments become null, again using \eqref{WXperp}, $|\delta x'_\pm\wedge V|$ goes to $|\delta x'_\pm\cdot V|$. Since $\delta x'_+\cdot V$ and $\delta x'_+\cdot V$ have the same sign, the sum of these terms is $|(\dot x\delta t)\cdot V|$. By \eqref{ipdef}, and using the fact that $V$ is timelike, this equals $\ip{(-\dot x)}{V}\,\delta t$, which is the contribution of the original segment of $\qcv$ to the left-hand side of \eqref{Vbound1}. Adding up such segments, we have $\lim_{\qcv'\to\text{null}}\int_{\qcv'}dt'|V_\perp|=\int_\qcv dt\ip{(-\dot x)}V$.
} 
Thus, \eqref{Vbound2}, when imposed on all timelike curves, is equivalent to \eqref{Vbound1} imposed on all timelike curves and therefore, by continuity, on all causal curves.

So far we have shown that \eqref{Vnormbound1} implies that \eqref{Vbound1} holds for all causal curves (or equivalently that \eqref{Vbound2} holds for all timelike curves). Remarkably, as shown in theorem \ref{thm:Vequivalence} (subsection \ref{sec:flowdefs}), the converse holds. This theorem is proved by constructing, given $V$, a function $\phi\in \bksliceset_{\rm c}$ such that $d\phi\pm V\in\fdc$. We can therefore choose to take
\eqref{Vbound1} (or \eqref{Vbound2}) 
as the defining norm bound for a V-flow, and dispense with the function $\phi$. In other words, we can define $\F$ as the set of 1-forms $V$ such that
\begin{equation}\label{Vflowdef2}
    d{*V}=0\,,\qquad
    *V|_{\I}=0
    \ee
\begin{equation}\label{Vnormbound2}
\forall \qcv\in \Qset\,,\,\int_\qcv dt\,\ip{(-\dot x)}{V}\le1\,,
\end{equation}
where $\Qset$ is the set of inextendible causal curves in $\M$. From this viewpoint, the function $\phi$ can be thought of as a certificate, or witness, for \eqref{Vbound1}, so that one does not have to check it for every causal curve.

\paragraph{Time-sheet bound:} The second bound is a bound on the flux of a V-flow through an arbitrary time-sheet $\ts$ (not necessarily homologous to $D(A)$, although that will be an important special case). As we show in subsection \ref{sec:lemmaproofs} (see lemma \ref{lem:TNL}), given an orthonormal pair of covectors $(T,N)$ with $T\in\fdt$, the constraint $d\phi\pm V\in \fdc$ implies\footnote{\, In lemma \ref{lem:TNL}, we also show the converse: if $(W,V)\in \fdc\times T^*$ obeys $N\cdot V\le-T\cdot W$ for any orthonormal pair $(T,N)$ with $T\in \fdt$, then $W\pm V\in \fdc$. Therefore, if we wish to use a scalar function and a local condition to guarantee \eqref{fluxbound0}, then $d\phi\pm V\in \fdc$ is the weakest condition we can impose.
}
\begin{equation}\label{Vperpbound}
     N\cdot V\le -T\cdot d\phi\,.
\end{equation}
The level sets of $\phi$ on $\ts$ can be written in terms of its level sets $\bkslice_t$ on $\bar\M$:
\begin{equation}
\surf_t:=\bkslice_t\cap\ts\,.
\end{equation}
We know from the fact that $d\phi\in \fdc$ that $\bkslice_t$, and hence $\surf_t$, is achronal. Assume temporarily that $\surf_t$ is spacelike (or empty) for all $t$; we will use continuity to address the null or partly null case below. Let $N$ be the unit normal covector field on $\ts$ and $T$ the unit future-directed covector tangent to $\ts$ and normal to $\surf_\ts$. The flux density of $V$ on $\ts$ is $N\cdot V$. The proper time between nearby surfaces $\surf_t$ and $\surf_{t+dt}$ is $dt/(-T\cdot d\phi)$, so the flux per unit spatial area on $\surf_t$ in the strip of $\ts$ between $\surf_t$ and $\surf_{t+dt}$ is $dt(N\cdot V)/(-T\cdot d\phi)$, which by \eqref{Vperpbound} is bounded above by $dt$. Integrating over $\surf_t$ and over $t$, we find
\begin{equation}\label{fluxbound0}
    \int_\ts *V\le\int_{-1/2}^{1/2} dt\,\area(\surf_t)\,.
\end{equation}
By continuity, \eqref{fluxbound0} holds even when $\surf_t$, for some values of $t$, is null or partly null (but still achronal). The right-hand side of \eqref{fluxbound0} is bounded above by the maximum area in $\Gamma_\ts$, so we have
\begin{equation}\label{normbound1}
    \int_\ts *V\le \sup_{\surf\in\Gamma_\ts}\area(\surf)\,.
\end{equation}
Again by continuity, \eqref{normbound1} also holds when $\ts$ is null or partly null. The bound \eqref{normbound1} is the analogue of the bound $\int_\surf*v\le\area(\surf)$ on the flux of a Riemannian flow through a surface $\surf$. It implies for example that an observer who carries a window of area $a$ along his worldline will see a total flux of $V$ through the window over his lifetime that is bounded above by $a$.

When applied to a time-sheet $\ts\in\tsset$ (i.e.\ a time-sheet homologous to $D(A)$ relative to $\I$), \eqref{normbound1} bounds the flux of $V$ through $D(A)$:
\begin{equation}\label{normbound2}
    \int_{D(A)}*V=\int_\ts*V\le \sup_{\surf\in\Gamma_\ts}\area(\surf)\,,
\end{equation}
where in the first equality we used the constraints $d{*V}=0$, $*V|_{\I}=0$. Maximizing the left-hand side of \eqref{normbound2} over V-flows and minimizing the right-hand side over time-sheets yields the inequality $\Sc\le\Splus$, for which we sketched a different proof in subsection \ref{sec:convexrelaxation}.

\subsection{U-flows}
\label{sec:Uflows}

In the previous subsection we started from the first, maximin formula for $\Sc$ in \eqref{Sccdef} and dualized on $\psi$ for fixed $\phi$. We can also start from the second, minimax formula, and dualize on $\phi$ for fixed $\psi$. The dualization is carried out in subsection \ref{sec:dualizations}. The result is a \emph{minimization} convex program:\footnote{\, It is possible to dualize \eqref{Uflowprog} on $\psi$, to obtain a minimax formula in terms of 1-forms $U$, $V$. However, this formula does not tell us anything new. See subsection \ref{sec:dualizations} for details.}
\begin{equation}\label{Uflowprog}
\Sc = \inf_{U\in\G}\int_{\I^+}*U\,,
\end{equation}
where $\G$ is the set of U-flows, and a U-flow is a 1-form $U$ on $\bar\M$ satisfying
\begin{equation}\label{Uflowdef1}
    d{*U}=0\,,\qquad
    *U|_{\I^0\cup\N}=0\,,
\ee
\be\label{Uflowdef2}
\exists\, \psi\in\tsset_{\rm c}\text{ s.t. }
    U\pm d\psi\in \fdc
\end{equation}
Here the condition $U\pm d\psi\in \fdc$ appearing in the norm bound \eqref{Uflowdef2} imposes a \emph{lower} bound on the flux, requiring $U\in \fdc$ and $|U|\ge|d\psi|$; see subsection \ref{sec:pointwise} (especially figure \ref{fig:WXfdc}) with $U$ taking place of $W$ and $d\psi$ taking place of $X$.

The symmetry of $\Sc$ can easily be understood in the language of U-flows. Under exchange of $A$ and $B$, whereas a V-flow and its scalar witness were mapped as $(V,\phi)\to(-V,\phi)$, preserving the objective and constraints, the corresponding map for a U-flow is $(U,\psi)\to(U,-\psi)$.

The U-flow formula \eqref{Uflowprog} is related to the V-flow one \eqref{maxVflow} via the minimax formulas \eqref{Sccdef}. These two formulas are also directly related by Lagrange duality on both variables, exchanging $(V,\phi$) for $(U,\psi)$; the dualization can be found in subsection \ref{sec:dualizations}. The relationships among these formulas is summarized in the following diagram:
\be\label{diagram}
\minCDarrowwidth50pt
\begin{CD}
\boxed{\sup_{\phi\in\bksliceset_{\rm c}}\ \infp_{\psi\in\tsset_{\rm c}}\  \int_\M\sqrt g\,\ip{d\phi}{d\psi}} @Z\text{minimax}Z
\text{theorem}Z \boxed{\infp_{\psi\in\tsset_{\rm c}}\ \sup_{\phi\in\bksliceset_{\rm c}}\  \int_\M\sqrt g\,\ip{d\phi}{d\psi}} \\
@X{\scriptsize\begin{array}{r}\psi\leftrightarrow V\\\text{duality}\end{array}}XX @XX{\scriptsize\begin{array}{r}\phi\leftrightarrow U\\ \text{duality}\end{array}}X \\
\boxed{\sup_{V\in\F}\int_{D(A)}*V} @Z{(V,\phi)\leftrightarrow(U,\psi)}Z\text{duality}Z \boxed{\inf_{U\in\G}\int_{\I^+}*U} 
\end{CD}
\ee

\subsubsection{Time-sheet-flows}
\label{sec:time-sheet-flows}

Analogously to the slice-flows discussed in subsection \ref{sec:slice-flows}, we can define a \emph{time-sheet-flow} as a U-flow such that, for some time-sheet $\ts\in\tsset$, $\psi=-1/2$ on the spacetime homology region interpolating between $\ts$ and $D(A)$, and $\psi=1/2$ on the complement. The flux of a time-sheet-flow is bounded below by the area of any surface $\surf\in\Gamma_\ts$ on $\ts$ (see \eqref{Ufluxbound1}), and in particular by the maximal-area surface:
\be
\int_{\I^+}*U\ge\sup_{\surf\in\Gamma_\ts}\area(\surf)\,.
\ee
In fact, as a consequence of the strong duality between the $U$ and $\phi$ programs, this inequality is saturated. Given a time-sheet $\ts\in\tsset$ and the corresponding function $\phi$, we have, for any $\psi\in\bksliceset_c$,
\be
\int_\M\sqrt g\,\ip{d\phi}{d\psi}=
\int_\M\sqrt g\,|d\phi\wedge d\psi|=\int_{-1/2}^{1/2}dt\,\area(\ts\cap\bkslice_t)\le\sup_{\surf\in\Gamma_\ts}\area(\surf)\,,
\ee
where $\bkslice_t$ is the level set of $\phi$ with $\phi=t$ and in the first equality we used the fact that $\ts$ is everywhere timelike so its normal covector $d\psi$ is everywhere spacelike. Strong duality of the $U$ and $\phi$ programs says that the inf of the $U$ objective equals the sup of the $\phi$ objective (for fixed $\psi$), so
\be
\inf_U\int_{\I^+}*U=\sup_{\surf\in\Gamma_\ts}\area(\surf)\,,
\ee
where the inf is over time-sheet-flows with fixed time-sheet $\ts$. Minimizing over time-sheets, we have, analogously to \eqref{sliceflowS-},
\be\label{timesheetflowS+}
\Splus=\inf_{U\text{ time-sheet-flow}}\int_{\I^+}*U\,.
\ee
So in a spacetime like that of figure \ref{fig:gap} where $S_c<S_+$, the flux cannot attain $S_c$. As with the slice-flows, this is a case where the convex relaxation (going from time-sheet-flows to general U-flows) changes the optimal value. 

Just as slice-flows are directly related via \eqref{sliceflow} to Riemannian flows, time-sheet-flows are related to Lorentzian flows, although the relation is a bit more complicated. Recall that a Lorentzian flow, as defined in \cite{Headrick:2017ucz}, is a 1-form $u$ on a Lorentzian manifold-with-boundary obeying
\be
u\in\fdt\,,\qquad |u|\ge1\,,\qquad d*u=0\,.
\ee
(No-flux boundary conditions may also be imposed on part of the manifold's boundary.) A time-sheet-flow involves a piece that is delta-function localized on the time-sheet, $\delta(x^1)u_adx^a$, where $\{u^1,u^a\}$ is a set of Gaussian normal coordinates about $\ts$ and $u$ is a 1-form on $\ts$ obeying $u\in\fdt$, $|u|\ge1$ (as well as the no-flux boundary condition $*u|_{\ts\cap\I^0}=0$). However, the flux of $U$ may enter or leave the time-sheet, so $u$ is not necessarily divergenceless, nor does $U$ necessarily vanish off of $\ts$. In fact, if $\ts$ has spacelike seams at which timelike pieces meet, then the flux must continue past the seam, so $U$ cannot vanish outside $\ts$.

\subsubsection{Worked example}

\begin{figure}[tbp]
\centering
\includegraphics[width=0.35\textwidth]{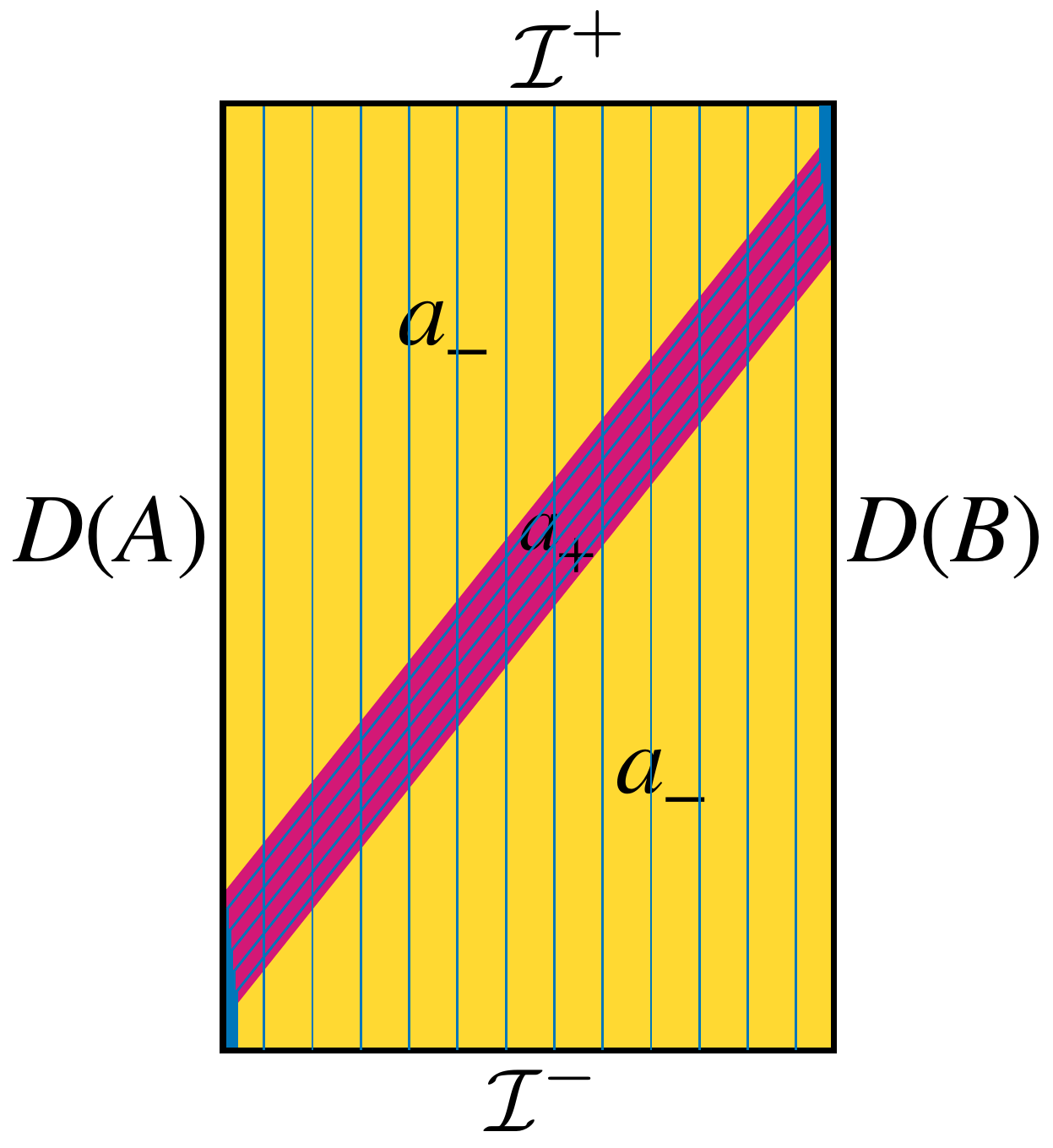}
\caption{\label{fig:gapU}
The U-flow $\hat U=\Uzero+\Uone$ on the spacetime of figure \ref{fig:gap}. The vertical lines represent the component $\Uzero$ \eqref{V0def}, while the diagonal lines in the stripe represent the component $\Uone$ \eqref{V1def}. The thick line in the upper right represents a delta-function carrying the flux from $\Uone$ along $D(B)$ toward $\I^+$; similarly for the thick line in the lower left.
}
\end{figure}

To see the effect of relaxation, we return again to the spacetime of figure \ref{fig:gap}. The minimum flux of any time-sheet-flow is $\Splus=a_+$. Instead of a single time-sheet, we can consider a smeared time-sheet, represented for example by the function $\hat\psi = x^1/L$. In the $a_-$ region, the smallest covector $U$ obeying $U\pm d\hat\psi\in \fdc$ is $dx^0/L$. This can be continued into the $a_+$ region to make a divergenceless 1-form:
\be
\label{U0def}
\Uzero=\begin{cases}
\frac1L \, dx^0\quad&(a_-\text{ region}) \\
\frac{a_-}{La_+} \, dx^0\quad&(a_+\text{ region})
\end{cases}\,.
\ee
$\Uzero$ has flux $a_-$ through $\I^+$. However, it does not obey $U\pm d\hat\psi\in \fdc$ in the $a_+$ region. We must add extra flux through the stripe. The minimal extra 1-form that we can add to satisfy this constraint is
\be
\Uone=\begin{cases}
0\quad&(a_-\text{ region}) \\
\left(1-\frac{a_-}{a_+}\right)\frac1L\frac1{1-\beta}\, (dx^0-\beta \, dx^1)\quad&(a_+\text{ region})
\end{cases}\,.
\ee
Since the dual vector field $\Uone^{\mu}$ is parallel to the edges of the stripe, $\Uone$ is divergenceless. However, it does not obey the no-flux boundary condition on $\N$. We therefore also have to add a delta-function 1-form in the $x^0$-direction on $D(A)$ from $\I^-$ to the top edge of the stripe, and on $D(B)$ from the bottom edge of the stripe to $\I^+$. The total flux of $\hat U=\Uzero+\Uone$ is
\be\label{Uhatflux}
\int_{\I^+}*\hat U = a_-+\frac{T'}L\frac{\beta}{1-\beta}\, (a_+-a_-)\,.
\ee
By virtue of \eqref{Lbetaconstraint}, this is smaller than $a_+$. On the other hand, it is larger than the flux \eqref{flux} for the V-flow of the previous subsection. In fact, this holds for any  pair $(V,\phi)$ and $(U,\psi)$ of feasible configurations, and is a consequence of the duality between the V-flow and U-flow programs. We can say more: as a consequence of the vertical dualities in the diagram \eqref{diagram}, the functional appearing in the minimax formulas \eqref{Sccdef} always sits between these two values:
\be\label{doubleineq}
\int_{D(A)}*V\le\int_\M\sqrt g\,\ip{d\phi}{d\psi}\le\int_{\I^+}*U\,.
\ee
For the functions $\hat\phi$, $\hat\psi$ we used in the previous subsection and this one, we have
\be
\int_\M\sqrt g\,\ip{d\hat\phi}{d\hat\psi} = a_-+\frac{T'}T\, (a_+-a_-)\,,
\ee
which indeed (using \eqref{Lbetaconstraint0}) obeys \eqref{doubleineq}.

This configuration $(\hat U,\hat\psi)$, like $(\hat V,\hat\phi)$, is not optimal. Indeed, we do not know what the optimal configuration is, nor the value of $\Sc$ for this spacetime.\footnote{\, 
One natural generalization of the scalar fields $\phi$ and $\psi$ being linear is for them to be piecewise linear, with different constant gradient inside and outside the stripe.  Setting the outside parts to be the same as before and optimizing over the gradients inside, we can see that the extremized fluxes of $V$ and $U$ approach each other further (but not enough to coincide), while maintaining the requisite nesting.  The detailed results are presented in appendix \ref{app:piecewiselin}.
} We only know that it lies somewhere between \eqref{flux} and \eqref{Uhatflux}.

\subsubsection{Bounds}
\label{sec:Uflowbounds}

U-flows obey non-local bounds analogous to those obeyed by V-flows. The condition $U\pm d\psi\in \fdc$ implies that $U\in\fdc$ and also that, for any covector $Y$,
\begin{equation}\label{UT bound}
\ip{U}{Y}\ge Y\cdot d\psi
\end{equation}
(see lemma \ref{lem:Yineq2} in subsection \ref{sec:lemmaproofs}). Let $\pcv$ be a curve in $\M$ starting in $D(A)$ and ending in $D(B)$, parametrized by $s$, and $\dot x$ be the covector dual to its tangent vector $dx^\mu/ds$. The bound \eqref{UT bound}, with $Y=\dot x$, together with the boundary conditions on $\psi$, imply
\begin{equation}\label{Ubound1}
\int_\pcv  ds\,\ip{U}{\dot x} \ge1\,.
\end{equation}
(Note that, since the pairing $\ip{\ }{\ }$ is homogeneous in the second argument, the left-hand side of \eqref{Ubound1} is independent of the parametrization of $\pcv$.) For spacelike $\pcv$, with $s$ the proper-distance parameter, \eqref{Ubound1} is equivalent to
\begin{equation}
\int_\pcv  ds\,|U_\perp| \ge1\,,
\end{equation}
where $U_\perp$ is the projection of $U$ perpendicular to $\dot x$.

According to theorem \ref{thm:Uequivalence} (subsection \ref{sec:flowdefs}), the converse is also true: if $U\in\fdc$ everywhere and \eqref{Ubound1} holds for all curves $\pcv$ from $D(A)$ to $D(B)$, then \eqref{Uflowdef2} holds. Thus we can take \eqref{Ubound1} as the defining norm bound for $U$ and dispense with the function $\psi$, defining a U-flow for $A$ as a 1-form $U$ obeying
\begin{equation}\label{Uflowdef3}
    d{*U}=0\,,\qquad
    *U|_{\I^0\cup\N}=0\,,
\ee
\be\label{Uflowdef4}
U\in\fdc\,,\qquad
\forall\,\pcv\in \Pset\,,\,
\int_\pcv  ds\,\ip{U}{\dot x} \ge1\,,
\end{equation}
where $\Pset$ is the set of all curves from $D(A)$ to $D(B)$. From this viewpoint, $\psi$ can be thought of as a witness for \eqref{Ubound1}.

We also have an analogue of the bound \eqref{normbound2} for the flux of $U$, proved similarly. From lemma \ref{lem:TNL}, we know that, for any orthonormal pair of covectors $(T,N)$ with $T\in\fdt$, the constraint $U\pm d\psi\in\fdc$ implies
\be
-T\cdot U\ge N\cdot d\psi\,.
\ee
Therefore, for any everywhere-spacelike slice $\bkslice$,
\begin{equation}\label{Ufluxbound1}
\int_{\I^+}*U = \int_{\bkslice}*U = \int_\bkslice \sqrt h\,(-T\cdot U)\ge
\int_\bkslice \sqrt h\,N\cdot d\psi=
\int_\bkslice \sqrt h\,|d_\bkslice\psi|=
\int_{-1/2}^{1/2}ds\,\area(\gamma_s)\,,
\end{equation}
where $d_\bkslice$ is the gradient on $\bkslice$; in going from $N\cdot d\psi$ to $|d_\bkslice\psi|$ we chose $N$ to be the normalized projection of $d\psi$ onto $\bkslice$; and $\gamma_s:=\ts_s\cap\bkslice$ is the level set of $\psi$ on $\bkslice$. Hence
\be\label{Ufluxbound}
\int_{\I^+}*U \ge
 \inf_{\surf\in\Gamma_\bkslice}\area(\surf)\,.
 \ee
By continuity, \eqref{Ufluxbound} holds for general slices. Minimizing the left-hand side over U-flows and maximizing the right-hand side over slices yields the inequality $\Sc\ge \Sminus$.

\subsection{Multiple regions \& subadditivity}
\label{sec:subadditivity}

So far, we have considered just two boundary regions $A,B$. But it is also interesting to consider more regions $A,B,C,\ldots$, allowing us to explore increasingly fine relational properties of the entanglement structure. As explained in subsection \ref{sec:EWCS}, any entangling surfaces have been excised in passing to the regulated spacetime $\M$, so we have $D(AB)=D(A)\cup D(B)$, etc., and the conformal boundary of $\M$ is $\N=D(A)\cup D(B)\cup D(C)\cup\cdots$. The quantity $\Sc(\A:\A^c)$ is then defined for any subset $\A$ of the regions.

It is important to understand where the choice of region enters in the max V-flow and min U-flow programs. The constraints \eqref{Vflowdef1} and \eqref{Vnormbound1} (or equivalently \eqref{Vnormbound2}) that define a V-flow do not depend on the choice of region, which instead enters in the objective, which is the flux of $V$ through $D(\A)$. The situation is the opposite for the U-flows: the objective \eqref{Uflowprog}, the flux of $U$ through the future boundary, is independent of the choice of region, whereas the constraints depend on it; more precisely, the conditions \eqref{Uflowdef1} do not depend on it, but the norm bound \eqref{Uflowdef2} does, through the boundary conditions on $\psi$: $\psi|_{D(\A)}=-1/2$, $\psi|_{D(\A^c)}=1/2$; the equivalent norm bound \eqref{Uflowdef4} also does through the set $\Pset$ of curves from $D(\A)$ to $D(\A^c)$.\footnote{\, It's important to note that this discussion of the dependence of the set of V- and U-flows on the choice of boundary region is in the context of a \emph{fixed} regulated spacetime, of which the regions being considered are unions of entire boundary connected components. As we will see when we return to the original unregulated spacetime in section \ref{sec:embedding}, the set of V-flows there \emph{does} depend on the choice of boundary region.}

Let us now specialize to the case of three boundary regions $A,B,C$. We then have three different $\Sc$ quantities: $\Sc(A:BC)$, $\Sc(B:AC)$, and $\Sc(C:AB)$, which we will write as $\Sc(A)$, $\Sc(B)$, and $\Sc(AB)$. We will show that they obey the subadditivity, or triangle, inequality:
\be\label{subadditivity}
\Sc(AB)\le \Sc(A)+\Sc(B)\,.
\ee
Interestingly, it is possible to prove \eqref{subadditivity} using either V-flows or U-flows. The proofs are quite different, illustrating the different points of view afforded by the two types of flows.

\begin{proof}[Proof via V-flows:]
Let $V$ be a maximal V-flow for $AB$. We then have
\be
\Sc(AB)=\int_{D(AB)}*V=\int_{D(A)}*V+\int_{D(B)}*V \le \Sc(A)+\Sc(B)\,,
\ee
where in the inequality we used the fact that $V$ is an allowed (though not necessarily maximal) V-flow for $A$ and for $B$, since the constraints \eqref{Vflowdef1}, \eqref{Vnormbound1} (or \eqref{Vnormbound2}) in the definition of a V-flow do not depend on the choice of boundary region. (This is is essentially the same proof as in the Riemannian setting \cite{Freedman:2016zud}.) \end{proof}

\begin{proof}[Proof via U-flows:] For the proof of \eqref{subadditivity} using U-flows, on the other hand, we have to contend with the fact that the definition of a U-flow \emph{does} depend on the choice of boundary region. Specifically, the boundary conditions for the function $\psi$ appearing in the norm bound \eqref{Uflowdef2} depend on the boundary region; equivalently, the curves $\pcv$ enforcing the norm bound in the form \eqref{Uflowdef4} connect the given boundary region to its complement. Let $U_A$ be a minimal U-flow for $A$ and $U_B$ a minimal U-flow for $B$. Then we claim that $U_{AB}:=U_A+U_B$ is an allowed (though not necessarily minimal) U-flow for $AB$. First, it clearly obeys the constraints \eqref{Uflowdef1}, which are linear. We can show that it also obeys the norm bound \eqref{Uflowdef2} or equivalently \eqref{Uflowdef4}.

For \eqref{Uflowdef2}, let $\psi_A$, $\psi_B$ be witness functions for $U_A$, $U_B$ respectively. These obey
\be
\left.\psi_A\right|_{D(A)}=-\frac12\,,\quad
\left.\psi_A\right|_{D(BC)}=\frac12\,,\qquad
\left.\psi_B\right|_{D(B)}=-\frac12\,,\quad
\left.\psi_B\right|_{D(AC)}=\frac12\,.
\ee
We then define
\be
\psi_{AB}:=\psi_A+\psi_B-\frac12\,.
\ee
This obeys
\be
\left.\psi_{AB}\right|_{D(AB)}=-\frac12\,,\quad
\psi_{AB}|_{D(C)}=\frac12\,.
\ee
Also, since $U_A\pm d\psi_A\in \fdc$ and $U_B\pm d\psi_B\in \fdc$, we have $U_{AB}\pm d\psi_{AB}\in \fdc$, so $U_{AB}$ obeys \eqref{Uflowdef2}.

To show that $U_{AB}$ obeys \eqref{Uflowdef4}, first note that clearly $U_{AB}\in\fdc$. Let $\pcv$  be a curve connecting $D(A)$ and $D(C)$. We then have
\be
\int_\pcv ds\,\ip{U_{AB}}{\dot x}\ge 
\int_\pcv ds\left(\ip{U_{A}}{\dot x}+\ip{U_{B}}{\dot x} \right) = \int_\pcv  ds\,\ip{U_A}{\dot x}+\int_\pcv ds\,\ip{U_B}{\dot x}\ge 1\,,
\ee
where in the first inequality we used the fact that $\ip{\ }{\ }$ is homogeneous and concave, hence superadditive, in the first argument; and in the second inequality we used the fact that $\pcv$  connects $D(A)$ to its complement $D(BC)$ (since $D(C)$ is contained in $D(BC)$) and $U_A$ is a U-flow for $A$, so the first integral is at least 1, and the second integral is clearly at least 0. Similarly for a curve connecting $D(B)$ and $D(C)$. Hence $U_{AB}$ obeys \eqref{Uflowdef4}.

We now have
\be
\Sc(AB)\le\int_{\I^+}*U_{AB} = \int_{\I^+}*U_A+\int_{\I^+}*U_B = \Sc(A)+\Sc(B)\,.
\ee
\end{proof}

The fact that the quantity $\Sc$ is subadditive immediately raises several other questions. First, are the minimax and maximin quantities $S_\pm$ also subadditive? We are not aware of either a proof or a counterexample, but we suspect that they are not. It is also natural to ask whether $\Sc$ obeys higher entropy inequalities such as strong subadditivity and MMI. In the Riemannian setting, flow-based proofs of these inequalities use the nesting property \cite{Freedman:2016zud,Headrick:2017ucz} and the existence of a max multiflow \cite{Cui:2018dyq}, respectively. The proofs of these properties do not carry over straightforwardly to the current Lorentzian setting.\footnote{\label{LorentzianNesting}\, 
Those proofs use strong duality between the max flow and relaxed min cuts programs. On the cut side, they require a smeared cut, or barrier function $\psi$, for a (nested or disjoint) set of boundary regions to be decomposed into a corresponding set of barrier functions. Due to the non-local nature of the constraints, we have not found a way to effect such a decomposition in the Lorentzian setting.
} We do not know whether these properties hold for U- and V-flows in a general spacetime, but we suspect that they do not. In subsection \ref{sec:multiple}, we will return to this issue in the more specific setting of standard holographic spacetimes.

\subsection{Proofs}
\label{sec:flowproofs}

In this subsection we provide proofs and derivations for many of the statements made in the rest of this section. At a technical level, this subsection is the heart of the paper. The reader who is mainly interested in the results of the paper is nonetheless permitted to skip ahead to section \ref{sec:threads}.

\subsubsection{Covector-pair lemmas}
\label{sec:lemmaproofs}

The following set of covector pairs, at a given point $x\in\M$, has played an important role throughout this section:
\be
\Kset:=\left\{(W,X)\in T^*\times T^*\quad \big| \quad W\pm X\in \fdc\right\}
\ee
Here we prove four lemmas providing alternative characterizations of $\Kset$ that have been used in various parts of this section. We also give two formulas for the pairing $\ip{\ }{\ }$ in terms of $\Kset$ that will be used in the dualizations of the next subsection.

\begin{lemma}\label{lem:convexhull}
$\Kset$ is the convex hull in $T^*\times T^*$ of 
\be
\Kset':=\left\{(W,X)\in \Kset\quad \big| \quad W\cdot X=0\right\}.
\ee
\end{lemma}

\begin{proof}
Suppose $(W,X)\in \Kset\setminus \Kset'$. Consider first the case $W\cdot X<0$; then $W+X\in \fdt$ (since we know $W+X\in\fdc$ and $(W+X)^2=(W-X)^2+4W\cdot X<0$). Choose a covector $Y$ such that $Y\cdot (W+X)=0$, $Y^2=-W\cdot X$. The reader can check that the two points $(W_\pm,X_\pm):=(W\pm Y,X\pm Y)$, whose average is $(W,X)$, are both elements of $\Kset'$. The case $W\cdot X>0$ is the same, but with $Y$ chosen such that  $Y\cdot (W-X)=0$, $Y^2=W\cdot X$.
\end{proof}

Now define a set of orthonormal pairs of covectors, one of which is timelike:
\be
\Lset:=\left\{(T,N)\in \fdt\times T^*\quad \big| \quad T\cdot N=0, \ |T|=|N|=1\right\}.
\ee
We can rewrite the condition on $(W,X) \in \Kset$ as $W$ being future-directed causal with $X$ lying in its causal diamond, which we can conveniently express in terms of $T$ and $N$ as follows:
\begin{lemma}\label{lem:TNL}
$(W,X)\in \Kset$ if and only if
\be
\forall\,(T,N)\in \Lset\,,\ \,X\cdot N\le-W\cdot T\,.
\ee
\end{lemma}

\begin{proof}
Given covectors $(T,N)\in T^*\times T^*$, we can null-decompose them by defining ${U^\pm:=(T\pm N)/2}$. Then $(T,N)\in \Lset$ if and only if $U^\pm$ are future-directed null covectors and $U^+\cdot U^-=-1/2$. 
For $(W,X)\in T^*\times T^*$, we have
\begin{equation}\label{VNWT}
    X\cdot N+W\cdot T=(W+X)\cdot U^++(W-X)\cdot U^-\,.
\end{equation}
If $(W,X)\in \Kset$ then $(W\pm X)\cdot U^\pm\le0$, so $X\cdot N\le -W\cdot T$. 

For the converse, suppose $W+X\not\in \fdc$. Then there exists a future-directed null covector $U^+$ such that $(W+X)\cdot U^+>0$. Let $U^-$ be a future-directed null covector such that $U^+\cdot U^-=-1/2$. By rescaling $U^+$ by a positive factor $\alpha$ and $U^-$ by $1/\alpha$, \eqref{VNWT} can be made positive. Similarly for $W-X\not\in \fdc$.
\end{proof}
In fact, using the $\ip{\ }{\ }$ product, we can re-express the condition $(W,X)\in\Kset$ in terms of a single arbitrary covector $Y$, in two different ways.

\begin{lemma}\label{lem:Yineq1}
$(W,X)\in \Kset$ if and only if
\be\label{Ycond1}
\forall\,Y\in \fdc\,,\ \, \ip{Y}{X}\le-Y\cdot W\,.
\ee
\end{lemma}

\begin{proof}
If \eqref{Ycond1} holds, then for any $Y\in \fdc$,
\be
Y\cdot(W\pm X)\le 
Y\cdot W+|Y\cdot X|\le
Y\cdot W+\ip{Y}{X}\le0\,,
\ee
 implying $(W,X)\in \Kset$.

Conversely, if $(W,X)\in \Kset$, then $W\in \fdc$ and, for any $Y\in \fdc$, $|Y\cdot X|\le-Y\cdot W$. Furthermore, for any $Y\in \fdt$, $|Y\wedge X|=|Y||X_\perp|$, where $X_\perp$ is the projection of $X$ perpendicular to $Y$. 
The condition $(W,X)\in \Kset$ implies $|Y||W_\perp\pm X_\perp|\le|Y||W_\parallel\pm X_\parallel|=-Y\cdot(W\pm X)$; taking the average of these two inequalities and applying the triangle inequality in the (spacelike) hyperplane orthogonal to $Y$ yields $|Y||X_\perp|\le -Y\cdot W$. Hence $\ip YX\le-Y\cdot W$. By continuity, the inequality holds for any $Y\in \fdc$.
\end{proof}

\begin{lemma}\label{lem:Yineq2}
$(W,X)\in \Kset$ if and only if
\be\label{Ycond2}
W\in \fdc\text{ and }\forall\,Y\in T^*\,,\ \,X\cdot Y\le\ip{W}{Y}\,.
\ee
\end{lemma}

\begin{proof}
If \eqref{Ycond2} holds then, for any $Y\in \fdc$, $(W+X)\cdot Y= -\ip{W}{Y}+X\cdot Y\le0$, so $W+X\in \fdc$. Similarly, for any $Y$ such that $-Y\in \fdc$, $(W-X)\cdot Y=\ip{W}{Y}-X\cdot Y\ge 0$, so $W-X\in \fdc$.

Conversely, if $(W,X)\in \Kset$, then for any $Y$ such that $Y\in \fdc$ or $- Y\in \fdc$, $X\cdot Y\le\ip WY$. For spacelike $Y$, assume first that $W\in \fdt$. Decomposing $X$ and $Y$ into components perpendicular and parallel to $W$, we have $|X_\perp|+|X_\parallel|\le|W|$ and $|Y_\parallel|<|Y_\perp|$, so
\begin{equation}
X\cdot Y= X_\perp\cdot Y_\perp+X_\parallel\cdot Y_\parallel\le |X_\perp||Y_\perp|+|X_\parallel||Y_\parallel|
<
\left(|X_\perp|+|X_\parallel|\right) \, |Y_\perp|\le |W||Y_\perp|=\ip W Y\,,
\end{equation}
where in the first inequality we used the Cauchy-Schwarz inequality in the (spacelike) hyperplane orthogonal to $W$. By continuity the inequality holds for $W\in \fdc$.
\end{proof}

The preceding two lemmas give bounds on the $\ip{\ }{\ }$ product, which we now show are tight.
\begin{lemma}\label{lem:supWU}
Given $X,Y\in T^*$,
\be\label{saturYcond1}
\sup_{\substack{W:\\(W,X)\in \Kset}}Y\cdot W =
\begin{cases}
+\infty\,,\quad&Y\not\in \fdc\\
-\ip{Y}{X}\,,\quad&Y\in \fdc
\end{cases}\,.
\ee
\end{lemma}

\begin{proof}
If $Y\not\in \fdc$, then there exists a covector $Z\in \fdt$ such that $Y\cdot Z>0$. For sufficiently large $\alpha$, $(\alpha Z,X)\in \Kset$. So $W\cdot Y$ can be made arbitrarily large for $(W,X)\in \Kset$.

If $Y\in \fdc$, then by lemma \ref{lem:Yineq1}, for any $W$ such that $(W,X)\in \Kset$, $Y\cdot W\le -\ip{Y}{X}$. 
So the supremum on LHS of \eqref{saturYcond1} likewise satisfies this inequality.  To show that it can be actually saturated, 
the following constructs a $W$ such that $(W,X)\in \Kset$ and $W\cdot Y=-\ip{Y}{X}$:
\be
W = \begin{cases}
X\,\quad&X\in \fdc \\
-X\,\quad&-X\in \fdc \\
|X|\,\hat Y_\perp\,\quad&X^2>0
\end{cases}\,,
\ee
where $\hat Y_\perp:=Y_\perp/|Y_\perp|$ is the unit covector in the direction of the projection of $Y$ orthogonal to $X$.
\end{proof}

\begin{lemma}\label{lem:supVX}
Given $W\in \fdc$, $Y\in T^*$,
\be\label{VXsup}
\sup_{\substack{X:\\ (W,X)\in \Kset}}X\cdot Y\textbf{} = \ip{W}{Y}\,.
\ee
\end{lemma}

\begin{proof}
From lemma \ref{lem:Yineq2}, we have that $X\cdot Y\le\ip{W}{Y}$. 
For $\pm Y\in\fdc$, the bound is achieved by $X=\mp W$. For $Y$ spacelike and $W\in \fdt$, 
the bound is achieved by
\be
X=|W|\,\hat Y_\perp\,,
\ee
where $\hat Y_\perp$ is the unit covector in the direction of the projection of $Y$ orthogonal to $W$. By continuity, \eqref{VXsup} then holds for $W\in \fdc$.
\end{proof}

\subsubsection{Dualizations}
\label{sec:dualizations}

In this subsection we carry out five Lagrange dualizations:
\begin{itemize}
\item on $\psi$ in the first formula for $\Sc$ in \eqref{Sccdef}, with $\phi$ fixed, to obtain the V-flow program \eqref{maxVflow};
\item on $\phi$ in the second formula for $\Sc$ in \eqref{Sccdef}, with fixed $\psi$, to obtain the U-flow program \eqref{Uflowprog};
\item on $(V,\phi)$ in the V-flow program \eqref{maxVflow}, to obtain the U-flow program \eqref{Uflowprog};
\item on $\phi$ in the V-flow program \eqref{maxVflow}, with $V$ fixed, to obtain a maximin formula involving $U$ and $V$;
\item and finally on $\psi$ in the U-flow program \eqref{Uflowprog}, with $U$ fixed, to obtain a minimax formula involving $U$ and $V$.
\end{itemize}
We include the last two dualities for completeness. As we will explain, the resulting maximin and minimax formulas do not tell us anything new.

For each duality, we will check whether Slater's condition holds. Slater's condition, a sufficient condition for strong duality, is the existence of a (not necessarily optimal) point in the relative interior of the feasible set for the primal problem (the interior with respect to the lowest-dimensional affine space containing the feasible set). For a review and examples of Lagrange dualization as applied to this type of problem see \cite{Headrick:2017ucz}.

For simplicity, when writing the programs to be dualized, we will drop the condition $-1/2\le\psi\le1/2$ in the definition of $\tsset_{\rm c}$ and the conditions $-1/2\le\phi\le1/2$, $d\phi\in \fdc$ in the definition of $\bksliceset_{\rm c}$. As we will now show, it makes no difference whether we include those conditions or not. In the $\psi$ case, given a function $\psi$ that satisfies the boundary conditions $\psi|_{D(A)}=-1/2$, $\psi|_{D(B)}=1/2$ but falls outside the range $[-1/2,1/2]$ somewhere in $\bar\M$, we can define a new function $\psi'$ that is in the range $[-1/2,1/2]$, and therefore is in $\tsset_{\rm c}$, by
\begin{equation}
\psi'=\begin{cases}
-1/2\,,\quad&\psi<-1/2 \\
\psi\,,\quad&-1/2\le\psi\le1/2\\
1/2\,,\quad&\psi>1/2
\end{cases}\,.
\end{equation}
Furthermore, if $U\pm d\psi\in \fdc$, then $U\pm d\psi'\in \fdc$. So the set of feasible $U$s is the same, whether or not the condition $-1/2\le\psi\le1/2$ is imposed. For $\phi$, the conditions $d\phi\in \fdc$, $-1/2\le\phi\le1/2$, are already implied by the conditions $d\phi\pm V\in \fdc$, $\phi|_{\I^\pm}=\pm1/2$, so there is no need to include them explicitly.

\paragraph{From $\psi$ to $V$, with fixed $\phi$:}
\label{sec:phiV}
The starting program is 
\be
\label{Vflow2}
    \text{minimize}\int_{\M}\sqrt g\ \ip{d\phi}{d\psi}\text{ over }\psi\text{ such that }
    \psi|_{D(A)}=-\frac12\,,\quad
    \psi|_{D(B)}=\frac12
\ee
(where the function $\phi$ is regarded as fixed and obeys $d\phi\in \fdc$). We rewrite this by introducing a 1-form $X$:
\be
\label{Vflow3}
    \text{minimize}\int_{\M}\sqrt g\ \ip{d\phi}{X}\text{ over }\psi,X\text{ such that }
    d\psi=X\,,\quad
    \psi|_{D(A)}=-\frac12\,,\quad
    \psi|_{D(B)}=\frac12\,.
\ee
Slater's condition merely amounts in this case to the existence of a feasible configuration $\psi,X$.

We will impose the constraint $d\psi=X$ using a Lagrange multiplier 1-form $V$. The boundary conditions will be imposed implicitly. The Lagrangian is
\be
\begin{split}
L[\psi,X,V]&=
\int_\M\sqrt g\ \left[\ip{d\phi}{X}-V\cdot(X-d\psi)\right] \\
&=\int_\M\sqrt g\ \left[\ip{d\phi}{X}-V\cdot X\right]-\int_\M\psi\,d{*V}-\int_{\I}\psi\,(*V)+\frac12\left(\int_{D(A)}*V-\int_{D(B)}*V\right),
\end{split}
\ee
where, in the second line, after integrating by parts and using Stokes' theorem in the form \eqref{Stokes}, we imposed the boundary conditions on $\psi$.

 We now minimize $L$ with respect to $\psi$ and $X$, with the locus of 1-forms $V$ such that $L$ is bounded below defining the dual constraints. The two terms in the bulk integrand involving $X$ are ${\ip{d\phi}{X}-V\cdot X}$. This is homogeneous in $X$, so if it is negative for any $X$ then it is unbounded below, and otherwise its minimum is 0. By lemma \ref{lem:Yineq2}, this quantity is non-negative for all $X$ if and only if $d\phi\pm V\in \fdc$, which is therefore a dual constraint. The bulk term involving $\psi$, $-\psi\,d{*V}$, is bounded below if and only if $d{*V}=0$, i.e.\ $V$ is divergenceless. Similarly, the boundary term involving $\psi$, $\int_\I\psi\,(*V)$, is bounded below if and only if $*V|_\I=0$. We are left only with the two terms in the large parentheses, which, given the last two conditions on $V$, are equal. We are finally left with the following dual program:
\begin{equation}
    \text{maximize}\int_{D(A)}*V\text{ over }V\text{ such that }
    d\phi\pm V\in \fdc\,,\quad
    d{*V}=0\,,\quad
    *V|_{\I}=0\,.
   \end{equation}

\paragraph{From $\phi$ to $U$, with fixed $\psi$:}
\label{sec:psiW}

The starting program is 
\be
\label{Vflow4}
    \text{maximize}\int_{\M}\sqrt g\,\ip{d\phi}{d\psi}\text{ over }\phi\text{ such that }
d\phi\in \fdc\,,\quad    \phi|_{\I^\pm}=\pm\frac12
\ee
(where the function $\psi$ is regarded as fixed). We rewrite this by introducing a 1-form $W$:
\be
\label{Vflow5}
    \text{maximize}\int_{\M}\sqrt g\ \ip{W}{d\psi}\text{ over }\phi,W\text{ such that }
    d\phi=W\,,\quad W\in \fdc\,,\quad
        \phi|_{\I^\pm}=\pm\frac12\,.
\ee
Slater's condition amounts here to the existence of a function $\phi_0$ obeying $d\phi_0\in \fdt$ everywhere and $\phi_0|_{\I^\pm}=\pm1/2$.\footnote{\, \label{fn:clockfnconstr}
Such a function $\phi_0$ can be constructed as follows. Let $\Omega$ be a continuous function on $\bar\M$ that is positive on $\M$ and such that the metric $\tilde g:=\Omega^2 g$ on $\M$ has finite total spacetime volume. Define $\phi_\pm(x)$ as the spacetime volume of $J^\pm(x)$ with respect to $\tilde g$. Then $\phi_\pm\to0$ on $\I^\pm$ and $\mp d\phi_+\in \fdt$ on $\M$ (with respect to both $\tilde g$ and $g$). The function
\be
\phi_0:=\frac12\tanh\left(\frac1{\phi_+}-\frac1{\phi_-}\right)
\ee
then has the claimed properties.
}

We will impose the constraint $d\phi=W$ using a Lagrange multiplier 1-form $U$. The boundary conditions on $\phi$ and the constraint $W\in \fdc$ will be imposed implicitly. The Lagrangian is
\be\label{LphiWU}
\begin{split}
L[\phi,W,U]&=
\int_\M\sqrt g \ \left[\ip{W}{d\psi}+(W-d\phi)\cdot U\right] \\
&=\int_\M\sqrt g\ \left[\ip{W}{d\psi}+ W\cdot U\right]+\int_\M\phi\,d{*U}+\int_{\I^0\cup\N}\phi\,(*U)+\frac12\left(\int_{\I^+}*U-\int_{\I^-}*U\right),
\end{split}
\ee
where in the second line, after integrating by parts and using Stokes' theorem in the form \eqref{Stokes}, we imposed the boundary conditions on $\phi$.

We now maximize $L$ with respect to $\phi$ and $W$. The bulk integrand $\ip{W}{d\psi}+W\cdot U$, by lemma \ref{lem:Yineq1} along with homogeneity in $W$, is bounded above (as a function of $W$ on $\fdc$) if and only if $U\pm d\psi\in \fdc$, in which case the maximum is 0. The bulk integrand $\phi\,d{*U}$ is bounded above as a function of $\phi$ if and only if $d{*U}=0$, i.e.\ $U$ is divergenceless.  The boundary integrand $\phi\,(*U)$ on $\I^0\cup\N$ is bounded above if and only if $*U|_{\I^0\cup\N}=0$. The last two conditions imply that the two terms in the big parentheses in \eqref{LphiWU} are equal. We are finally left with the following program:
\begin{equation}
    \text{minimize}\int_{\I^+}*U\text{ over }U\text{ such that }
    U\pm d\psi\in \fdc\,,\qquad
    d{*U}=0\,,\qquad
    *U|_{\I^0\cup\N}=0\,.
   \end{equation}

\paragraph{From $(V,\phi)$ to $(U,\psi)$:}
\label{sec:VUduality}

Here we start with the V-flow program (with local norm bound in terms of the function $\phi$) and dualize it to obtain the U-flow program (with local norm bound in terms of the function $\psi$).

The V-flow program is:
\begin{multline}\label{Vflow6}
    \text{maximize}\int_{D(A)}*V\text{ over }(V,\phi)\text{ such that }\\
    d\phi\pm V\in \fdc\,,\qquad
    d{*V}=0\,,\qquad
    *V|_{\I}=0\,,\qquad
    \phi|_{\I^\pm}=\pm\frac{1}{2}\,.
   \end{multline}
We rewrite this by introducing 1-forms $W^\pm:=(d\phi\pm V)/2$ and eliminating $V$:
\begin{align}
    \text{maximize}\int_{D(A)}*(W^+-W^-)\text{ over }(W^+,W^-,\phi)\text{ such that }\\
    W^++W^-&=d\phi\,,\label{Wdphi}\\
    d*(W^+-W^-)&=0\,,\label{divfree}\\
    W^\pm\in \fdc\,,\qquad
    *(W^+-W^-)|_{\I}=0\,,\qquad
    \phi|_{\I^\pm}&=\pm\frac{1}{2}\,.\label{implicit}
   \end{align}
To show that Slater's condition holds, we set $W^+=W^-=d\phi/2$ and $\phi=\phi_0$, where $\phi_0$ is the function on $\bar\M$ obeying $d\phi_0\in \fdt$ everywhere in $\M$ and $\phi_0|_{\I^\pm}=\pm1/2$ (such as the one indicated in footnote \ref{fn:clockfnconstr}).

We will treat \eqref{Wdphi}, \eqref{divfree} as explicit constraints, introducing Lagrange multipliers $U$ (a 1-form) and $\psi$ (a scalar) respectively to enforce them. The constraints \eqref{implicit} will be treated implicitly. The Lagrangian functional is
\be
\begin{split}
L[W^+,W^-,\phi,U,\psi]&=    \int_{D(A)}*(W^+-W^-)+\int_\M\left[\sqrt{g}\,U\cdot(W^++W^--d\phi)-\psi\, d*( W^+-W^-)\right] \\
&= \int_{D(A)}\left[(1+\psi)*\!(W^+-W^-)+(*U)\, \phi\right]
        +\int_{D(B)}\left[(*U)\, \phi+\psi*\!(W^+-W^-)\right]\\
&\qquad\qquad        {}+\frac{1}{2}\int_{\I^+}*U-\frac{1}{2}\int_{\I^-}*U
        +\int_{\I^0}(*U)\, \phi        \\
&\qquad\qquad                {}+\int_\M\left[\sqrt{g}\,\left(
        (U+d\psi)\cdot W^++(U-d\psi)\cdot W^-\right)+\phi\, (d*U)\right],
\end{split}
\ee
where in the second line we integrated by parts and used Stokes' theorem in the form \eqref{Stokes} and the boundary conditions on $W^\pm$ and $\phi$ in \eqref{implicit}.

Following the usual procedure, our task now is to (1) find the constraints on $(U,\psi)$ that are necessary and sufficient for $L$ to be bounded above as a functional of $(W^+,W^-,\phi)$; and (2) assuming those constraints are satisfied, maximize $L$ with respect to $W^\pm$ and $\phi$. The results of (1) and (2) are the constraints and objective respectively of the dual program. The integrand on $D(A)$ is bounded above (as a function of $W^\pm$ and $\phi$) if and only if $\psi=-1$ and $*U|_{D(A)}=0$, in which case it vanishes. Similarly, the integrand on $D(B)$ is bounded above if and only if $\psi=0$ and $*U|_{D(B)}=0$, in which case it vanishes. On $\I^\pm$, there is no constraint on $\psi$ or $U$. On $\I^0$ the integrand is bounded above if and only if $*U|_{\I^0}=0$, in which case it vanishes. On $\M$, the first term is bounded above (given the implicit constraint $W^\pm\in \fdc$) if and only if $U\pm d\psi\in \fdc$, in which case the maximum is at $W^\pm=0$ and vanishes. The second term is bounded above if and only if $d*U=0$, in which case it vanishes. All in all we are left with the following program:
\begin{multline}
    \text{minimize }\left( \frac{1}{2}\int_{\I^+}*U-\frac{1}{2}\int_{\I^-}*U\right)
    \text{ over }(U,\psi)\text{ such that} \\
    U\pm d\psi\in \fdc\,,\qquad d*U=0\,,\qquad\psi|_{D(A)}=-1\,,\qquad\psi|_{D(B)}=0\,,\qquad *U|_{\N\cup\I^0}=0\,.
\end{multline}
Given the no-flux boundary condition for $*U$ on $\N \cup \I^0$, and the fact that $\I^+$ is homologous to $\I^-$ with opposite orientation, relative to $\N \cup \I^0$, the two terms in the objective are actually equal, so we might as well combine them into one term. (In fact, we could have gotten the same result by shifting the boundary condition for $\phi$ by $1/2$ in the primal program \eqref{Vflow6}.) Finally, we can restore the symmetry between $D(A)$ and $D(B)$ by shifting $\psi$ by $-1/2$. 
After these cosmetic adjustments, we end up with the U-flow program:
\begin{multline}\label{Uflow2}
    \text{minimize }\int_{\I^+}*U
    \text{ over }(U,\psi)\text{ such that} \\
    U\pm d\psi\in \fdc\,,\qquad d*U=0\,,\qquad\psi|_{D(A)}=-\frac{1}{2}\,,\qquad\psi|_{D(B)}=\frac{1}{2}\,,\qquad *U|_{\N\cup\I^0}=0\,.
\end{multline}

\paragraph{From $\phi$ to $U$, with fixed $V$:} The primal program is
\be\label{phiVprimal}
    \text{maximize}\int_{D(A)}*V\text{ over }\phi\text{ such that }d\phi\pm V\in \fdc\,,\quad \phi|_{\I^\pm}=\pm\frac12\,.
\ee
Notice that the objective here does not depend on the variable. In such a case the program amounts to a feasibility test: if a feasible point exists, the program returns the objective, and if not it returns $-\infty$ (the supremum of the empty set). In the former case, Slater's condition may or may not be satisfied, and, in the latter case, it is definitely not satisfied. However, we will see explicitly that the dual program has the same optimal value.

We introduce a 1-form $W$:
\be
    \text{maximize}\int_{D(A)}*V\text{ over }\phi,W\text{ such that }d\phi=W\,,\quad W\pm V\in \fdc\,,\quad \phi|_{\I^\pm}=\pm\frac12\,.
\ee
We will impose the constraint $d\phi=W$ using a Lagrange multiplier 1-form $U$, and impose the other constraints implicitly. The Lagrangian is
\be
\begin{split}
L[\phi,W,U]&=\int_{D(A)}*V+\int_\M\sqrt g\,(W-d\phi)\cdot U \\
&=\int_{D(A)}*V+\int_\M\sqrt g\,W\cdot U+\int_\M\phi\, d{*U}+\int_{\I^0\cup\N}\phi(*U)+\frac12\left(\int_{\I^+}*U-\int_{\I^-}*U\right).
\end{split}
\ee
We now maximize the Lagrangian over $\phi$ and $W$. Requiring the Lagrangian to be bounded above with respect to $\phi$ implies $d{*U}=0$, $*U|_{\I^0\cup\N}=0$, which makes the two terms in parentheses equal. Meanwhile, $W\cdot U$ is bounded above with respect to $W$ (such that $W\pm V\in \fdc$) if and only if $U\in \fdc$, in which case the maximum is $-\ip{U}{V}$ (see lemma \ref{lem:supWU}). All in all, we are left with the following dual program:
\be\label{dualonphi}
\text{minimize}
\int_{D(A)}*V + \int_{\I^+}*U-\int_\M\sqrt g\,\ip{U}{V}\text{ over $U$ such that }U\in \fdc\,,\quad d{*U}=0\,,\quad *U|_{\I^0\cup\N}=0\,.
\ee
We now wish to show that strong duality holds, which in this case means that if the primal \eqref{phiVprimal} is feasible then the dual objective equals the primal objective, and if not then the dual objective equals $-\infty$. The first case is established by weak duality (which says that the dual objective is bounded below by the primal objective) and the fact that the last two terms in the dual objective can be set to zero by setting $U=0$. It remains to be shown that the dual objective attains $-\infty$ when the primal is infeasible. In theorem \ref{thm:Vequivalence}, we show that infeasibility implies the existence of an inextendible causal curve $\qcv$ such that $\int_\qcv dt\,\ip{(-\dot x)}{V}>1$. If $U$ is a delta-function flux tube on $\qcv$ with flux $\alpha$, then
\be
\int_\M\sqrt g\,\ip{U}{V} = \alpha \int_\qcv dt\,\ip{(-\dot x)}{V}>\alpha =  \int_{\I^+}*U\,.
\ee
Hence the dual objective can be made arbitrarily large and negative.

\paragraph{From $\psi$ to $V$, with fixed $U$:} 
The primal here is
\be
    \text{minimize}\int_{\I^+}*U\text{ over }\psi\text{ such that }U\pm d\psi\in \fdc\,,\quad \psi|_{D(A)}=-\frac12\,,\quad \psi|_{D(B)}=\frac12\,.
\ee
The story here is similar to the previous one; in particular, again the primal objective is independent of the variable, so the program amounts to a feasibility test. The dualization proceeds very similarly to the previous one, this time using lemma \ref{lem:supVX}; we leave the details as an exercise to the reader. The dual program is
\be
\text{maximize}
 \int_{\I^+}*U+\int_{D(A)}*V-\int_\M\sqrt g\,\ip{U}{V}\text{ over $V$ such that }
 d{*V}=0\,,\quad *V|_{\I}=0\,.
\ee
If the primal is infeasible, then by theorem \ref{thm:Uequivalence} and letting $V$ be a flux tube along a curve $\pcv\in \Pset$ on which $\int_\pcv ds\,\ip{U}{\dot x}<1$, the dual objective is unbounded above, establishing strong duality.

\paragraph{Discussion:} From the last two dualities we find the following minimax pair of formulas for $\Sc$:
\begin{align}
\label{maximin3}
\Sc&=\sup_{V}\infp_{U} \left(\int_{\I^+}*U+\int_{D(A)}*V-\int_\M\sqrt g\,\ip{U}{V} \right)\\
\label{minimax3}
&=\infp_{U}\sup_{V} \left(\int_{\I^+}*U+\int_{D(A)}*V-\int_\M\sqrt g\,\ip{U}{V}\right) \,,
\end{align}
where in both formulas $U$, $V$ obey
\be
U\in \fdc\,,\quad d{*U}=0\,,\quad *U|_{\I^0\cup\N}=0\,,\qquad
 d{*V}=0\,,\quad *V|_{\I}=0\,.
 \ee
Note that, since the first two terms in the objective functional are linear and $\ip{U}{V}$ is concave in $U$ and convex in $V$, the objective functional is convex in $U$ and concave in $V$, as one would expect for a minimax pair with equality. In \eqref{maximin3}, the inf on $U$ simply forces $V$ to be feasible (otherwise the objective is unbounded below, as explained below \eqref{dualonphi}), in which case the first and third terms cancel and the objective reduces to the second term, returning us to the max V-flow formula \eqref{maxVflow}. Similarly, \eqref{minimax3} returns us to the min U-flow formula \eqref{Uflowprog}. Thus, these two formulas do not teach us anything new.

\subsubsection{Equivalence of norm bounds}
\label{sec:flowdefs}

In this subsection we prove the following theorems:

\begin{theorem}\label{thm:Vequivalence}
Let $V$ be a 1-form on $\bar\M$. The following conditions are equivalent:
\begin{equation}\label{Vphicond}
\exists\ \phi\in\bksliceset_{\rm c}
\text{ s.t. }  
    d\phi\pm V\in \fdc
\end{equation}
\begin{equation}\label{Vbound}
\forall\,\qcv\in \Qset\,,\, \int_\qcv dt\,\ip{(-\dot x)}{V}\le1\,,
\end{equation}
where $\Qset$ is the set of inextendible causal curves in $\M$, $t$ is any parameter along the curve $\qcv$, and $\dot x$ is the covector dual to the tangent vector $dx^\mu/dt$ of the curve $\qcv$.
\end{theorem}

\begin{theorem}\label{thm:Uequivalence}
Let $U$ be a 1-form on $\bar\M$. The following conditions are equivalent:
\begin{equation}\label{Upsicond}
 \exists\ \psi\in\tsset_{\rm c}\text{ s.t. }
U\pm    d\psi\in \fdc
\end{equation}
\begin{equation}\label{Ubound}
U\in\fdc\,,\qquad
\forall \,\pcv\in \Pset\,,\,
    \int_\pcv ds\,\ip{U}{\dot x}\ge1\,,
\end{equation}
where $\Pset$ is the set of curves in $\M$ starting in $D(A)$ and ending in $D(B)$, $s$ is any parameter along the curve $\pcv$, and $\dot x$ is the covector dual to the tangent vector $dx^\mu/ds$ of the curve $\pcv$ .
\end{theorem}

Notice that the theorems do not require the conditions \eqref{Vflowdef1}, \eqref{Uflowdef1} in the definitions of V- and U-flows respectively.

We will need the following lemma, which provides a generalization of the usual Hamilton-Jacobi formula from classical mechanics to a situation where the Lagrangian is a convex but not necessarily differentiable function of the velocity. The Hamilton-Jacobi formula states that the variation in the on-shell action under a change in the final position is given by the final value of the canonical momentum; here, the same formula applies, but the canonical momentum is a subgradient, rather than the gradient, of the Lagrangian with respect to the velocity.
\begin{lemma}\label{thm:HJconvex}
Let $M$ be a manifold (or manifold with boundary) and $L:TM\to\R\cup\{\infty\}$ a function on its tangent bundle; we write $L(x,v)$ where $x\in M$ and $v\in T_xM$. Suppose $L$ is convex in $v$ for fixed $x$; we write $\partial_vL(x,v)$ for its subdifferential with respect to $v$ (this is a convex subset of the cotangent space at $x$, $T^*_xM$). Fix real numbers $t_0<t_1$ and a set $I\subseteq M$ (the allowed initial positions). Define the following function on $M$ (the ``on-shell action''):
\be
S_{\rm min}(y):=\inf_{\substack{X\\ X(t_0)\in I\\X(t_1)=y}}\int_{t_0}^{t_1}dt\,L(X(t),\dot X(t))\,.
\ee
Then, on the domain where $S_{\rm min}$ is finite, it is continuous and almost everywhere differentiable, and at a differentiable point $y$,
\be
dS_{\rm min}(y)\in \partial_vL(y,\dot X(t_1))\,,
\ee
where $X$ is any minimizing trajectory.\footnote{\, If, on the domain where it is finite, $L$ is differentiable with respect to $x$ with gradient $F(x,v)$, then a minimizing trajectory $X$ will admit a solution $\pi$ to the generalized Euler-Lagrange equations
\be
\pi(t)\in \partial_vL(X(t),\dot X(t))\,,\qquad\dot\pi(t) = F(X(t),\dot X(t))\,,
\ee
and $dS_{\rm min}=\pi(t_1)$ for some solution.}
\end{lemma}
We will not prove this lemma. It can be proven either (1) directly by augmenting the standard calculus of variations with the subdifferential calculus\footnote{\, There exists a large mathematical literature on the calculus of variations with non-differentiable Lagrangians. For introductions, see e.g.\ \cite{Rockafellar2001,MR3026831}.} or (2) by regularizing $L$ to make it a finite and differentiable function of $v$, applying the standard results, and then taking the limit as the regulator is removed; see \cite{HeadrickNonsmooth}.

As a warm-up to proving theorems \eqref{thm:Vequivalence} and \eqref{thm:Uequivalence}, we will first prove an analogous result in the simpler Riemannian setting, which will also be useful to us in subsection \ref{sec:RMFMC}.

\begin{theorem}\label{thm:Riemannian}
Let $M$ be a Riemannian manifold (or manifold with boundary), $A$, $B$ disjoint subsets of $M$, $\Pset$ the set of curves in $M$ connecting $A$ and $B$, and $\lambda$ a non-negative function on $M$. Then the following conditions are equivalent:
\begin{equation}\label{Riemannpsicond}
 \exists\ \psi:M\to[-1/2,1/2] \text{ s.t. } \psi|_A=-1/2\,,\quad\psi|_B=1/2\,,\quad |d\psi|\le\lambda
\end{equation}
\begin{equation}\label{lambdabound}
\forall \,\pcv\in \Pset\,,\,
    \int_\pcv ds\,\lambda\ge1\,,
\end{equation}
where $s$ is the proper distance along the curve $\pcv$.
\end{theorem}

\begin{proof}
First assume \eqref{Riemannpsicond}. Then for any $\pcv\in\Pset$,
\be
\int_\pcv ds\,\lambda\ge\int_\pcv ds\,|d\psi|\ge\int_\pcv d\psi=\psi|^B_A=1
\ee
(where in the second inequality we used the Cauchy-Schwarz inequality), implying \eqref{lambdabound}.

Now assume \eqref{lambdabound}. Rewrite the integral over $\pcv$ in a reparametrization-invariant way: $\int_\pcv ds\,\lambda = \int_\pcv dt\,|\dot x|\lambda$. Define the following functions on $M$:
\begin{equation}
    \psi_-(y):=\inf_{\substack{\pcv\\ A\text{ to }y}}\int_\pcv  dt\,|\dot x|\lambda \,,\qquad
        \psi_+(y):=\inf_{\substack{\pcv\\ B\text{ to }y}}\int_\pcv  ds\,|\dot x|\lambda\,,
\end{equation}
where the infimum is over curves $\pcv$ from $A$ to $y$ and from $B$ to $y$ respectively. Clearly for all $y\in M$ we have
\be\label{psiplusminus0}
\psi_-(y)+\psi_+(y)\ge1\,.
\ee
We also have
\be
\lim_{y\to A}\psi_-(y)=0\,,\qquad
\lim_{y\to B}\psi_+(y)=0\,.
\ee
We now calculate the gradient of $\psi_\pm$ using lemma \ref{thm:HJconvex}. The ``Lagrangian'' here, $|v|\lambda$, is a convex but not differentiable function of $v$. Its subdifferential is
\be
\partial_v(|v|\lambda) = \begin{cases}
\{\lambda v_\mu/|v|\}\,,&v\neq0\\
\{w_\mu:|w|\le\lambda\}\,,&v=0
\end{cases}
\ee
(where on the right-hand side we switched to an index notation to represent the covectors); hence $|d\psi_\pm|\le\lambda$.\footnote{\, Use of lemma \ref{thm:HJconvex} is actually overkill in this case, because by a reparametrization we can always choose to make the final velocity non-zero, so that the usual canonical momentum is well-defined and therefore $|d\psi_\pm|=\lambda$. Nonetheless we persist, in order to gain practice in applying that lemma.} By \eqref{psiplusminus0}, the regions where $\psi_-<1/2$ and where $\psi_+<1/2$ do not overlap. We therefore define\footnote{\, Another function that would also work is $\psi:=(\psi_--\psi_+)/(2(\psi_-+\psi_+))$. The argument is similar to the one for $\phi$ in the proof of theorem \ref{thm:Vequivalence}.}
\be
\psi(y):=\begin{cases}
\psi_-(y)-1/2\,,\quad&\psi_-(y)<1/2 \\
1/2-\psi_+(y)\,,\quad&\psi_+(y)<1/2 \\
0\,,\quad&\text{otherwise}
\end{cases}\,.
\ee
This function is continuous and obeys the constraints of \eqref{Riemannpsicond}.
\end{proof}

The proofs of theorems \ref{thm:Vequivalence} and \ref{thm:Uequivalence} are essentially the same, just with more complicated integrands along the curves and, in the first case, switching maximization for minimization. We start with the closer analogue, theorem \ref{thm:Uequivalence}.

\begin{proof}[Proof of theorem \ref{thm:Uequivalence}]
The proof that \eqref{Upsicond} implies \eqref{Ubound} was given in subsection \ref{sec:Uflowbounds}, but we repeat it here for completeness. $U\pm d\psi\in \fdc$ implies $U\in\fdc$. By lemma \ref{lem:Yineq2} it also implies $\ip{U}{\dot x}\ge\dot x\cdot d\psi$. Integrating over $s$ and using the boundary conditions $\psi|_{D(A)}=-1/2$, $\psi|_{D(B)}=1/2$ yields the inequality in \eqref{Ubound}.

For the converse, given $U$ satisfying \eqref{Ubound}, define the following functions on $\M$:
\begin{equation}
    \psi_-(y):=\inf_{\substack{\pcv\\ D(A)\text{ to }y}}\int_\pcv  ds\,\ip{U}{\dot x} \,,\qquad
        \psi_+(y):=\inf_{\substack{\pcv\\ D(B)\text{ to }y}}\int_\pcv  ds\,\ip{U}{\dot x}\,,
\end{equation}
where the infimum is over curves $\pcv$ from $D(A)$ to $y$ and from $D(B)$ to $y$ respectively. Clearly for all $y\in \M$ we have
\be\label{psiplusminus}
\psi_-(y)+\psi_+(y)\ge1\,.
\ee
We also have
\be
\lim_{y\to D(A)}\psi_-(y)=0\,,\qquad
\lim_{y\to D(B)}\psi_+(y)=0\,.
\ee

We now calculate the gradient of $\psi_\pm$ using lemma \ref{thm:HJconvex}. Assume first that $U$ is timelike at $y$. If, for the minimizing trajectory, $\dot x$ is timelike or spacelike at $y$, we have:
\be\label{picalc}
\ip{U}{\dot x}=\begin{cases}
U\cdot\dot x\\
-U\cdot\dot x \\
|U||\dot x_\perp|
\end{cases}\qquad
\Rightarrow\qquad
\pi =\begin{cases}
U\,,\quad&-\dot x\in \fdt  \\
-U\,,\quad&\dot x \in \fdt\\
|U|\,\dot x_{\perp}/|\dot x_{\perp}|\,,\quad& \dot x\text{ spacelike}
\end{cases}
\ee
(where $\dot x_\perp$ is the projection of $\dot x$ orthogonal to $U$). In all three cases, $U\pm\pi\in \fdc$. What if $\dot x$ is null? The gradient of $\ip{U}{\dot x}$ with respect to $\dot x$ is discontinuous when $\dot x$ is null. The subgradient set (or subdifferential) is the convex hull of the gradients on either side of the discontinuity, which we have already computed in \eqref{picalc}. Since the set of covectors $X$ obeying $U\pm X\in \fdc$ is convex (for fixed $U$), all the subgradients obey that condition, and therefore $\pi$ does. The case where $U$ is null can be treated as a limit of the timelike case, or directly using the formula  $\ip{U}{\dot x}=|U\cdot\dot x|$. We conclude that $U\pm d\psi_-\in \fdc$  and $U\pm d\psi_+\in \fdc$ everywhere in $\M$. 

By \eqref{psiplusminus}, the regions where $\psi_-<1/2$ and where $\psi_+<1/2$ do not overlap. We therefore define\footnote{\, Another function that would also work is $\psi:=(\psi_--\psi_+)/(2(\psi_-+\psi_+))$. The argument is similar to the one for $\phi$ below.}
\be
\psi(y):=\begin{cases}
\psi_-(y)-1/2\,,\quad&\psi_-(y)<1/2 \\
1/2-\psi_+(y)\,,\quad&\psi_+(y)<1/2 \\
0\,,\quad&\text{otherwise}
\end{cases}\,.
\ee
This function is continuous and obeys the constraints of \eqref{Upsicond}.
\end{proof}

\begin{proof}[Proof of theorem \ref{thm:Vequivalence}]
For the $V$-case we proceed similarly. The proof that \eqref{Vphicond} implies \eqref{Vbound} was given in subsection \ref{sec:Vflowbounds}, but we repeat it here for completeness. By lemma \ref{lem:Yineq1}, with $Y=-\dot x$, the condition $d\phi\pm V\in \fdc$ implies $\ip{(-\dot x)}{V}\le \dot x\cdot d\phi$. Integrating over $t$ and using the boundary conditions $\phi|_{\I^\pm}=\pm1/2$ yields the inequality in \eqref{Vbound}.

For the converse, given a 1-form $V$ obeying \eqref{Vbound}, we define the following functions on $\M$:
\begin{equation}
    \phi_-(y)=\sup_{\substack{\qcv\text{ causal}\\\I^-\text{ to }y}}\int_\qcv  dt\,\ip{(-\dot x)}{V} \,,\qquad
        \phi_+(y)=\sup_{\substack{\qcv\text{ causal}\\y\text{ to }\I^+}}\int_\qcv  dt\,\ip{(-\dot x)}{V}\,.
\end{equation}
Clearly for all $y\in \M$ we have
\be\label{phisumbound}
\phi_-(y)+\phi_+(y)\le1\,.
\ee
We also have
\be
\lim_{y\to \I^-}\phi_-(y)=0\,,\qquad
\lim_{y\to \I^+}\phi_+(y)=0\,.
\ee

We now calculate the gradient of $\phi_-$. In this case, the integrand is differentiable as a function of the velocity $\dot x^\mu$ on its domain, namely the set of future-directed causal vectors. However, having a constrained velocity again leads to a subtlety in the application of the Hamilton-Jacobi formula, since the velocity on the maximizing trajectory may be on the boundary of the domain (i.e.\ may be null), in which case the canonical momentum may differ from its naive value.\footnote{\, In the familiar case of the action for a massive particle, $-m\int dt\sqrt{1-\dot x^2}$, the velocity is similarly constrained to the future light cone. However, the classical trajectory is never null, so the constraint is never active and the canonical momentum is indeed given by its ``naive'' value.}
To deal with this issue, we enlarge the domain to be the whole tangent space, and implement the constraint by defining the integrand to equal $-\infty$ whenever the velocity is outside the future light cone. Thus we write
\begin{equation}
\begin{split}
    \phi_-(y)&=\sup_{\substack{\qcv\\ \I^-\text{ to }y}}\int_\qcv  dt\,\begin{cases}\ip{(-\dot x)}{V}\,,\quad
    &-\dot x\in \fdc \\ 
    -\infty\,,\quad&\text{otherwise}
    \end{cases}\,,\\
        \phi_+(y)&=\sup_{\substack{\qcv\\y\text{ to }\I^+}}\int_\qcv  dt\,\begin{cases}\ip{(-\dot x)}{V}\,,\quad
    &-\dot x\in \fdc\\
        -\infty\,,\quad&\text{otherwise}
    \end{cases}\,.
    \end{split}
\end{equation}
(Recall that the velocity vector $\dot x^\mu$ is future-directed causal if and only if the dual covector $\dot x$ obeys $-\dot x\in \fdc$.) We now deal with the various cases in turn. When $V$ is spacelike and $\dot x^\mu$ is timelike, $\ip{(-\dot x)}{V}=|V||\dot x_\perp|$, so the gradient with respect to $\dot x^\mu$ is $\pi_\mu=-|V|\dot x_{\perp\mu}/|\dot x_\perp|$, which satisfies $\pi\pm V\in \fdc$. In the limit that $\dot x^\mu$ becomes null, $\pi$ goes to infinity in the direction of $-\dot x$; the subgradient set has only this one element, so again $\pi\pm V\in \fdc$. When $V$ is future-directed timelike or null, $\ip{(-\dot x)}{V}=V\cdot\dot x$, so for timelike $\dot x^\mu$, $\pi=V$, and again $\pi\pm V\in \fdc$. For null $\dot x^\mu$, however, the integrand is discontinuous so the gradient is undefined; the subgradient set consists of all covectors $\pi$ of the form $V-\alpha\dot x$ for $\alpha\ge0$; again $\pi\pm V\in \fdc$. Similarly for the case where $V$ is past-directed timelike or null. We conclude that $d\phi_-\pm V\in \fdc$.

By the same argument, but referring to a change in the initial rather than final position (and noting that $\phi_+$ decreases to the future), we have $-d\phi_+\pm V\in \fdc$.

We now define the function
\be
\phi:=\frac{\phi_--\phi_+}{2(\phi_-+\phi_+)}\,.
\ee
We have
\be
d\phi = \frac{\phi_+}{(\phi_-+\phi_+)^2}\ d\phi_-+\frac{\phi_-}{(\phi_-+\phi_+)^2}\, (-d\phi_+)\,.
\ee
Since $d\phi_-$ and $-d\phi_+$ both satisfy $X\pm V\in \fdc$, and since the set of such covectors is closed under linear combinations of the form $\alpha_1X_1+\alpha_2X_2$ with $\alpha_1,\alpha_2\ge0$, $\alpha_1+\alpha_2\ge1$, we find (using \eqref{phisumbound}) that $d\phi$ is also in the set. This function thus obeys the constraints of \eqref{Vphicond}.
\end{proof}

\section{Thread distributions}
\label{sec:threads}

A bit thread is a bulk curve connecting two boundary points in a holographic spacetime, that represents a unit of entanglement between boundary regions. One way to define a configuration of bit threads is as the field lines (or integral curves) of a flow. However, such a configuration can also be defined simply as a set of curves obeying some density bound. The two definitions are related but not equivalent, the main difference being that, whereas the field lines of a flow cannot intersect, there is a priori no such constraint for a set of curves. Both descriptions have been used in the literature, and the relationship between them has been discussed to some extent (e.g.\ \cite{Cui:2018dyq,Headrick:2020gyq}).\footnote{\, A similar distinction between a vector field and its field lines arises for the electic and magnetic fields. In that case, as textbooks often emphasize, the vector fields are the ``real'' physical objects, while the field lines are merely a device for visualizing the fields. One way to see this is that the vector fields obey the superposition principle: the solution for a combination of sources is the sum of the individual solutions, while it would be incorrect to combine the respective sets of field lines. In the case at hand, it remains to be seen whether the flows or the thread distributions --- if either one --- reflect the underlying physics of holographic entanglement.}  In this section we will develop the second description more systematically, showing how to apply the technology of convex optimization and prove max flow-min cut type theorems directly in terms of threads, without appealing to flows. We start in subsection \ref{sec:RMFMC} in the Riemannian setting. The proofs there are new as far as we know, and of some interest in their own right. In subsection \ref{sec:VUthreads}, we then apply the same ideas to define thread versions of the V-flows and U-flows of the previous section. The method can also be applied in other settings, such as the graph max flow-min cut theorem and the Lorentzian min flow-max cut theorem \cite{Headrick:2017ucz}.\footnote{\, Riemannian threads have recently been generalized to so-called hyperthreads \cite{Harper:2021uuq,Harper:2022sky}.} We leave it to the interested reader to work out the details in those cases.

\subsection{Riemannian thread distributions}
\label{sec:RMFMC}

In this subsection we will define a thread distribution on a Riemannian manifold, prove analogues of the max flow-min cut and max multiflow theorems, and explain how to map a flow or multiflow to a thread distribution and back.

Throughout this subsection we fix a compact oriented $d$-dimensional Riemannian manifold-with-boundary $M$.\footnote{\, 
Note that we purposefully denote the manifold and its dimensionality differently from the previous section, in order to make an easier contact to the Lorentzian case of the next subsection wherein the present Riemannian manifold may be viewed as a slice $\bkslice$.  Correspondingly, the regions $A$, $B$ are now codimension-0 within the boundary $\partial M$.
} 
Let $\Pset$ be a set of (unoriented\footnote{\, Threads can be defined to be either oriented or unoriented; in fact in the original definition they were oriented \cite{Freedman:2016zud}. However, making them unoriented as we do here simplifies the analysis a bit, especially in the multiflow case, and also reflects more closely the physics of entanglement, which is a symmetric concept (in the sense that $A$ is
entangled with $B$ to the same degree that $B$ is entangled with $A$).}) curves in $M$. (Below we will specialize to the case where $\Pset$ is the set of curves connecting regions $A$ and $B$ on the boundary of $M$, but for now we can be general.) Let $\delta(x,y)$ be the delta function on $M\times M$ supported on the diagonal $y=x$ and normalized such that, for any $x$, $\int_M d^dy\sqrt{g(y)}\,\delta(x,y)=1$. We then define the following function on $M\times \Pset$:
\begin{equation}\label{Deltadef1}
    \Delta(x,\pcv) := \int_\pcv  ds\,\delta(x,y(s))\,,
\end{equation}
where $s$ is a proper-distance parameter along $\pcv$ and $y(s)$ is the corresponding point in $M$. For fixed $x$, $\Delta(x,\pcv)$ is a delta-function on $\Pset$ supported on the codimension-$(d-1)$ locus of curves passing through $x$. Note also that $\Delta(x,\pcv)$ reflects the multiplicity of times $\pcv$ passes through $x$. Given a measure $\muth$ on $\Pset$,\footnote{\, 
We remind the reader that a measure $\muth$ on $\Pset$ assigns to each subset $P\subseteq\Pset$ a non-negative number $\muth(P)$ obeying certain properties (such as $\muth(P_1\cup P_2)=\muth(P_1)+\muth(P_2)$ when $P_1\cap P_2=\emptyset$), and can be used to define the Lebesgue integral of a function on $\Pset$. We will be fairly sloppy about the analysis and measure-theory here, ignoring issues such as what regularity conditions are imposed on threads, what topology is put on the set of threads $\Pset$, what conditions are put on the measure $\muth$, what constitutes a measurable set, and so on.
} we define the \emph{thread density} at a given point $x\in M$ by $\int_\Pset d\muth(\pcv)\Delta(x,\pcv)$. It is a density in the sense that its integral over any region $r\subseteq M$ equals the length of $\pcv\cap r$, integrated over $\Pset$ with respect to $\muth$:
\begin{equation}
    \int_rd^d\!x\sqrt{g}\int_\Pset d\muth(\pcv)\Delta(x,\pcv) = \int_\Pset d\muth(\pcv)\int_{\pcv \cap r}ds\,.
\end{equation}
A \emph{thread distribution} is a measure $\muth$ such that the thread density nowhere exceeds 1:\footnote{\, 
When threads are allowed to intersect, as they are here or in the case of a multiflow (but not in the case of a flow), there exist several distinct natural density bounds, which were extensively explored in \cite{Headrick:2020gyq}. The bound \eqref{distdef} corresponds in the language of \cite{Headrick:2020gyq} to the so-called $\nu_v$-bound and is the most amenable to dualization and therefore the easiest to analyze.
}
\begin{equation}\label{distdef}
\forall\,x\in M\,,\quad\int_\Pset d\muth(\pcv)\,\Delta(x,\pcv)\le1\,.
\end{equation}
Since the measures on $\Pset$ form a convex cone, and \eqref{distdef} is a set of convex (in fact linear) constraints, the thread distributions also form a convex set. Therefore the problem of maximizing the total measure $\muth(\Pset)$ subject to \eqref{distdef} defines a linear program:
\begin{equation}\label{Rmf2}
    \text{Maximize}\quad
    \muth(\Pset)\quad\text{over measure $\muth$ on $\Pset$ subject to:}\quad
\forall\, x\in M\,,\quad \int_{\Pset}d\muth(\pcv)\,\Delta(x,\pcv)\le1\,.
    \end{equation}

Dualizing this linear program is straightforward. The Lagrange multiplier for the constraint is a function $\lambda$ on $M$, constrained to be non-negative since it is enforcing an inequality constraint. The Lagrangian is
\be
\begin{split}
L[\muth,\lambda]&=\muth(\Pset)+\int_Md^d\!x\sqrt g\,\lambda(x)\left(1-\int_\Pset d\muth(\pcv)\,\Delta(x,\pcv)\right) \\
&=\int_Md^d\!x\sqrt g\,\lambda(x)+\int_\Pset d\muth(\pcv)\left[1-\int_\pcv ds\,\lambda(y(s))\right],
\end{split}
\ee
where in the second line we rearranged terms, switched the order of the integrations over $\pcv$ and $x$, applied the definition \eqref{Deltadef1} of $\Delta$, and performed the integration over $x$ in the last term. The Lagrangian is bounded above as a function of $\muth$ if and only if, for all $\pcv$, the quantity in square brackets is non-positive; this becomes the constraint for the dual program, and when it is satisfied the second term vanishes on the maximum. The dual program is thus:
\begin{equation}\label{Rmc2}
    \text{Minimize}\quad
    \int_M d^d\!x\sqrt g\,\lambda(x)\quad\text{over function $\lambda:M\to\R^+$ subject to:}\quad
\forall\, \pcv \in \Pset\,,\quad 
\int_\pcv ds\,\lambda(y(s))
\ge1\,.
\end{equation}
Strong duality is guaranteed for linear programs by the existence of a feasible configuration for the primal; for \eqref{Rmf2}, such a feasible configuration is the zero measure $\muth=0$.

We will now apply the duality between \eqref{Rmf2} and \eqref{Rmc2} to prove an analogue, in the language of thread distributions, of the Riemannian max flow-min cut theorem. Given non-overlapping regions $A,B$ of $\partial M$, let $\Pset$ be the set of curves in $M$ with one endpoint in $A$ and the other in $B$. Given a hypersurface $\surf$ in $M$ homologous to $A$ relative to $\eowsurf:=\partial M\setminus AB$, every curve $\pcv\in \Pset$ intersects $\surf$. Therefore a delta function supported on $\surf$ and normalized in the obvious way (i.e.\ equal to $\delta(x^1)$ in Gaussian normal coordinates with $x^1$ the coordinate transverse to $\surf$) is a feasible function $\lambda$ for \eqref{Rmc2}. The objective evaluated on this function equals the area of $\surf$. So by the weak duality between \eqref{Rmf2} and \eqref{Rmc2} we have, for any thread distribution $\muth$,
\begin{equation}
    \muth(\Pset) \le \area(\surf)\,.
\end{equation}
The max flow-min cut theorem asserts that this inequality is tight, namely
\begin{equation}
    \sup\muth(\Pset) = S(A):=\inf_{\surf\sim A}\area(\surf)\,,
\end{equation}
where the supremum is over thread distributions and the infimum over surfaces homologous to $A$ relative to $\eowsurf$.  Given strong duality, to prove the theorem it is sufficient to show that the optimal value of \eqref{Rmc2} equals $S(A)$, in other words that the integral of any feasible $\lambda$ is bounded below by $S(A)$. By theorem \ref{thm:Riemannian}, given a feasible $\lambda$, there exists a function $\psi:M\to[-1/2,1/2]$ such that $\psi|_A=-1/2$, $\psi|_B=1/2$, $|d\psi|\le\lambda$. We then have
\be\label{MFMCproof}
\int_M\sqrt g\,\lambda \ge \int_M\sqrt g\,|d\psi|\ge S(A)\,,
\ee
where in the second inequality we used the coarea formula and the fact that the level sets of $\psi$ are homologous to $A$ (relative to $\eowsurf$).

The generalization to multiple boundary regions is straightforward, yielding a thread-distribution version of the max multiflow theorem \cite{Cui:2018dyq}. Let $\{A_i\}$ be a set of non-overlapping regions of $\partial M$, let $\Pset_i$ be the set of threads connecting $A_i$ to $\cup_{j\neq i}A_j$, and set $\Pset:=\cup_i\Pset_i$. We will now show that a thread distribution exists such that $\muth(\Pset_i)=S(A_i)$ for all $i$. We have $\muth(\Pset)=\sum_i\muth(\Pset_i)/2$, since each curve in $\Pset$ contributes to two terms in the sum. Hence the primal objective is bounded above by $\sum_iS(A_i)/2$, and achieves this bound if and only if $\muth(\Pset_i)=S(A_i)$ for all $i$. We will now show that the dual objective is bounded below by the same quantity, implying that both bounds are achieved. This part of the proof is essentially the same as that of the multiflow-based theorem, so we will be brief; details can be found in \cite{Cui:2018dyq}. Given a feasible function $\lambda$ for the dual program, we define $\psi_i(x)$ as the minimal integral of $\lambda$ over any path connecting $A_i$ to $x$. The level set of $\psi_i$ for any value between $0$ and $1$ is homologous to $A_i$ (relative to $\eowsurf:=\partial M\setminus\cup_iA_i$), so its area is bounded below by $S(A_i)$. For any $x\in M$ and $i\neq j$, $\psi_i(x)+\psi_j(x)\ge1$, so the level sets of $\psi_i$ and $\psi_j$ for values between 0 and $1/2$ do not intersect. This bounds the dual objective below by $\sum_iS(A_i)/2$.

\subsubsection{Mapping between flows \& thread distributions}
\label{sec:flowthreadmapping}

Of course, it's no coincidence that thread distributions obey the same theorems as flows and multiflows. As we will now explain, a flow (or multiflow) can be mapped to a thread distribution and vice versa.\footnote{\, A similar mapping in the graph setting, between flows or multiflows and set of threads (or paths), was discussed in \cite{Headrick:2020gyq} (see appendix A).} Roughly speaking, the field lines (or streamlines, or integral curves) of a flow correspond to the curves in a thread distribution. However, there are a few complications in this mapping. First, while a field line of a divergenceless vector field cannot end in the interior of $M$, it need not end on the boundary --- it can be a loop, or keep going forever in both directions without endpoints.\footnote{\, An example of a flow whose field lines go forever in both directions is  $dx^1+\alpha dx^2$ on the square torus, where $\alpha$ is irrational. Similar examples can be constructed on manifolds with boundary. Note that the divergenceless condition prevents a field line from having just one endpoint on the boundary.} Second, the field lines are naturally oriented, whereas threads are unoriented. Third, the field lines of a flow don't intersect, whereas the threads in a distribution may do so. These complications will be dealt with as we go along.

We first define the map from a flow $v$ to a thread distribution $\muth_v$. As above, we fix boundary regions $A$, $B$ and define $\eowsurf:=\partial M\setminus AB$ and $\Pset$ as the set of curves with one endpoint in $A$ and the other in $B$. We also assume that the flow $v$ obeys the no-flux boundary condition on $\eowsurf$, $*v|_{\eowsurf}=0$. Let $\Pset_v\subset \Pset$ be the set of field lines of $v$ running from $A$ to $B$ (with their orientations dropped), $r$ be the set of points in $M$ through which they pass, and $\tilde v$ the 1-form equal to $v$ on $r$ and $0$ elsewhere. We can then define the measure $\muth_v$ on $\Pset$ by $\muth_v(\Pset\setminus \Pset_v)=0$ and $\muth_v(A')=\int_{A'}*\tilde v$ for any subregion $A'\subseteq A$. At any point $x\in M$, the density of threads equals $|\tilde v(x)|$, so we have
\begin{equation}\label{muvdensity}
    \int_\Pset d\muth_v(\pcv)\Delta(x,\pcv) =|\tilde v(x)|\le |v(x)|\le1\,.
\end{equation}
We conclude that $\muth_v$ obeys the density bound and is therefore a thread distribution. 
The total measure of $\Pset$ is at least the flux of $v$ on $A$:
\begin{equation}\label{vfluxbound}
    \muth_v(\Pset)\ge\int_A*v\,. 
\end{equation}
The reason it's not necessarily an equality is that any field lines running from $B$ to $A$ contribute zero to the left-hand side but negatively to the right-hand side. While both the set of flows and the set of thread distributions are convex, the convexity is implemented differently in the two cases; in general, given flows $v_1,v_2$ and $\alpha\in(0,1)$,
\begin{equation}
    \muth_{\alpha v_1+(1-\alpha)v_2}\neq\alpha\muth_{v_1}+(1-\alpha)\muth_{v_2}\,.
\end{equation}
Equality holds if and only if $\tilde v_1,\tilde v_2$ are everywhere parallel (i.e.\ $\tilde v_1\cdot\tilde v_2=|\tilde v_1||\tilde v_2|$). Otherwise the distribution on the right-hand side includes intersecting threads, whereas the left-hand side, like any distribution derived from a flow, does not.

To convert a thread distribution $\muth$ into a flow $v_\muth$,\footnote{\, The subscript $\muth$ in $v_\muth$ is not an index but rather indicates the dependence on the thread distribution $\muth$.} we can turn each thread into a delta-function-localized ``flux tube'', and then integrate over $\Pset$. We then have
\begin{equation}\label{vfrommu}
v_\muth(x) :=\int_{\Pset}d\muth(\pcv)\,\Delta(x,\pcv)\,\dot x\,,
\end{equation}
where $\dot x$ is the unit tangent covector to $\pcv$ at $x$, with the orientation along $\pcv$ from $A$ to $B$. The integrand $\Delta(x,\pcv)\dot x$, as a function of $x$ for fixed $\pcv$, is a flux tube supported on $\pcv$, and is divergenceless by virtue of the fact that $\pcv$ has no endpoints in the interior of $M$. Hence $v_\muth$ is divergenceless. Using the Cauchy-Schwarz (or triangle) inequality, we have
\begin{equation}
    \left|v_\muth(x)\right| \le  \int_\Pset d\muth(\pcv)\,\Delta(x,\pcv)\le1\,,
\end{equation}
so $v_\muth$ is indeed a flow. With this construction, the flux of $v_\muth$ on $A$ equals the number of threads connecting $A$ and $B$:
\begin{equation}\label{vfluxequal}
\int_A*v_\muth= \muth(\Pset)\,. 
\end{equation}
Unlike the above mapping from flows to thread distributions, this one does preserve convex combinations:
\begin{equation}
v_{\alpha\muth_1+(1-\alpha)\muth_2} = \alpha v_{\muth_1}+(1-\alpha)v_{\muth_2}\,.
\end{equation}

The relations \eqref{vfluxbound} and \eqref{vfluxequal} together imply that the ``max flow'' is the same whether computed using flows or thread distributions:
\begin{equation}
\sup_\muth\muth(\Pset) = \sup_{v}\int_A*v\,.
\end{equation}
Note also that these two mappings are not inverses of each other; given a flow $\hat v$ and thread distribution $\hat\muth$ we have in general neither $v_{\muth_{\hat v}}=\hat v$ (since flow lines running from $B$ to $A$ are discarded in transforming $\hat v$ into a thread distribution) nor $\muth_{v_{\hat\muth}}=\hat\muth$ (since any intersecting threads in $\hat\muth$ are recombined in transforming $\hat\muth$ into a flow).

Similar constructions convert a multiflow into a thread distribution and vice versa. We remind the reader that, given a decomposition $\{A_i\}$ of $\partial M$, a multiflow is a set of 1-forms $v_{ij}$ ($i<j$) obeying
\begin{equation}
d{*v_{ij}}=0\,,\qquad\sum_{i,j}|v_{ij}|\le1\,,\qquad *v_{ij}|_{\partial M\setminus(A_i \cup A_j)}=0\,.
\end{equation}
The corresponding thread distribution is the sum of the distributions $\muth_{v_{ij}}$ for the component flows. Via the analogue of \eqref{vfrommu} for each pair of distinct boundary regions, a thread distribution can also be converted into a multiflow.

\subsection{V-thread \& U-thread distributions}
\label{sec:VUthreads}

In this subsection we return to the Lorentzian setting of the rest of the paper, as laid out in subsection \ref{sec:setup}. We also continue to fix a decomposition of the boundary $\N$ into $D(A)$ and $D(B)$, and to suppress the $A$- and $B$-dependence of various quantities, for example writing $\Sc$ for $\Sc(A:B)$. In section \ref{sec:flows}, we defined V-flows and U-flows, and will correspondingly define two types of threads, V-threads and U-threads. Like the Riemannian threads of the previous subsection, both types of threads will be subject to density bounds; unlike in that case, however, these density bounds will be non-local. It will turn out that the density bound for V-threads is enforced by U-threads and vice versa, and the problem of maximizing the number of V-threads is dual to the problem of minimizing the number of U-threads (as was the case for the max V-flow and min U-flow programs \eqref{maxVflow}, \eqref{Uflowprog}). The solution to both of these problems will also equal the quantity $\Sc$ defined in \eqref{Sccdef}, as we will show by explaining how a V/U-thread distribution can be converted into a V/U-flow.

We define a \emph{V-thread} $\pcv$ as an unoriented open curve in $\M$ with one end in $D(A)$ and the other in $D(B)$, and a \emph{U-thread} $\qcv$ as an inextendible causal curve in $\M$ (which necessarily has one end on $\I^+$ and the other on $\I^-$). V-threads are not required to have any particular causal character. $\Pset$ and $\Qset$ are the set of all V- and U-threads respectively. We define the following function on $\Qset \times \Pset$:
\begin{equation}\label{Deltadef}
\Delta(\qcv,\pcv) := \int_\qcv dt\int_\pcv ds\,\delta(x(t),y(s))\,\ip{(-\dot x)}{\dot y}\,,
\end{equation}
where $t$ is a parameter along $\qcv$, $x(t)$ is the corresponding point in $\M$, and $\dot x$ is dual covector to the tangent vector $dx(t)/dt$; similarly for $\pcv$, $y$, and $s$. (For simplicity we use the overdot for both $d/dt$ and $d/ds$.) Since the wedgedot function $\ip{\ }{\ }$ is bihomogeneous, $\Delta(\qcv,\pcv)$ is independent of the parametrization of $\qcv$ and $\pcv$. $\Delta$ is a delta-function supported on the codimension-$(D-2)$ locus in $\Qset\times \Pset$ on which $\qcv$ and $\pcv$ intersect. Given a measure $\muth$ on $\Pset$, we define the density of V-threads along a given U-thread $\qcv$ as $\int_\Pset d\muth(\pcv)\,\Delta(\qcv,\pcv)$. (Note that, given the definition of $\Delta$, this definition automatically includes an integration along $\qcv$, so this is an integrated, not pointwise, density.) We define a \emph{V-thread distribution} as a measure $\muth$ on $\Pset$ such that this density never exceeds 1:
\begin{equation}\label{Vtddef}
 \forall\, \qcv \in \Qset\,,\quad \int_\Pset d\muth(\pcv)\, \Delta(\qcv,\pcv)\le1\,.
\end{equation}
Similarly, given a measure $\nuth$ on $\Qset$, we define the density of U-threads along a given V-thread $\pcv$ as $\int_\Qset d\nuth(\qcv)\,\Delta(\qcv,\pcv)$, and a \emph{U-thread distribution} (for $A$) as a measure $\nuth$ on $\Qset$ satisfying
\begin{equation}\label{Utddef}
\forall\, \pcv\in \Pset\,,\quad 
\int_\Qset d\nuth(\qcv)\, \Delta(\qcv,\pcv)\ge1\,.
\end{equation}
Essentially, the U-threads have to form a sufficient ``barrier'' separating $D(A)$ from $D(B)$.

The problem of maximizing $\muth(\Pset)$ over V-thread distributions defines a linear program:
\begin{equation}\label{Vtprogram}
    \text{Maximize}\quad
    \muth(\Pset)\quad\text{over measure $\muth$ on $\Pset$ subject to:}\quad
\forall\, \qcv\in \Qset\,,\quad \int_{\Pset}d\muth(\pcv)\,\Delta(\qcv,\pcv)\le1\,.
    \end{equation}
The problem of minimizing $\nuth(\Qset)$ over U-thread distributions also defines a linear program:
\begin{equation}\label{Utprogram}
    \text{Minimize}\quad
    \nuth(\Qset)\quad\text{over measure $\nuth$ on $\Qset$ subject to:}\quad
\forall\, \pcv\in \Pset\,,\quad 
\int_\Qset d\nuth(\qcv)\, \Delta(\qcv,\pcv)\ge1\,.
\end{equation}
The programs \eqref{Vtprogram}, \eqref{Utprogram} are dual to each other. Strong duality is again guaranteed by the existence of a feasible point for one of the programs, namely $\muth=0$ for \eqref{Vtprogram}, so we have
\be
\sup_\muth\muth(\Pset) = \infp_\nuth\nuth(\Qset)\,.
\ee
Next we will show that these quantities both equal $\Sc$.\footnote{\,
In the subsection on Riemannian thread distributions, below \eqref{MFMCproof}, it was explained how to define a thread distribution for multiple boundary regions, the analogue of a multiflow. One may be tempted to do the same in the Lorentzian setting. However, for the reasons given in footnote \ref{LorentzianNesting}, this is not straightforward.
}

\subsubsection{Mapping between flows \& thread distributions}

Just as, in subsection \ref{sec:flowthreadmapping}, we converted Riemannian flows to thread distributions and vice versa, we can convert V/U-flows to V/U-thread distributions and vice versa. We start with the V-flows, which are more parallel with the Riemannian case. 
To convert a V-flow $V$ into a V-thread distribution, we first strip away any field lines that do not run from $D(A)$ to $D(B)$. Call the resulting 1-form $\tilde V$, and set of field lines $\Pset_V\subset \Pset$. We then define the measure $\muth_V$ on $\Pset$ so that $\muth_V(\Pset\setminus \Pset_V)=0$ and $\muth_V(D')=\int_{D'}*\tilde V$ for any subregion $D'\subseteq D(A)$. For any U-thread $\qcv$, we then have, similarly to \eqref{muvdensity},
\begin{equation}
\int_\Pset d\muth_V(\pcv)\,\Delta(\qcv,\pcv)=\int_\qcv dt\,\ip{(-\dot x)}{\tilde V}\le\int_\qcv dt\,\ip{(-\dot x)}{V}\le1\,.
\end{equation}
The first equality follows from the definition \eqref{Deltadef} of the function $\Delta$, the first inequality from the fact that $\tilde V$ equals either $V$ or 0 at every point in $\M$, and the last inequality from the bound \eqref{Vbound1} obeyed by any V-flow. Thus $\muth_V$ is indeed a thread distribution. Furthermore, similarly to \eqref{vfluxbound},
\begin{equation}\label{Vfluxbound}
\muth_V(\Pset) \ge \int_{D(A)}*V\,.
\end{equation}
To go from a V-thread distribution $\muth$ to a V-flow $V_\muth$,\footnote{\, The subscript on $V_\muth$ is not an index but rather denotes the dependence on the measure $\muth$.} we orient the threads from $D(A)$ to $D(B)$ and put a delta-function flux tube on each:
\be
V_\muth(x):=\int_{\Pset}d\muth(\pcv)\int_\pcv ds\,\delta(x,y(s))\,\dot y\,.
\ee
For any U-thread $\qcv$, we have
\be
\int_\qcv dt\,\ip{(-\dot x)}{V_\muth} \le
\int_{\Pset
}d\muth(\pcv)\Delta(\qcv,\pcv)\le1\,,
\ee
where in the first inequality we used the convexity of $\ip{\ }{\ }$ in its second argument; hence $V_\muth$ is indeed a V-flow. We also have
\be\label{Vfluxequal}
\int_{D(A)}*V_\muth = \muth(\Pset)\,.
\ee
The mappings between V-flows and V-thread distributions, and the relations \eqref{Vfluxequal}, \eqref{Vfluxbound} imply that the maximum number of threads in $\Pset$, for any V-thread distribution, equals $\Sc$:
\be\label{maxVthread}
\sup_\muth\muth(\Pset)=\sup_{V\in\F}\int_{D(A)}*V=\Sc\,.
\ee

Similarly, we can convert a U-flow into a U-thread distribution and back. One difference with the V case is that every U-thread and every field line of a U-flow necessarily goes from $\I^-$ to $\I^+$, so there is no need to select a subset of the threads or field lines, or to fuss with orientations; the analogue of \eqref{Vfluxbound} is therefore an equality. We leave the details as an exercise to the reader. In the end, similarly to \eqref{maxVthread}, we have
\be\label{maxUthread}
\inf_\nuth\nuth(\Qset)=\inf_{U\in\G}\int_{\I^+}*U=\Sc\,.
\ee

With \eqref{maxVthread} and \eqref{maxUthread} in hand, we can add another row to our diagram \eqref{diagram} of formulas for the quantity $\Sc$:
\be\label{diagram2}
\minCDarrowwidth50pt
\begin{CD}
\boxed{\sup_{\phi\in\bksliceset_{\rm c}}\infp_{\psi\in\tsset_{\rm c}} \int_\M\sqrt g\,\ip{d\phi}{d\psi}} @Z\text{minimax}Z
\text{theorem}Z \boxed{\infp_{\psi\in\tsset_{\rm c}}\sup_{\phi\in\bksliceset_{\rm c}} \int_\M\sqrt g\,\ip{d\phi}{d\psi}} \\
@X{\scriptsize\begin{array}{r}\psi\leftrightarrow V\\\text{duality}\end{array}}XX @XX{\scriptsize\begin{array}{r}\phi\leftrightarrow U\\ \text{duality}\end{array}}X \\
\boxed{\sup_{V\in\F}\int_{D(A)}*V} @Z{(V,\phi)\leftrightarrow(U,\psi)}Z\text{duality}Z \boxed{\inf_{U\in\G}\int_{\I^+}*U} \\
@X{\scriptsize\begin{array}{r}\text{V-flow$\leftrightarrow$V-thread}\\ \text{conversion}\end{array}}
XX  @XX
{\scriptsize\begin{array}{l}\text{U-flow$\leftrightarrow$U-thread}\\ \text{conversion}\end{array}}
X \\
\boxed{\sup_\muth\muth(\Pset)} @Z{\muth\leftrightarrow\nuth}Z
\text{duality}Z
\boxed{\inf_\nuth\nuth(\Qset)}
\end{CD}
\ee

%---------------------------------------------------

\section{Holographic spacetimes}
\label{sec:solutions}

In the previous sections, for given boundary regions $A$, $B$, we defined the three quantities $\Splus$, $\Sminus$, $\Sc$, gave several formulas for each, and showed that they obey the relation
\be\label{Srelation}
\Sminus\le\Sc\le\Splus\,.
\ee
All of this was essentially assuming just global hyperbolicity for the bulk spacetime --- no field equations, energy conditions, or boundary conditions were invoked.

Starting in this section, we specialize to ``standard'' classical holographic spacetimes, in other words ones in which the metric obeys AdS boundary conditions as well as the null energy (or curvature) condition (NEC). In such spacetimes, the maximin quantity $S_-$ equals the area $\SHRT=S(A:B)$ of the HRT (or entanglement wedge cross-section, EWCS) surface $\HRT$ \cite{Wall:2012uf,paper1}. In subsection \ref{sec:HRTsurface}, we will show that both inequalities in \eqref{Srelation} are saturated, so $\Sc$ and $\Splus$ also equal $\SHRT$. Altogether, counting the two new formulas for $S_-$ and two formulas for $S_+$ given in subsection \ref{sec:maximin}, as well as the six formulas for $\Sc$ shown in \eqref{diagram2}, we have therefore provided ten new covariant formulas for the holographic EE. In subsection \ref{sec:maxVflows}, we will focus on the V-flow and U-flow formulas of section \ref{sec:flows}, and see what the optimal flows look like. Then, in subsection \ref{sec:multiple}, we will generalize our results to multiple boundary regions.

\subsection{Relation to the HRT formula}
\label{sec:HRTsurface}

To orient the reader, we start by previewing the logic of the argument. There are two crucial features of the holographic setup which allow us to saturate both inequalities in \eqref{Srelation} and thereby collapse the maximin, minimax, and convex-relaxed values into one quantity.  The first is the NEC which ensures that each null normal congruence from any extremal surface has non-positive and decreasing expansion.  The second is less familiar, involving AdS boundary conditions in the more general context of allowing end-of-the-world branes comprising $\I^0$.  Together, these ensure that within the time-sheet formed by the bulk part of the entanglement wedge boundary, the HRT surface is area-maximizing within its relative homology class.  On the other hand, there exists a maximin Cauchy slice containing the HRT surface, on which the latter is area-minimizing within its relative homology class. Taken together, these two observations will then suffice to establish the HRT surface as a global saddle point.

To flesh this out more explicitly, it will be useful to set up a few deﬁnitions. Recall that, given a slice $\bkslice$, $\Gamma_{\bkslice}$ (defined below \eqref{maximin}) is the set of surfaces $\surf$ in $\sigma$ homologous to $A_\bkslice:=D(A) \cap \bkslice$ (relative to $\eowsurf_\bkslice:=\I^0\cap\bkslice$); and, given a time-sheet $\ts$ homologous to $D(A)$ (relative to $\I$), $\Gamma_\ts$ (defined at the end of section \ref{sec:setup}) is the set of surfaces of the form $\ts\cap\bkslice$ for some slice $\bkslice$. The union of the former over all slices, or equivalently the union of the latter over all time-sheets, defines the full set of surfaces that are spacelike-homologous to $A$, which we call $\Gamma$:
\be
\Gamma:=\bigcup_{\bkslice\in\bksliceset}\Gamma_\bkslice = \bigcup_{\ts\in\tsset}\Gamma_\ts = \left\{\bkslice\cap\ts|\,\bkslice\in\bksliceset,\ts\in\tsset\right\}.
\ee
We denote by $\Gamma_{\rm ext}$ the set of surfaces in $\Gamma$ that are extremal.\footnote{\, By an extremal surface we mean one that extremizes the area with respect to all variations in position, including those that move the intersection with $\I^0$ (if any). Such a surface has vanishing mean curvature vector and intersects $\I^0$ orthogonally.} We then have
\be\label{originalHRT}
\SHRT:=\inf_{\surf\in\Gamma_{\rm ext}}\area(\surf)\,.
\ee
The minimizer (or any minimizer, if there is more than one) is the (or an) HRT/EWCS surface $\HRT$.

It will also be useful to slightly expand the definition of a time-sheet to allow null pieces. Specifically, we write $\bar\tsset$ for the set of hypersurfaces that are piecewise timelike or null and homologous to $D(A)$ (whereas $\tsset$ only includes piecewise timelike hypersurfaces). Since any element of $\bar\tsset$ is a limit of (piecewise) timelike hypersurfaces, this does not change the infimum in the minimax formula:
\be
S_+=\infp_{\ts\in\bar\tsset}\ \sup_{\surf\in\Gamma_\ts}\area(\surf)\,.
\ee

\begin{figure}[tbp]
\centering
\includegraphics[width=0.45\textwidth]{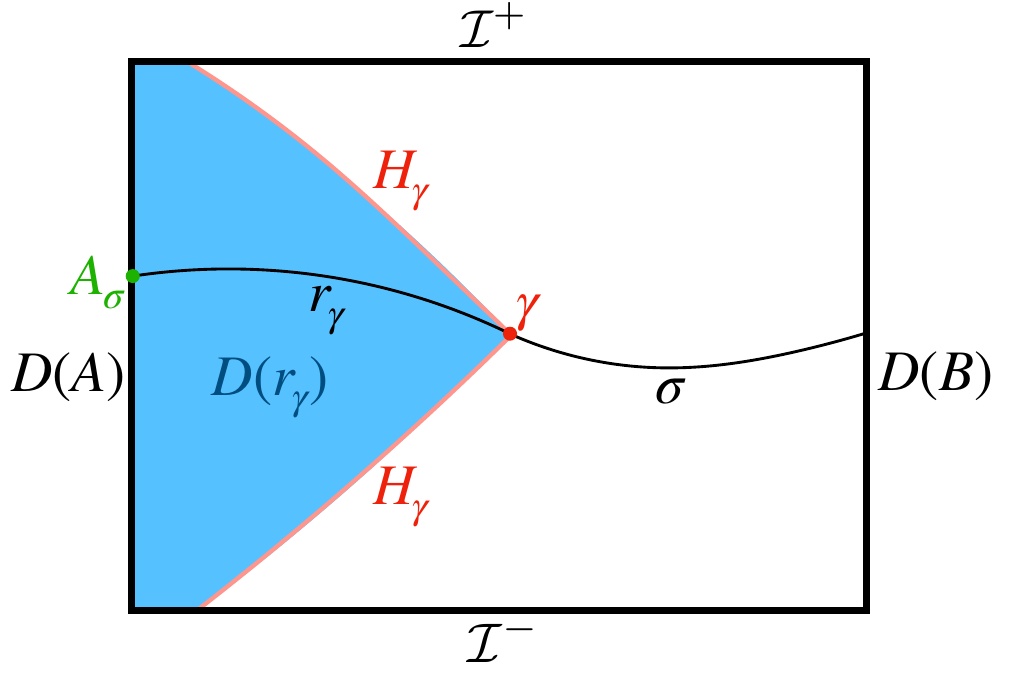}\caption{\label{fig:holoST}
Illustration of geometric constructs associated to an element $\surf\in\Gamma_{\rm ext}$, an extremal surface spacelike-homologous to $A$, described in item \eqref{gammafacts}: $\bkslice$ is a slice containing $\surf$; $A_\bkslice:=\bkslice\cap D(A)$; $r_\surf$ is the homology region on $\sigma$ interpolating between $\surf$ and $A_\bkslice$; $D(r_\surf)$ is the bulk domain of dependence of $r_\surf$; and the null surface $H_\surf$ is the bulk part of the boundary of $D(r_\surf)$. The following general facts about these constructs are important: $D(r_\surf)$ and $H_\surf$ are independent of the choice of $\bkslice$; $H_\surf$ is homologous to $D(A)$ (relative to $\I$); and $\surf$ is maximal on $H_\surf$.
}
\end{figure}

Two important sets of facts about surfaces in $\Gamma_{\rm ext}$, hinging on AdS boundary conditions and the NEC are the basis for what follows (proofs and discussion can be found in \cite{Wall:2012uf,Marolf:2019bgj,Grado-White:2020wlb,paper1}):
\renewcommand\labelenumi{(\theenumi)}
\begin{enumerate}
\item\label{gammafacts} Given any surface $\surf\in\Gamma_{\rm ext}$, let $\bkslice$ be a slice containing $\surf$, $r_\surf\subset\bkslice$ the homology region on $\bkslice$ between $\surf$ and 
$A_\bkslice$, $D(r_\surf)$ its causal domain, and $H_\surf$ the bulk part of the boundary of $D(r_\surf)$, i.e.\ $H_\surf:=\partial D(r_\surf)\setminus(\I\cup\N)$. $D(r_\surf)$ and $H_\surf$ are independent of the choice of $\bkslice$, and furthermore:
\begin{enumerate}
\item $H_\surf$ is made up of congruences of null geodesics shot orthogonally from $\surf$ toward the past and future in the direction of $D(A)$.\footnote{\, 
This requires, in addition to AdS boundary conditions and NEC, the assumption that $\I^0$ (if non-empty) is ``well-behaved'', in the sense that while generators of $H_\surf$ may fall into $\I^0$, they never emerge from $\I^0$. (On the other hand, in the case $\I^0=\surf^0$ where the above-mentioned energy conditions don't enter, any generators of  $\partial D(r_\surf)$ that emanate from this $\surf^0$ are part of $\I$ and hence not included in $H_\surf$.)
} Since $\surf$ is extremal, the expansion $\theta$ on these congruences is initially zero, and by virtue of the NEC and Raychaudhuri and Einstein equations, $\dot\theta\le0$. Therefore, $\surf$ has maximal area among surfaces in $\Gamma$ contained in $H_\surf$.\footnote{\, 
A-priori, if one just arbitrarily extends $\surf^0$ in a time direction and declares it to be $\I^0$, one would have to worry that $H_\surf$ contains future light cone from $\surf \cap \I^0$, i.e.\ that $\partial I^+(\bkslice \setminus r_\surf) \setminus (\I\cup\N)$ does in fact contain some generators emanating from $\I^0$, in which case $H_\surf$ would have an expanding part, and it would no longer follow that $\surf$ is maximal on $H_\surf$.  However, this contingency is excluded by the previous footnote.  The underlying assumption that generators of $H_\surf$ can't emanate from $\I^0$ is therefore a strong (and crucial) one.
}
\item $D(r_\surf)\cap\N=D(A)$, implying that $H_\surf$ is homologous to $D(A)$ (relative to $\I$). Since $H_\surf$ is null, $H_\surf\in\bar\tsset$.
\end{enumerate}
An important special case is the HRT surface $\HRT$, for which $D(r_\HRT) =: \ew(A)$ is the entanglement wedge and $H_\HRT=:\hor(A)$ is the entanglement horizon. However, it is important that properties (a) and (b) hold for any surface $\surf\in\Gamma_{\rm ext}$, not just the HRT surface. See figure \ref{fig:holoST} for a sketch illustrating these constructs.
\item The sup and inf in the maximin formula \eqref{maximin} are attained, and the maximin surface $\surf_-$ is the HRT surface, $\surf_-=\HRT$, so $\Sminus=\SHRT$. The hard parts here are showing that the sup is attained and that the maximin surface is extremal (which naively follows from the fact that it is extremal with respect to variations in both the space and time normal directions; however there are subtleties to this argument). The fact that $\surf_-$ is minimal among extremal surfaces is argued as follows. For $\surf\in\Gamma_{\rm ext}$, let $\surf':=H_{\surf}\cap\bkslice_-$ (where $\bkslice_-$ is a maximin slice). Then by the fact that $\surf$ is maximal on $H_{\surf}$, $\area(\surf)\ge\area(\surf')$; and by the fact that $\surf_-$ is minimal on $\bkslice_-$, $\area(\surf')\ge\area(\surf_-)$.
\end{enumerate}

The fact that $\HRT$ is, by (1), maximal on the hypersurface $\hor(A)$ and, by (2), minimal on the slice $\bkslice_-$, makes it a global saddle point, establishing that it is not only the maximin but also the minimax surface: $\surf_-=\surf_+=\HRT$ and $\Sminus=\Splus=\SHRT$ (see \eqref{globalsaddle}). We can also make the argument a slightly different way: by property (1), any surface $\surf\in\Gamma_{\rm ext}$ is maximal on the hypersurface $H_\surf$, so by \eqref{originalHRT} $\Splus\le \SHRT$; the min-max inequality $\Sminus\le\Splus$ and the fact $\Sminus=\SHRT$ then imply the equality of all three. Finally, since $\Sc$ is caught between $\Sminus$ and $\Splus$, we have
\be\label{allequal}
\Sminus=\Sc=\Splus=\SHRT\,,
\ee
Eq.\ \eqref{allequal} opens the door to using the flow- and thread-based convex programs to calculate the HRT area.

One way to think about the collapse of the three quantities $S_-$, $\Sc$, and $S_+$ into a single quantity is that the contents of sections \ref{sec:relaxation}--\ref{sec:threads} are essentially ``kinematics'' --- things that follow just from having a (globally hyperbolic) Lorentzian spacetime. The ``dynamics'' that collapses those three quantities into a single one stems from the properties of the entanglement horizon, which require field equations, energy conditions, and boundary conditions.

\subsection{Optimal flows}
\label{sec:maxVflows}

In this subsection we will describe solutions, or optimal configurations, for the V-flow and U-flow formulas \eqref{maxVflow}, \eqref{Uflowprog}. The solutions for the other formulas can be deduced from these. As we will see, the optimal configurations are strongly constrained by the properties of holographic spacetimes. These constraints are simplest to describe in the generic case, by which we mean that (1) $\HRT$ is the unique minimal surface on some slice $\bkslice_-$, and (2) $\dot\theta<0$ on both the future and past branches of the entanglement horizon $\hor(A)$, implying that $\HRT$ is the unique maximal-area acausal surface in $\hor(A)$.\footnote{\, A necessary and sufficient condition for $\dot\theta<0$ everywhere on $\hor(A)$ is for the following inequality to hold at every point on $\HRT$ and for both null normals $k^\mu$:
    \begin{equation}\label{generic}
k^\mu k^\nu(R_{\mu\nu}+K_{\mu\lambda\rho}{K_\nu}^{\lambda\rho})>0\,,
    \end{equation}
    where $K_{\mu\lambda\rho}$ is the extrinsic curvature tensor of $\HRT$. Here we used the fact that the initial shear of the congruence is $\bdyslice_{\lambda\rho}=k^\mu K_{\mu\lambda\rho}$.}

We start with the V-flow case. In order for the flux of $V$ through $D(A)$, and therefore through $\hor(A)$, to equal the area of $\HRT$, all of that flux must pass through $\HRT$, and there must be no flux elsewhere on $\hor(A)$; in other words, the flux $*V|_{\hor(A)}$ must be proportional to a delta-function supported on $\HRT$. The same reasoning applies, of course, to $\hor(B)$; thus, the flux squeezes from $\ew(A)$ into $\ew(B)$ through $\HRT$, without passing through its past or future $I^\pm(\HRT)$.

\begin{figure}[tbp]
\centering
\includegraphics[width=0.85\textwidth]{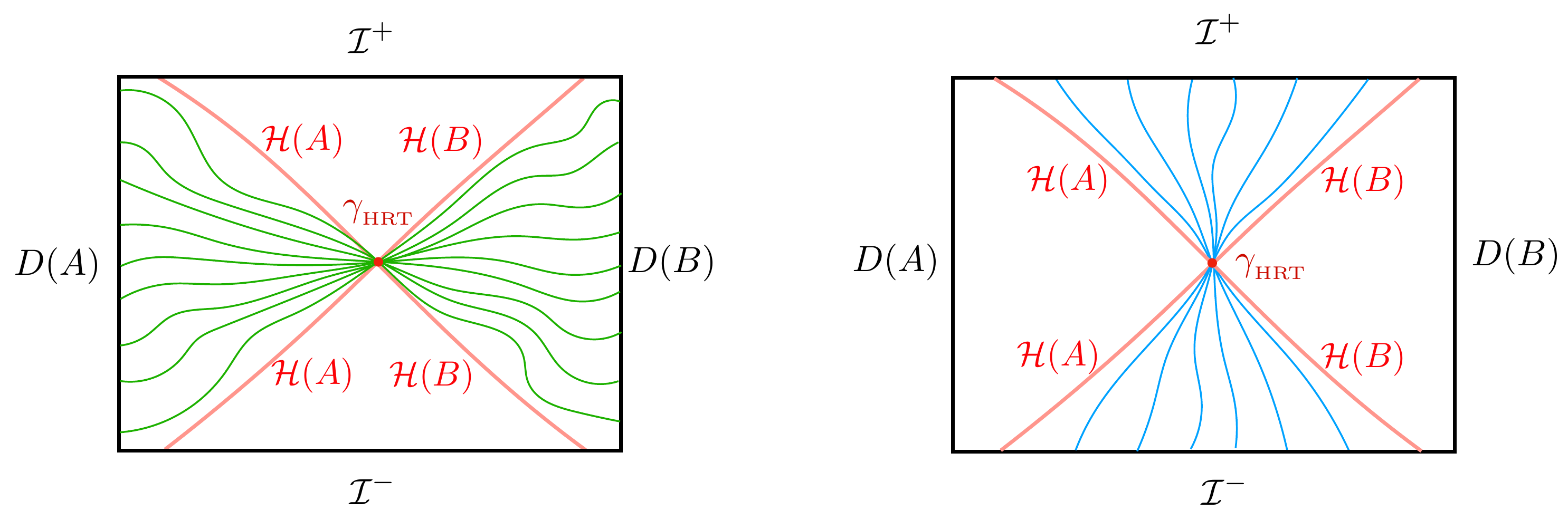}
\caption{\label{fig:maxVUflow}
Cross section of generic max V-flow \figL\ and min U-flow \figR. The max V-flow squeezes through the HRT surface, being excluded from its past and future. The min U-flow also squeezes through the HRT surface, being excluded from the entanglement wedges.  Notice that while the V-flow can have somewhere-timelike flow lines, the U-flow lines cannot be spacelike anywhere.}
\end{figure}

The function $\phi$ that witnesses the norm bound on $V$ is also strongly constrained. For a thin ribbon $\ts$ that follows a timelike curve starting at $x$ and ending at $y$, with constant spatial area $a$, the same reasoning that gave \eqref{fluxbound0} implies
\begin{equation}\label{newfluxbound}
\int_\T*V\le a\, (\phi(y)-\phi(x))\,.
\end{equation}
If the ribbon passes through $\HRT$, spatially parallel to it, then the flux through it equals $a$, so $\phi(y)-\phi(x)\ge1$. But given that $|\phi|\le1/2$ everywhere in $\bar\M$, this implies
\begin{equation}\label{phisol}
    \phi|_{I^\pm(\HRT)}=\pm\frac{1}{2}\,.
\end{equation}
In other words, all of the level sets of $\phi$ pass through $\HRT$. Given the condition $d\phi\pm V\in \fdc$, $V$ must vanish in any region where $\phi$ is constant, so \eqref{phisol} implies
\begin{equation}\label{Vsol2}
    V|_{I^\pm(\HRT)}=0\,.
\end{equation}
Although we already showed that none of the flux from $D(A)$ to $D(B)$ passes through $I^\pm(\HRT)$ (since it all passes through $\HRT$), \eqref{phisol} implies that there can be no stray flux circulating in closed loops in those regions. \eqref{phisol}, \eqref{Vsol2} only constrain $V$ and $\phi$ in the regions $I^\pm(\HRT)$. In the complementary regions $\ew(A)$, $\ew(B)$, there remains enormous freedom in the solution, including not only how the flux is spread out in space --- as for Riemannian flows (or bit threads) --- but also how it is spread out in time. See figure \ref{fig:maxVUflow} for a picture of a ``typical'' max V-flow.

We want to emphasize that \eqref{phisol}, \eqref{Vsol2} are properties of \emph{solutions} to the V-flow program. There are many perfectly acceptable V-flows that do not obey those equations --- but they will not have maximal flux. The V-flow program, which in its definition knows nothing about $\HRT$, in maximizing the flux of $V$ naturally discovers not only the HRT surface $\HRT$ but also the entanglement wedges $\ew(A)$, $\ew(B)$.\footnote{\label{foot:shadows}\, 
In fact, in this regard, the full collection of maximizing flows does better than the collection of HRT surfaces anchored on all possible subregions: The latter, being smooth, generically can't probe arbitrarily close to the crossover seams typically present on the entanglement horizon. Hence only a subset of the entanglement wedge would be reachable by HRT surfaces, whereas the entire entanglement wedge is reachable by maximizing flows.
An avatar of this was seen already in the Riemannian context \cite{Freedman:2016zud}: the bit threads can penetrate `entanglement shadows', namely bulk regions inaccessible by RT surfaces.}

Optimal U-flows, which minimize the flux of $U$ through $\I^+$, have a similar structure, but ``turned on the side''.\footnote{\, 
One additional difference being that while the V-flows can have timelike parts, the U-flows cannot have spacelike parts, so in this sense the floppiness of U-flows is slightly more rigidified.}
Given a slice $\bkslice_-$ on which $\HRT$ is the unique minimal-area surface, a minimal U-flow $U$ saturates \eqref{Ufluxbound}, implying that the flux $*U|_{\bkslice_-}$ is a delta-function supported on $\HRT$. The conditions \eqref{Uflowdef1} and $U\in \fdc$ then imply that $U$ vanishes everywhere in the causal domain of $\bkslice_-\setminus\HRT$, which is to say the union of the entanglement wedges $\ew(A)$, $\ew(B)$, or more precisely their interiors:
\begin{equation}\label{minUflowcond1}
    U|_{\text{int}\ew(A)\cup \text{int}\ew(B)}=0\,.
 \end{equation}
At any point where $U=0$, $d\psi=0$ as well, so $\psi$ is constant on $\ew(A)$ and $\ew(B)$. Given the boundary conditions on $\psi$, we have
\begin{equation}\label{minUflowcond2}
\psi|_{\ew(A)}=-\frac{1}{2}\,,\qquad\psi|_{\ew(B)}=\frac{1}{2}\,.
\end{equation}
See figure \ref{fig:maxVUflow} for a typical min U-flow. As for V-flows, the U-flow program does not know, in its definition, where the entanglement wedges are, and a non-minimizing U-flow need not obey \eqref{minUflowcond1}, \eqref{minUflowcond2}. Rather, in minimizing the flux of $U$ the U-flow program naturally discovers the HRT surface and entanglement wedges.

Among the optimal V-flows are slice-flows, where the slice is any maximin slice. Among the optimal U-flows are time-sheet-flows, where the time-sheet for example is the entanglement horizon $\hor(A)$.

\subsection{Multiple regions}
\label{sec:multiple}

So far, in most of the paper (except for section \ref{sec:subadditivity}), we have for clarity considered a single boundary region $A$ and its complement $B:=A^c$.  However, as explained in section \ref{sec:subadditivity}, all our constructions can be generalized to the case of multiple regions of interest $A,B,C,\ldots$ sharing a boundary Cauchy slice. We assume these regions are separated by buffers, and take our regulated spacetime $\M$ to be the joint entanglement wedge of all these regions $\W(ABC\cdots)$. We also define the complement of any given region to consist of the union of all the other ones, $A^c=BC\cdots$ etc.  In this context it is interesting to ask when can we utilize a single optimal flow configuration --- either U- or V-flow --- for multiple regions, and what would such a flow look like.

Let us first recall the Riemannian case.  In the RT (min cut) formulation, the inclusion-exclusion argument of \cite{Headrick:2013zda} shows that homology regions are monotonic under inclusion, hence that RT surfaces for nested regions (say $A$ and $AB$) do not intersect transversely (though they can be coincident on connected components). As shown in \cite{Headrick:2017ucz}, strong duality relates the nesting property for homology regions to nesting property for flows.  In the latter case, this states that there exists a flow $v$ that simultaneously maximizes on any collection of nested regions.  On the other hand, when RT surfaces do intersect transversely (say for regions $AB$ and $BC$), then it is clear that there cannot exist any flow simultaneously maximizing on both crossing regions.  This is because at the intersection of the RT surfaces, the normals point in different directions, while a single maximizing flow has to be of unit norm and normal to the RT surfaces everywhere.

\begin{figure}[tbp]
\centering
\includegraphics[width=0.9\textwidth]{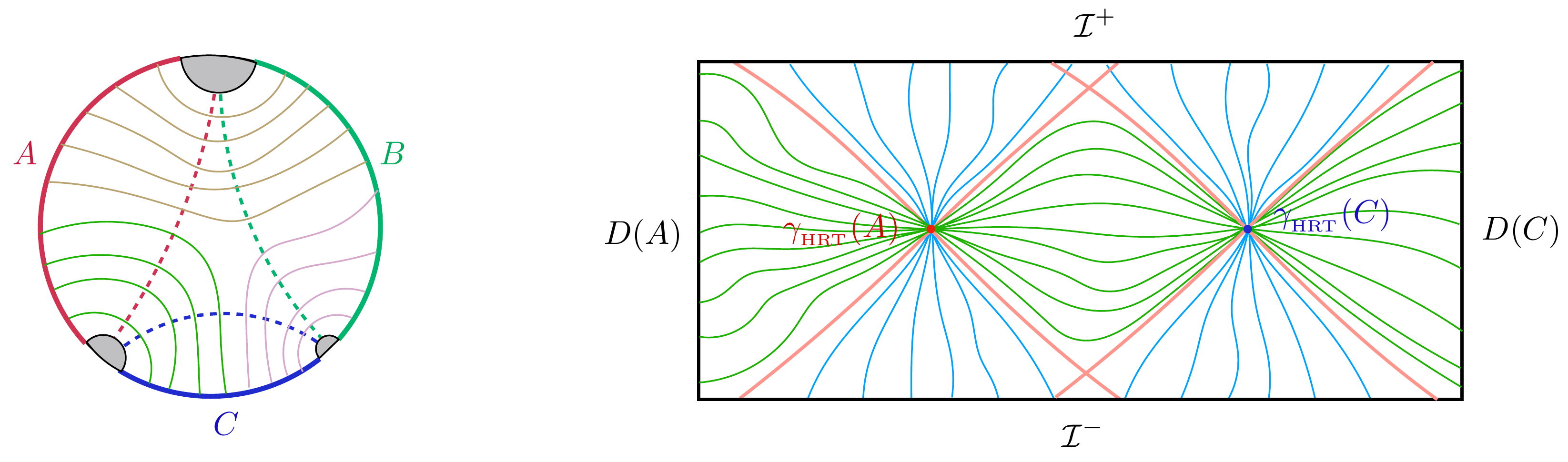}
% \includegraphics[width=0.25\textwidth]{figs/nesting_Riem.pdf}
% \hspace{1.5cm}
% \includegraphics[width=0.6\textwidth]{figs/nesting_xs.pdf}
\caption{\label{fig:nestingxs}
\figL:
Illustration of a max V-flow for nested regions $A$ and $AB$ on the Poincare disk.  The dashed curves indicate the HRT surfaces for the respective regions, and the solid thin curves indicate the V-flow lines, color-coded by which regions they connect.  
\figR:
To illustrate the temporal behavior of the flow lines for a generic max V-flow (green) as well as a generic min U-flow (blue), we take a cross-section of the left diagram between $A$ and $C$, so that only the $A-C$ V-flow lines (green in the left figure) are visible.
}
\end{figure}

Let us now return to the Lorentzian case.  In the holographic context, entanglement wedge nesting applies whenever the NEC is satisfied \cite{Wall:2012uf}.  In other words, for nested boundary regions $A$ and $AB$, the entanglement wedges are nested, which automatically implies that the respective HRT surfaces $\HRT(A)$ and $\HRT(AB)$ cannot be timelike-separated, and in fact lie on a common maximin slice. Hence there exists a V-flow configuration (namely a slice flow) which passes through both HRT surfaces. A generic maximizing V-flow would be dispersed everywhere outside the future and past of the two HRT surfaces $\HRT(A)$, $\HRT(AB)$ (and collimated through both of them), in other words within the following spacetime region:
\be
\W(A) \cup \W((AB)^c)\cup (W(AB)\cap W(A^c))\,.
\ee
Similarly, one can combine the minimizing U-flows in the complementary regions, likewise collimated through both HRT surfaces.  However, whereas a single V-flow line can pass through both HRT surfaces, a single U-flow line must pass through exactly one of them. This is illustrated in figure \ref{fig:nestingxs}.

Now what about crossing regions, e.g.\ $AB$ and $BC$?  Since HRT surfaces are bulk codimension-2, two HRT surfaces anchored on crossing regions will generically not intersect.  If they do, the above argument for the Riemannian case immediately generalizes to the present context.  However, when the HRT surfaces $\HRT(AB)$ and $\HRT(BC)$ do not intersect transversely, they must be timelike-separated somewhere, which means they would violate the V-norm bound.

%---------------------------------------------------

\section{Embedding in the full spacetime \& removing the regulator}

In subsection \ref{sec:EWCS}, we explained that, by focusing on the EWCS (either as a regulator for the EE or for its own sake), we could work entirely in the joint $AB$ entanglement wedge, discarding the rest of the original spacetime. But, in the other direction, we can also easily translate the various new formulas for the EWCS that we've developed in this paper back to the original spacetime. Here we will focus on three formulas in particular: minimax, max V-flow, and min U-flow; the others follow the same pattern, and we leave it to the reader to work out the details. We can even remove the regulator entirely, and define max V-flows and min U-flows when the EE is infinite.

\subsection{Embedding in the full spacetime} 
\label{sec:embedding}

\begin{figure}[tbp]
\centering
\includegraphics[width=0.45\textwidth]{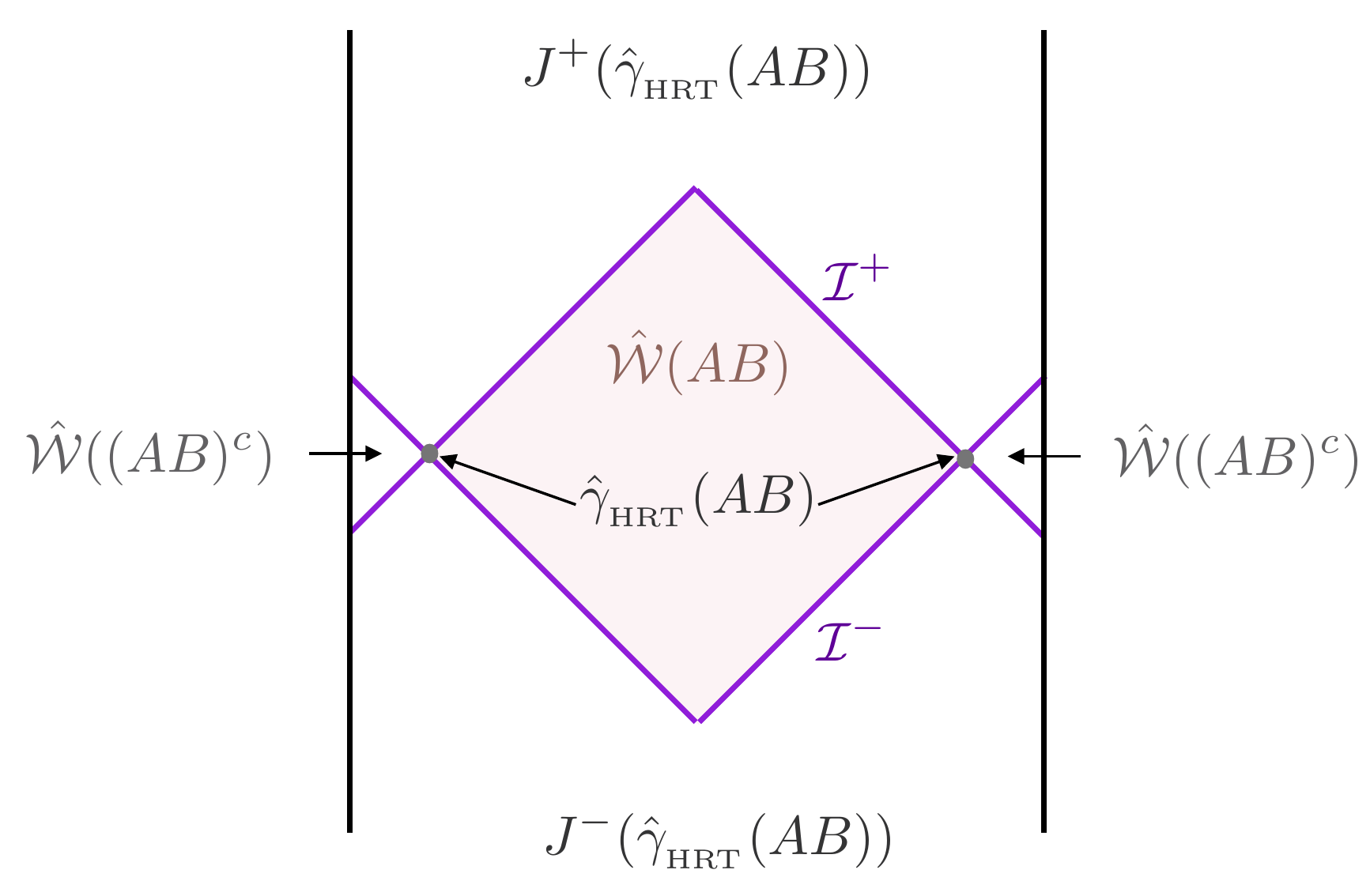}
\caption{\label{fig:embedST}
Schematic illustration of the full spacetime $\hat{\M}$, partitioned into 4 causal domains specified by $\hHRT(AB)$, so that the regulated spacetime (shaded) is $\M = \hat\W(AB)$.  This is to be visualized as a cross-section of AdS (cf.\ right panel of figure \ref{fig:ewcsreg}, but rotated by 90$^\circ$, with $A$ in front and $B$ in back).
}
\end{figure}

Let us first set up some notation for this purpose. Let $\hat\M$ be the original holographic spacetime containing $\M$, with future/past boundaries $\hat\I^\pm$ (which may be singularities or at infinite time) and timelike spatial boundary (end-of-the-world brane) $\hat\I^0$ (which may be empty); define $\hat\I:=\hat\I^+\cup\hat\I^-\cup\hat\I^0$. (In general we use hatted symbols to denote geometric constructs in the original, unregulated spacetime.) Let $\hHRT(AB)$ be the HRT surface in $\hat\M$ for $AB$. We assume that $\hat\M$ obeys the null energy condition and asymptotically AdS boundary conditions, so it is decomposed along the $AB$ entanglement horizons $\hat\hor(AB)$, $\hat\hor((AB)^c)$ into four spacetime regions: the entanglement wedge of $AB$, which is $\M$; the complementary entanglement wedge $\hat{\W}((AB)^c)$; and $J^\pm(\hHRT(AB))$. (Cf.\ figure \ref{fig:embedST}.)
The future/past boundary $\I^\pm$ of $\bar\M$ is made up of the future/past entanglement horizons $\hat\hor^\pm(AB)$ possibly along with part or all of $\hat\I^\pm$.

Everything we discuss below can be done for multiple boundary regions $A,B,\ldots$, as discussed in subsection \ref{sec:multiple}, provided of course that they lie on a common boundary Cauchy slice. Then the joint entanglement wedge $\hat\W(AB)$ is replaced by the entanglement wedge of their union $\hat\W(AB\cdots)$, and similarly for their joint HRT surface.

\paragraph{Minimax}
Here we start from the minimax formula in the regulated spacetime \eqref{minimax4}, which we repeat here for convenience:
\begin{equation}\label{minimax5}
\Splus= 
\infp_{\ts\in\tsset}\ 
\sup_{\substack{\surf\subset\ts\\ \text{achronal}}}
\area(\surf)\,.
\end{equation}
Any time-sheet $\ts\in\tsset$ may be extended to a piecewise timelike hypersurface $\hat\ts$ in $\hat\M$ that is homologous to $D(A)$ relative to $\partial\hat\M\setminus(D(A)\cup D(B))$. Conversely, any such hypersurface, when restricted to $\M$, is in $\tsset$. We therefore define $\hat\tsset$ as the set of such hypersurfaces. Within such a hypersurface, we want to maximize the area of an achronal surface contained in $\W(AB)$:
\be\label{minimaxfullspacetime}
S_{\rm HRT} = \inf_{\hat\ts\in\hat\tsset}
\sup_{\substack{\surf\subset\ts\cap\W(AB)\\ \text{achronal}}}\area(\surf)\,.
\ee
(The condition that $\hat\ts$ is piecewise timelike outside of $\hat\W(AB)$ is actually unnecessary, since we are only measuring the area of a surface within $\hat\W(AB)$.)

\paragraph{V-flows} 
Starting with the definition \eqref{Vflowdef1}, \eqref{Vnormbound1} of a V-flow in the regulated spacetime, we can define a V-flow in $\hat\M$. Any slice for $\hat\M$ that contains $\hHRT(AB)$ becomes a slice for $\bar\M$ when restricted to it. We therefore define $\hat\bksliceset$ as the set of slices for $\hat\M$ containing $\hHRT(AB)$. Its relaxation $\hat\bksliceset_{\rm c}$ is then given by the following set of functions:  
\be\label{Schatdef}
\hat\bksliceset_{\rm c} := \bigg\{\phi:\hat\M\to\left[-\frac12,\frac12\right]
\quad \left| \quad
\phi|_{J^\pm(\hHRT(AB))}=\phi|_{\hat\I^\pm}=\pm\frac12\,,\,d\phi\in \fdc\right\} \ .
\ee
We can in fact be far more parsimonious in restricting $\phi$.
Note that, given the conditions $d\phi\in \fdc$ and $\phi|_{\hat\I^\pm}=\pm1/2$, the condition $\phi|_{J^\pm(\hHRT(AB))}=\pm1/2$ follows just from requiring that $\phi=-1/2$ infinitesimally to the past of $\hHRT(AB)$ and $\phi=+1/2$ infinitesimally to the future of $\hHRT(AB)$. 
We can then rewrite \eqref{Schatdef} as
\be\label{Schatdef2}
\hat\bksliceset_{\rm c}  = \bigg\{\phi:\hat\M\to\left[-\frac12,\frac12\right]
\quad \left| \quad
\phi|_{\HRT(AB)^\pm}=\phi|_{\hat\I^\pm}=\pm\frac12\,,\,d\phi\in \fdc\right\},
\ee
where $\hHRT(AB)^\pm$ are infinitesimally to the future and past of $\hHRT(AB)$.

We then define a V-flow on $\hat\M$ as a 1-form $V$ obeying
\be\label{Vflowdefhat}
    d{*V}=0\,,\qquad
    *V|_{\hat\I}=0
    \ee
    \be
    \label{Vnormboundhat}
\exists\,\phi\in \hat\bksliceset_{\rm c}
\text{ s.t. }        d\phi\pm V\in \fdc   \,.
\ee
Let $\hat\F$ be the set of such 1-forms.

A V-flow $V\in\F$ can be extended to one in $\hat\F$ simply by setting $V=0$ on $\hat\M\setminus\M$. The resulting 1-form on $\hat\M$ is divergenceless and obeys $*V|_{\hat\I^\pm}=0$ by virtue of the no-flux boundary condition on $\I^\pm$. The witness function $\phi\in\bksliceset_{\rm c}$ can similarly be extended to one in $\hat\bksliceset_{\rm c}$ by setting it equal to $\pm1/2$ on $J^\pm(\hHRT(AB))$ and to any function obeying $d\phi\in\fdc$ on the complementary entanglement wedge $\hat\W((AB)^c)$.

Conversely, any $V\in\hat\F$, when restricted to $\M$, gives a $V\in\F$. To show this, we restrict $\phi\in\hat\bksliceset_{\rm c}$ to $\M$ to obtain $\phi\in\bksliceset_{\rm c}$, and similarly restrict $V$ to $\bar\M$. The only thing that needs to be checked is the boundary condition $*V|_{\I^\pm}=0$. As mentioned above, $\I^\pm$ is made up of $\hat\hor(AB)$ together with (part or all of) $\hat\I^\pm$. We already have the boundary condition on $\hat\I^\pm$ from \eqref{Vflowdefhat}, and the boundary condition on $\hat\hor(AB)$ is enforced by the condition $d\phi\pm V\in \fdc$, given that $d \phi=0$ in $J^\pm(\hHRT(AB))$.

These mappings between $\F$ and $\hat\F$ do not change the flux on $D(A)$, so we have
\begin{equation}\label{Vfullspacetime}
\SHRT = \sup_{V\in\F}\ \int_{D(A)}*V= \sup_{V\in\hat\F}\ \int_{D(A)}*V\,.
\end{equation}

\paragraph{U-flows}
Starting from the definition \eqref{Uflowdef1}, \eqref{Uflowdef2} of a U-flow in the regulated spacetime, we can also easily define a U-flow on $\hat\M$. First we define the relaxed time-sheet set $\hat\tsset_{\rm c}$ as any function $\psi:\hat\M\to[-1/2,1/2]$ that, when restricted to $\bar\M$, is an element of $\tsset_{\rm c}$ (as defined in \eqref{psicond2}):
\begin{equation}\label{psicondhatM}
\hat\tsset_{\rm c}:=\bigg\{\psi:\hat\M\to\left[-\frac12,\frac12\right]
\quad \left| \quad
\psi|_{D(A)}=-\frac12\,,\,
\psi|_{D(B)}=\frac12
\right\}.
\end{equation}
Imposing the condition $U\pm d\psi\in \fdc$ along with the no-flux condition on $D(A)\cup D(B)\cup\hat\I^0$ gives us the definition of a U-flow:
\be\label{Uflowdefhat}
    d{*U}=0\,,\qquad
    *U|_{D(A)}=*U|_{D(B)}=*U|_{\hat\I^0}=0
    \ee
    \be
    \label{Unormboundhat}
\exists\,\psi\in \hat\tsset_{\rm c}
\text{ s.t. }        U\pm d\psi\in \fdc   \,.
\ee
Let $\hat\G$ be the set of U-flows.

Notice that there is no boundary condition on $\psi$ or $U$ except on $D(A)$, $D(B)$, and $\hat\I^0$. In particular, the flux of $U$ is free to enter or leave the spacetime on the rest of the boundary. This means that the total flux must be measured on a slice for $\M$, or on $\I^\pm$, not on a slice for $\hat\M$.

Any U-flow on $\hat\M$, $U\in\hat\G$, can be restricted to a U-flow on $\bar\M$. Conversely, a U-flow on $\bar\M$,  $U\in\G$, can be extended to one on $\hat\M$. There are many ways to carry out this extension consistent with the conditions $d*U=0$, $U\pm d\psi\in \fdc$. Here we will present one way that is relatively simple, although the resulting 1-form is singular on the boundary $\hat\hor(AB)$ between $\M$ and  $\hat\M\setminus\bar\M$. First, we set $\psi=0$ on $\hat\M\setminus\bar\M$. As a result, the only condition on $U$ in that region is $U\in \fdc$, which is trivial to satisfy simply by extending the field lines from $\bar\M$ along any set of non-intersecting causal curves. Call the resulting 1-form on $\hat\M$ $U_0$. By construction, $U_0$ is divergenceless and obeys $U_0\pm d\psi\in \fdc$ everywhere, except possibly on $\hat\hor(AB)$. By setting $\psi=0$ in $\hat\M\setminus\bar\M$, we have potentially introduced a discontinuity in $\psi$ on $\hat\hor(AB)$. Let $\lambda$ be a local coordinate in a neighborhood of a point on $\hat\hor(AB)$, such that $\lambda=0$ on $\hat\hor(AB)$; the discontinuity implies a delta-function in $d\psi$ proportional to $\delta(\lambda)d\lambda$. Since, on $\M$, $\psi\in[-1/2,1/2]$, the largest the discontinuity can be (in absolute value) is $1/2$. So by adding the 1-form $U_1:=\pm(1/2)\delta(\lambda)d\lambda$ (where the sign is determined by the condition $U_1\in \fdc$) to $U_0$, we are guaranteed to satisfy the condition $(U_0+U_1)\pm d\psi\in \fdc$. $U_1$, however, is not divergenceless, since the transverse area element is not constant along the null generators of $\hat\hor(AB)$, and those null generators begin (on $\hHRT(AB)$) and end (on caustics and crossover seams). Again, this can be fixed by adding causal field lines in $\hat\M\setminus\bar\M$ to $U_1$. Finally, we set $U=U_0+U_1$.

The important thing about the above construction is that the 1-form $U$ is unchanged in $\M$ itself. Therefore its flux across any slice $\bkslice$ for $\M$ is unchanged. So we have
\begin{equation}\label{Ufullspacetime}
\SHRT = \inf_{U\in\G}\ \int_{\bkslice}*U= \inf_{U\in\hat\G}\ \int_{\bkslice}*U\,.
\end{equation}

\subsection{Removing the regulator}
\label{sec:removing}

The quantity $\SHRT$ computed in \eqref{minimaxfullspacetime}, \eqref{Vfullspacetime}, \eqref{Ufullspacetime} is obviously most meaningful when it is finite, which is the case when either the EE is finite (e.g.\ in a multiboundary wormhole) or we are computing a (finite) EWCS. Nonetheless, it is possible to define a max V-flow and a min U-flow in the absence of a regulator, and these can be used to compute naturally finite quantities like the mutual information of separated boundary regions.

In the limit where $AB$ covers an entire boundary slice, in other words we remove the regulating buffer separating them, the joint HRT surface $\hHRT(AB)$ retreats to the boundary and becomes the mutual entangling surface $\partial A=\partial B$. This entangling surface partitions the conformal boundary into four regions: $D(A)$, $D(B)$, $J^\pm(\partial A)$. 

\paragraph{Minimax:} In this limit, the time-sheet set $\hat\tsset$ becomes the set of piecewise-timelike hypersurfaces in $\hat\M$ homologous to $D(A)$ relative to $\hat\I^+\cup\hat\I^-\cup\hat\I^0\cup J^+(\partial A)\cup J^-(\partial A)$. Note that this homology condition forces the time-sheet $\hat\ts$ to contain $\partial A$. Furthermore, in order to ensure that the surface $\surf\subset\hat\ts$ is contained in the $AB$ entanglement wedge --- i.e.\ is excluded from the \emph{bulk} chronal future and past of $\partial A$ --- we require it to include $\partial A$:
\be
S_{\rm HRT} = \inf_{\hat\ts\in\hat\tsset}
\sup_{\substack{\partial A\subset\surf\subset\hat\ts \\ \text{achronal}}}\area(\surf)\,.
\ee
Of course, for any $\hat\ts$, the area of $\surf$ is unbounded above, so the minimax is infinite.

\paragraph{V-flows:} In this limit, the boundary condition on $\phi$ becomes $\phi|_{\hat I^\pm}=\phi|_{\partial A^\pm}=\pm1/2$:
\be\label{Vflownoreg}
\hat\bksliceset_{\rm c} := \bigg\{\phi:\hat\M\to\left[-\frac12,\frac12\right]
\quad \left| \quad
\phi|_{\hat I^\pm}=\phi|_{\partial A^\pm}=\pm\frac12\,,\,d\phi\in \fdc\right\}.
\ee
In view of the condition $d\phi\in \fdc$, the condition $\phi|_{\partial A^\pm}=\pm1/2$ is equivalent to $\phi|_{J^\pm(\partial A)}=\pm1/2$. The rest of the V-flow definition \eqref{Vflowdefhat}, \eqref{Vnormboundhat} remains unchanged. 

For multiple boundary regions, the entangling surface $\partial A$ in \eqref{Vflownoreg} is replaced by the union of all the entangling surfaces. Thus, as in the regulated spacetime, the definition of a V-flow is uniform for all the regions and groupings thereof.

Despite the fact that $\SHRT$, and therefore the maximal flux of $V$, is infinite, it is nonetheless possible to define the notion of a \emph{max V-flow}. We say that a 1-form $\Delta V$ is an \emph{augmentation} of the V-flow $V$ if it has positive flux, $\int_{D(A)}*\Delta V>0$, and $V+\Delta V$ is a V-flow. $V$ is a max V-flow if it does not admit an augmentation. Clearly when the flux is finite this definition agrees with the usual one --- a V-flow with maximal flux --- but it also makes sense when the flux is infinite. We can use it, for example to define the mutual information between separated regions, without ever introducing a regulator at an intermediate step; see \cite{Freedman:2016zud}.

\paragraph{U-flows:} The U-flow definition \eqref{psicondhatM}, \eqref{Uflowdefhat}, \eqref{Unormboundhat} is unchanged in this limit. We can measure the flux of $U$ through any slice $\bkslice$ for $\hat\M$ anchored on the entangling surface $\partial A$. Note that, in this limit, the no-flux boundary conditions $*U|_{D(A)}=*U|_{D(B)}$ does not preclude flux from entering or leaving the bulk through $\partial A$ itself; in particular, for a delta-function valued 1-form such as $U_1$, a finite amount of flux may enter or leave through $\partial A$, and this would not be counted in the flux through $\bkslice$. Since, with the regulator removed, $\SHRT$ diverges, every U-flow must have infinite flux. As with the V-flows, we can nonetheless define the notion of a \emph{min U-flow} as a U-flow $U$ that does not admit a \emph{diminution}, i.e.\ a 1-form $\Delta U$ such that $U+\Delta U$ is a U-flow and $\int_\bkslice*U<0$.

\section{Discussion \& future directions}
\label{sec:discussion}

In this paper, we have derived several new dual formulations of the HRT holographic EE formula. Here we will briefly discuss their notable features and possible applications, as well as possible avenues for further investigation.

First, the V-flow (or V-thread) and U-flow (or U-thread) formulas are ``fully covariant'', in the sense of depending explicitly only on the boundary domain of dependence $D(A)$, and making no reference at any step to any boundary or bulk slices or other extraneous geometrical structure. They automatically and naturally incorporate the spacelike-homology condition on the HRT surface, and both the surface and entanglement wedge fall out naturally from the optimization.

It is also notable that, unlike the original HRT and maximin formulas, these V- and U-flow formulas involve pure maximization and minimization respectively. Furthermore, they define convex programs. Convexity affords a tremendous technical advantage for many purposes. For example, it is plausible that these programs may be the basis for efficient numerical methods for computing EEs in time-dependent backgrounds. Possessing a dual pair of convex programs is even better. Whether the convexity reflects something physical about holographic entanglement, or is merely a useful rewriting of a fundamentally non-convex formula, remains to be seen.

Another possible set of applications of the new flow formulas, as well as the new minimax formula, has to do with properties of holographic EEs. For example, holographic entropies in the static setting are subject to a known, infinite set of inequalities, going beyond those obeyed by general quantum states and defining the so-called RT entropy cone \cite{Bao:2015bfa}. One of these (monogamy of mutual information \cite{Hayden:2011ag}) has been proven to hold covariantly using the maximin formula \cite{Wall:2012uf}. The rest have been proven to hold covariantly only in the case of $2+1$ dimensional bulk \cite{Czech:2019lps}, although the proto-entropy formulation developed in \cite{Hubeny:2018trv,Hubeny:2018ijt} strongly suggests that they should hold universally. The new formulas may be useful for establishing these and other properties of holographic EEs.

More generally, it would be interesting to explore the conceptual implications of the new formulas, and whether they contain some message about the relation between entanglement and spacetime geometry. As discussed in the introduction, there is a sense in which the V-threads may be thought of as Bell pairs in a distillation of the state. Conversely, the U-threads may in some sense be disentanglers; in this view, the EE counts the minimum number of distentanglers required to completely factorize the state, or equivalently the number of operations required to form the state out of an unentangled state (in the sense of entanglement of formation).

Another possible conceptual application concerns tensor network toy models of holography; since the Riemannian bit threads naturally live on a standard tensor network (see e.g.\ \cite{Chen:2018ywy}), the covariant threads, both V and U, may help in understanding how to incorporate time into tensor networks. Similar comments apply to the relation between threads and holographic quantum error correcting codes; see for example \cite{Harlow:2016vwg}.

A more technical application would be to the problem of bulk reconstruction. Specifically, it seems likely that the bulk metric can be reconstructed from the set of all allowed thread configurations. (This was explored in the static case in \cite{Freedman:2016zud}.) An advantage over the HRT surfaces is that the threads probe so-called entanglement shadows, spacetime regions not touched by any HRT surface. Furthermore, it should be possible to understand the relation between entanglement and the Einstein equation \cite{Lashkari:2013koa,Faulkner:2013ica,Swingle:2014uza} in the language of threads, as was done for Riemannian threads in \cite{Agon:2020mvu}. (See footnotes \ref{foot:reconstruction}, \ref{foot:shadows} for related comments.)

It is curious that even in the generalized (non-holographic) setting, subadditivity was a property of the convex-relaxed quantity $\Sc$, despite the setup not affording a theory wherein to identify this quantity with an EE.  A natural question is whether strong subadditivity, and even the higher holographic inequalities, are likewise obeyed by $\Sc$. 
This is closely related to the question of whether the V- and U-flows obey the nesting property and admit max/min multiflows.
Since our arguments for subadditivity crucially relied on convexity, it would also be interesting to explore what inequalities, if any, are obeyed by the non-convex maximin $\Sminus$ and minimax $\Splus$ quantities (again, in the non-holographic setting where $\Sc$, $\Splus$, $\Sminus$ are all distinct).

There are several clear directions for generalizing the work in this paper. Generalizations of Riemannian bit threads to include bulk quantum and higher-curvature corrections have already been explored \cite{Harper:2018sdd,Rolph:2021hgz,Agon:2021tia}; it is natural to do so also for the covariant threads. The quantum corrections in particular should provide a thread picture for the Page curve of an evaporating black hole derived from the quantum extremal surface formula \cite{Engelhardt:2014gca,Almheiri:2019psf,Penington:2019npb}. Another direction would be to covariantize the various weaker norm bounds for multiflows considered in \cite{Headrick:2020gyq}, as well as thread and hyperthread constructions for multipartite entanglement measures \cite{Harper:2019lff,Harper:2020wad,Lin:2020yzf,Lin:2021hqs,Harper:2021uuq,Harper:2022sky,Lin:2022aqf}. Finally, covariantizing the bit threads opens the door to applying them in non-AdS backgrounds, such as asymptotically flat, de Sitter, and other cosmological spacetimes.\footnote{\, For some work already done in this direction, see \cite{Shaghoulian:2022fop}.}

%---------------------------------------------------

\acknowledgments

We would like to thank the following individuals for helpful conversations and feedback on the manuscript: C. Agon, N. Engelhardt, M. Freedman, G. Grimaldi, J. Harper, J. Pedraza, M. Rangamani, A. Rolph, M. Rota, B. Stoica, H. Verlinde, E. Witten.

We would also like to thank UC Davis, the Simons Center for Geometry and Physics, the Aspen Center for Physics, the Kavli Institute for Theoretical Physics, the Perimeter Institute for Theoretical Physics, the Galileo Galilei Institute, and the MIT Center for Theoretical Physics for hospitality while this research was being carried out.

MH is supported in part by the Simons Foundation through the \emph{It from Qubit} Collaboration and by the U.S. Department of Energy grant DE-SC0009986. VH has been supported in part by the U.S. Department of Energy grant DE-SC0009999. Both authors are also supported by the U.S. Department of Energy grant DE-SC0020360 under the HEP-QIS QuantISED program.

%---------------------------------------------------

\appendix

\section{Table of notation}
\label{sec:notation}

To assist the reader we provide here a table of the various symbols we use, their meanings, and the subsections where they are first discussed (starting in section \ref{sec:background}). For ease of orientation, whenever practical, we try to distinguish different classes of objects by different fonts.  In particular, we reserve the mathcal font ($\I,\M,\N,\ldots$) for important geometrical constructs and the mathscript font (e.g.\ $\bksliceset,\tsset,\ldots$) for sets of objects (slices, flows, etc.).  We sequester special fonts for spacetime curves, and for their measures. We also use hatted symbols to refer to constructs in the original, unregulated spacetime ($\hat\M$, $\hat\I^\pm$, etc.), and unhatted symbols to refer to constructs in the regulated spacetime.

\begin{center}
\begin{longtable}{r|l|l}
symbol & meaning & subsection\\
\hline\hline
\endhead
$\hat\M$ & original, unregulated spacetime  & \ref{sec:EWCS}\\
$\hat\I^\pm$ & future (past) boundary of $\hat\M$ & \\ 
$\hHRT(AB)$ & HRT surface of $AB$ in $\hat\M$ & \\
$\hat{\W}(AB)$ & entanglement wedge of $AB$ in $\hat\M$ & \\
$\M$, $\hat\W(AB)$     & (regulated) bulk spacetime ($AB$ entanglement wedge) & \\
$\N$     & conformal boundary of $\M$ ($D(AB)$)\\
$\I^\pm$ & future/past boundary of $\M$ \\
$\bdyslice$ &  slice for $\N$ \\
\hline
$\I^0$ & timelike boundary of $\M$ at finite distance &\ref{sec:setup} \\& \quad(end-of-the-world brane and/or $\hHRT(AB)$)   \\ 
$\I$ & $\I^+\cup\I^0\cup\I^-$ \\
$\bar\M$ & $\M\cup\N\cup\I$ \\
$\bkslice$ & (Cauchy) slice for $\bar\M$ \\
$\eowsurf_\bkslice$ & $\I^0\cap\bkslice$ \\
$\bksliceset$ & set of all  slices for $\bar\M$ \\
$A,B,(C)$ &   disjoint regions covering $\bdyslice$ \\
$D(A)$  & boundary domain of dependence of $A$ \\
$A_\bkslice$ & $D(A)\cap\bkslice$ \\
$\ts$ &  time-sheet (piecewise timelike hypersurface in $\bar\M$) \\
$\Gamma_\ts$ & set of all surfaces of the form $\bkslice\cap\ts$ for some $\bkslice\in\bksliceset$ \\
$\tsset$, $\tsset_A$ & set of all time-sheets homologous to $D(A)$ relative to $\I$ \\ \hline
$N$, $T$, $U$, \ldots &covector (field) on $\bar\M$  &\ref{sec:pointwise}\\
$T^*_x$, $T^*$ & cotangent space at a given point in $\M$ \\
$\fdc$ & set of future-directed causal covectors at a given point in $\M$ \\
$\fdt$ 
& set of future-directed timelike covectors at a given point in $\M$  \\
$\ip WX$ & $\max\{|W\cdot X|,|W\wedge X|\}$ (which we call the wedgedot)
\\ \hline
$\surf$ &   codimension-2 achronal surface in $\bar\M$ &\ref{sec:relaxation}\\ 
 $\Gamma_\bkslice$, $\Gamma_{\bkslice,A}$ & set of all surfaces in $\bkslice$ homologous to $A_\bkslice$ relative to $\surf^0_\bkslice$ \\ 
$S_-$, $S_-(A:B)$ & maximin area \\
\hline
$v$ & 1-form (covector field) on $\bkslice$ &\ref{sec:maximin}\\
$\F_\bkslice$ & set of all flows on $\bkslice$\\ 
$S_+$ & minimax area &
\\
\hline
$\phi$, $\psi$ & scalar function on $\bar\M$ & \ref{sec:convexrelaxation}\\
$\bkslice_t$ & level set of $\phi$ where $\phi=t$ \\
$\ts_s$ & level set of $\psi$ where $\psi=s$ \\
$\bksliceset_{\rm c}$ & relaxed set of slices \\
$\tsset_{\rm c}$ & relaxed set of time-sheets \\
$S_{\rm c}$ & convex maximin/minimax \\ \hline
$\F$ & set of all V-flows & \ref{sec:Vflows} \\ \hline
$\qcv$ & inextendible causal curve in $\M$; U-thread & \ref{sec:Vflowbounds} \\
$\Qset$ & set of all inextendible causal curves in $\M$ \\ \hline
$\G$ & set of all U-flows & \ref{sec:Uflows}\\ \hline
$\pcv$ & curve in $\M$ connecting $D(A)$ and $D(B)$; V-thread & \ref{sec:Uflowbounds} \\
$\Pset$ & set of all curves in $\M$ connecting $D(A)$ and $D(B)$ \\ \hline
$\Kset$ & set of pairs $(X,Y)$ of covectors at a given point in $\M$  & \ref{sec:lemmaproofs} \\
& \qquad such that $X\pm Y\in \fdc$ \\
\hline
$\muth$ & measure on $\Pset$ & \ref{sec:VUthreads} \\
$\nuth$ & measure on $\Qset$ \\ \hline
$\Gamma$ & set of surfaces spacelike-homologous to $A$& \ref{sec:HRTsurface} \\
$\Gamma_{\rm ext}$ & set of extremal surfaces spacelike-homologous to $A$ \\
$r_\surf$ & homology region on a slice $\bkslice$ interpolating between $\surf$ and $A_\bkslice$ \\
$H_\surf$ & bulk part of boundary of $D(r_\surf)$ \\
$\HRT$, $\HRT(A)$ & HRT/EWCS surface for $A$  \\
$\SHRT$, $S(A:B)$ & area of $\HRT$ \\
$\ew(A)$ & $D(r_\HRT)$, entanglement wedge of $A$ \\
$\hor(A)$ & $H_\HRT$, entanglement horizon of $A$ \\
$\surf_-$ & maximin surface \\
$\surf_+$ & minimax surface \\ \hline
% $\hat\M$ & original, unregulated spacetime  & \ref{sec:embedding}\\
% $\hat\I^\pm$ & future (past) boundary of $\hat\M$ & \\ 
$\hat\I^0$ & end-of-the-world brane in $\hat\M$ & \ref{sec:embedding}\\
% $\hHRT(AB)$ & HRT surface of $AB$ in $\hat\M$ & \\
% $\hat{\W}(AB)$ & entanglement wedge of $AB$ in $\hat\M$ & \\
$\hat\hor(AB)$ 
& entanglement horizon of $AB$ in $\hat\M$ & \\
$\hat\tsset$ & set of all time-sheets homologous to $D(A)$ in $\hat\M$ & 
\\
$\hat\bksliceset_{\rm c}$ & relaxed set of slices for $\bar\M$ in $\hat\M$ & 
\\ $\hat\F$ & set of V-flows on $\hat\M$ \\
$\hat\tsset_{\rm c}$ & relaxed set of time-sheets for $\bar\M$ in $\hat\M$  & 
\\
$\hat\G$ & set of U-flows on $\hat\M$ 
\end{longtable}
\end{center}

%---------------------------------------------------

\section{Piecewise-linear scalars in worked example}
\label{app:piecewiselin}

In section \ref{sec:relaxation} 
we have introduced a toy model of a (non-holographic) spacetime which exemplifies that in general the maximin and minimax surface prescriptions need not coincide.  
The spacetime, illustrated in figure \ref{fig:gap},
consists of a $1+1$ dimensional patch of Minkowski spacetime times a piecewise-constant-area sphere.  This sphere has larger area $a_+$ within a timelike stripe across the spacetime and smaller area $a_-$ outside this stripe (as indicated by the coloring).  
Since all slices and timesheets encounter both the stripe and its exterior,  the area of minimal surface on each slice (and hence also when maximized over all slices) is $a_-$, while the area of a maximal surface on each timesheet (and hence also when minimized over all timesheets) is $a_+>a_-$.
On the other hand, in the fully convex-relaxed context, the minimax theorem ensures that the maximin and minimax values do coincide.
To get a better sense of the mechanism of how this happens, 
in section \ref{sec:flows}
we considered partially-relaxed V-flow and U-flow programs.  In particular, for the V-flow, we fixed a field $\phi$ growing linearly in time (corresponding to uniformly-smeared slice), and found the minimal flux of 1-form $V$ is larger than $a_-$; cf.\ \eqref{flux}.
Similarly for the U-flow, we fixed a uniformly-smeared timesheet given by a field $\psi $ growing linearly in space, and found the maximal flux of 1-form $U$ is smaller than $a_+$ (but still larger than the V-flux); cf.\ \eqref{Uhatflux}.
Even though the 1-form fields $d\phi$ and $d\psi$ were constant throughout the spacetime, the corresponding optimized V-flow and U-flow covector flow lines were refracted through the stripe.  

This suggests a natural generalization of the setup, with the scalar fields $\phi$ and $\psi$ required to be only piecewise-linear, with their gradient flow lines likewise allowed to be refracted through the stripe.  In this appendix we provide the details of this generalization, wherein we keep the gradient values outside the stripe as before, but allow for a transverse component inside the stripe.  As expected, this brings the optimized U and V fluxes still closer together, though not yet fully coincident.  One reason one should not expect coincident values a-priori is that the V-flux has knowledge of $T$ but not $L$, while the U-flux has knowledge of $L$ but not $T$.  Nevertheless, it is intriguing to see that the actual expressions, while more complicated than for the fully linear case, simplify substantially once we optimize over the slope of the refracted parts.\footnote{\, 
	A further natural generalization would be to allow an independent transverse component in each region; given the nature of the geometry, one might expect that this would in fact suffice to bring the minimax and maximin values to coincide, but we leave this as an exercise for the reader.
}  

For ease of comparing the V-flow and U-flow programs, we will carry out the computations in parallel.  We will adopt the notation of adding a tilde, on the U-flow side, to the corresponding (undecorated) quantities  on the V-flow side.  We will also use a more compact notation for the spacetime coordinates, denoting $x^0 := t$ and $x^1 := x$.  For the piecewise-linear case, we will have one extra parameter $\zV$ for the V-flow and one extra parameter $\zU$ for the U-flow compared to the linear case.  These parameters correspond to the slope and inverse slope of the refracted $\phi$ and $\psi$ level sets, respectively, such that $\zV=0$ and $\zU=0$ corresponds to the original unrefracted level sets.

The piecewise linear scalar fields are chosen so as to be continuous, satisfy the requisite boundary conditions,\footnote{\, 
	There is also a delta-function contribution at $\N$ in the U-flow case which we will deal with separately.
}
and have parallel gradients with the requisite boundary outside the stripe.  Defining the quantities
\begin{equation}\label{eq:Qdef}
	\QV \equiv \frac{1}{1- \zV \, \beta \, \left( 1 - \frac{T'}{T} \right)}
	\qquad \text{and} \qquad
	\QU \equiv \frac{1}{1- \frac{\zU}{\beta} \, \left( 1 -  \frac{\beta \,T'}{L} \right)}
\end{equation}	
we can write the scalar fields as
\begin{equation}\label{eq:phidef}
\phi(t,x) =
\begin{cases}
	\phi_u = \QV \, \left[ \frac{1}{2} \, \zV \, \beta \, \frac{T'}{T} + (1-\zV \, \beta)  \frac{t}{T} \right]
		\qquad&(\text{upper } a_-\text{ region}) \\
	\phi_s = \frac{\QV}{T}  \, \left[ t -  \zV \, x \right]
		\qquad&(\text{stripe } a_+\text{ region}) \\
	\phi_l =  \QV \, \left[ - \frac{1}{2} \, \zV \, \beta \, \frac{T'}{T} + (1-\zV \, \beta)  \frac{t}{T} \right]
		\qquad&(\text{lower } a_-\text{ region}) 
\end{cases}
\end{equation}	
and 
\begin{equation}\label{eq:psidef}
\psi(t,x) =
\begin{cases}
	\psi_u = \QU \, \left[ - \frac{1}{2} \, \zU \, \frac{T'}{L} + \left( 1 - \frac{\zU}{\beta} \right)  \frac{x}{L} \right]
		\qquad&(\text{upper } a_-\text{ region}) \\
	\psi_s = \frac{\QU}{L}  \, \left[- \zU \,  t + x \right]
		\qquad&(\text{stripe } a_+\text{ region}) \\
	\psi_l =  \QU \, \left[ \frac{1}{2} \, \zU \, \frac{T'}{L} + \left( 1 - \frac{\zU}{\beta} \right)  \frac{x}{L} \right]
		\qquad&(\text{upper } a_-\text{ region})
\end{cases}
\end{equation}	
These give the gradients
\begin{equation}\label{eq:dphidpsi}
d\phi = 
\begin{cases}
	\frac{\QV}{T}  \, \left( 1-\zV \, \beta \right) \, dt \\
	\frac{\QV}{T}  \, \left[ dt -  \zV \, dx \right]
\end{cases}
\ , \qquad
d\psi =
\begin{cases}
	\frac{\QU}{L}  \, \left( 1 - \frac{\zU}{\beta} \right) \, dx
		\qquad \qquad&(a_-\text{ region}) \\
	\frac{\QU}{L}  \, \left[- \zU \,  dt + dx \right]
		\qquad \qquad&(a_+ \text{ region}) 
\end{cases}
\end{equation}	
Note that in the limit $\zV,\zU \to 0$, we recover the constant gradients $d\phi = \frac{1}{T} \, dt$ and $d\psi = \frac{1}{L} \, dx$.

We first fix the $a_-$ region contribution to the V and U flows using the  gradients $d\phi$ and $d\psi$ in this region, and then apply the divergence-free condition to fix the part in the strip.  This generalizes the quantities $\Vzero$ and $\Uzero$ of equations \eqref{V0def} and  \eqref{U0def}  respectively:
\begin{equation}\label{eq:VzUz}
\Vzero = 
\begin{cases}
	\frac{\QV}{T}  \, \left( 1-\zV \, \beta \right) \, dx \\
	\frac{a_-}{a_+}  \, \frac{\QV}{T}  \, \left( 1-\zV \, \beta \right) \, dx
\end{cases}
\ , \qquad
\Uzero =
\begin{cases}
	\frac{\QU}{L}  \, \left( 1 - \frac{\zU}{\beta} \right) \, dt
		\qquad \qquad&(a_-\text{ region}) \\
	\frac{a_-}{a_+}  \, \frac{\QU}{L}  \, \left( 1 - \frac{\zU}{\beta} \right) \, dt 
		\qquad \qquad&(a_+ \text{ region}) 
\end{cases}
\end{equation}	
As for the linear case, we now introduce extra flow parallel with the stripe so as to satisfy the requisite conditions on V-flows and U-flows.  
In particular, letting
\begin{equation}\label{eq:VoUo}
\Vone = 
\begin{cases}
	0 \\
	\aV \, (-dt+\beta \, dx)\
\end{cases}
\ , \qquad
\Uone =
\begin{cases}
	0
		\qquad \qquad&(a_-\text{ region}) \\
	\aU \, (dt-\beta \, dx)\
		\qquad \qquad&(a_+ \text{ region}) 
\end{cases}
\end{equation}	
we want to find $\aV, \aU$
so as to ensure respectively that the total V-flux $V=\Vzero+\Vone$  is maximized subject to $d\phi \pm V \in \fdc$, and that  the total U-flux $U=\Uzero+\Uone$  is minimized subject to $U \pm d\psi \in \fdc$.\footnote{\, 
	For both the V-flow and the U-flow the divergencefree condition was already  implemented by the relative coefficient in \eqref{eq:VzUz} and the fact that the flow in \eqref{eq:VoUo} remains parallel to the stripe.  For the V-flow this also satisfies the boundary condition $*V|_{\I}=0$ of \eqref{Vflowdef1} since the stripe doesn't reach $\I$, while for the U-flow, we need to divert the stripe flow $\Uone$ along $\N$ in order to satisfy the corresponding boundary condition $*U|_{\I^0\cup\N}=0$ of \eqref{Uflowdef1}.
}
Both conditions only need to be ensured in the stripe, since outside the flows already saturate it by construction.
In the V-flow case, it turns out the stronger condition (the saturation of which determines $\aV$) is $d\phi + V \in \fdc$, 
while in the U-flow case the stronger condition (whose saturation gives $\aU$) is $U-d\psi \in \fdc$.  
These respectively give
\begin{equation}\label{eq:alphaVU}
\aV = \frac{\QV}{T} \, \frac{1}{1+\beta} \left[ 1 + \zV - \frac{a_-}{a_+} \, (1- \zV \, \beta) \right]
	\qquad \text{and} \qquad
\aU = \frac{\QU}{L} \, \frac{1}{1-\beta} \left[ 1 - \zU - \frac{a_-}{a_+} \, \left(1- \frac{\zU}{\beta}\right) \right]
\end{equation}
The total flux is obtained by integrating the sum of \eqref{eq:VzUz} and \eqref{eq:VoUo} with \eqref{eq:alphaVU}.
The net V-flux is then
\begin{equation}\label{eq:Vfluxz}
\int_{D(A)}*V =
 \QV \, \left[ a_- \, (1- \zV \, \beta) + \frac{\beta}{1+\beta} \, \frac{T'}{T} \, \left[ (a_+ - a_-) + \zV \, ( a_+ + \beta \,  a_- \right] \right]
\end{equation}
which reduces to $\eqref{flux}$ when $\zV=0$, and increases monotonically with $\zV$.  Similarly, the net U-flux is
\begin{equation}\label{eq:Ufluxz}
\int_{\I^+}*U=
 \QU \, \left[ a_- \left(1- \frac{\zU}{\beta}\right) + \frac{\beta}{1-\beta} \frac{T'}{L} \, \left[ (a_+ - a_-) -\frac{\zU}{\beta} \, (\beta \,   a_+ - a_- )\right] \right]
\end{equation}
which reduces to $\eqref{Uhatflux}$ when $\zU=0$, and likewise increases monotonically with $\zU$.  

This means that the maximal V-flux is reached at the maximal allowed value of $\zV$, namely when $d\phi + V$ becomes null, while the minimal U-flux is reached at the minimal allowed value of $\zU$, namely where $U-d\psi$ becomes null (which can happen only when $\beta \, a_+ < a_-$):
\begin{equation}\label{eq:zVUopt}
\zVm = \frac{a_- + \beta \, a_+}{\beta \, a_- + a_+} 
	\qquad \text{and} \qquad
\zUm = \frac{a_+ - \beta \, a_-}{\beta \, a_+ - a_-} 
\end{equation}
Note that $\zVm \in (\beta,1)$ (which ensures $d\phi \in \fdc$), with its
limits attained as $\frac{a_-}{a_+}\to 0$ and $1$, respectively,
while (in the $\beta \, a_+ < a_-$ regime) $\zUm<-1$.
 On the other hand, in the $\beta \, a_+ > a_-$ regime, the latter slope can get arbitrarily large and negative while $U-d\psi \in \fdt$.

Substituting $\zV=\zVm $ and $\zU=\zUm$ from \eqref{eq:zVUopt} into \eqref{eq:Vfluxz} and \eqref{eq:Ufluxz}, we find a simpler set of expressions:
\begin{equation}\label{eq:VUfluxopt}
		\int_{D(A)}*V =
					a_+ \, \frac{\beta \, \frac{T'}{T} \, a_+ + 
				\left[ 1 - \beta^2 \, \left( 1- \frac{T'}{T} \right) \right] \, a_-}
				{\beta \, \frac{T'}{T} \, a_- + 
				\left[ 1 - \beta^2 \, \left( 1- \frac{T'}{T} \right) \right] \, a_+} 
			\ , \qquad
		\int_{\I^+}*U = 
					a_+ \, \frac{\beta^2 \, \frac{T'}{L} \, a_+ + 
				\left[ 1 - \beta \, \left( \beta - \frac{T'}{L} \right) \right] \, a_-}
				{\beta^2 \, \frac{T'}{L} \, a_- + 
				\left[ 1 - \beta \, \left( \beta - \frac{T'}{L} \right) \right] \, a_+}\,.
\end{equation}
These are written in the form which makes them manifestly $< a_+$, but one can also easily verify that they are both $> a_-$ and that they are nested within the values \eqref{flux} and \eqref{Uhatflux} obtained in section \ref{sec:flows},
\begin{equation}\label{eq:fluxnesting}
			a_- 
		\	<\ \int_{D(A)}*V \mid_{\zV=0}
		\	<\ \int_{D(A)}*V \mid_{\zV=\zVm}
		\	<\ \int_{\I^+}*U \mid_{\zU=\zUm}
		\	<\ \int_{\I^+}*U  \mid_{\zU=0}
		\	<\ a_+
\end{equation}
In the regime where $\beta \, a_+ > a_-$ we have $\zVm \to -\infty$, and in this limit the  U-flux becomes
\begin{equation}\label{eq:Uminzinf}
	\int_{\I^+}*U = \frac{a_- \, L \, (1-\beta) + \beta \, T' \, (\beta \, a_+- a_-)}{(1-\beta) \, (L-\beta \, T')}
\end{equation}
The two expressions for U-flux, \eqref{eq:VUfluxopt} and \eqref{eq:Uminzinf}, of course agree when $\beta \, a_+= a_-$, and reduce to $a_- /(1-\frac{\beta \, T'}{L})$.

Also note that the optimized slopes of the $\phi$ and $\psi$ level sets in the stripe are 
$\zVm < 1$ and $1/\zUm$, which leads to the diffracted flow lines (i.e.\ the V and U threads) having slopes through the stripe
\begin{equation}\label{eq:optslopeVU}
V\text{-slope} = \frac{a_+ + \beta \, a_-}{\beta \, a_+ + a_-} = \frac{1}{\zVm} >1
	\qquad \text{and} \qquad
U\text{-slope} = \frac{a_+ - \beta \, a_-}{\beta \, a_+ - a_-} = \zUm <-1
\end{equation}
So we see that the V-threads are in fact normal to the constant-$\phi$ contours in the stripe and hence timelike there (but still automatically guaranteed to have smaller slope than $1/\beta$), whereas they are parallel to the constant-$\phi$ contours and hence spacelike outside of the stripe.
Similarly, the U-threads are normal to the constant-$\psi$ contours, and hence timelike, everywhere.

To see what this actually looks like in a specific example, let 
\begin{equation}\label{eq:STexvals}
L=1, \qquad
T=2, \qquad
T'=0.3, \qquad
\beta = 0.7, \qquad
a_-=1, \qquad
a_+=1.3
\end{equation}
\begin{figure}[tbp]
\centering
\includegraphics[width=0.2\textwidth]{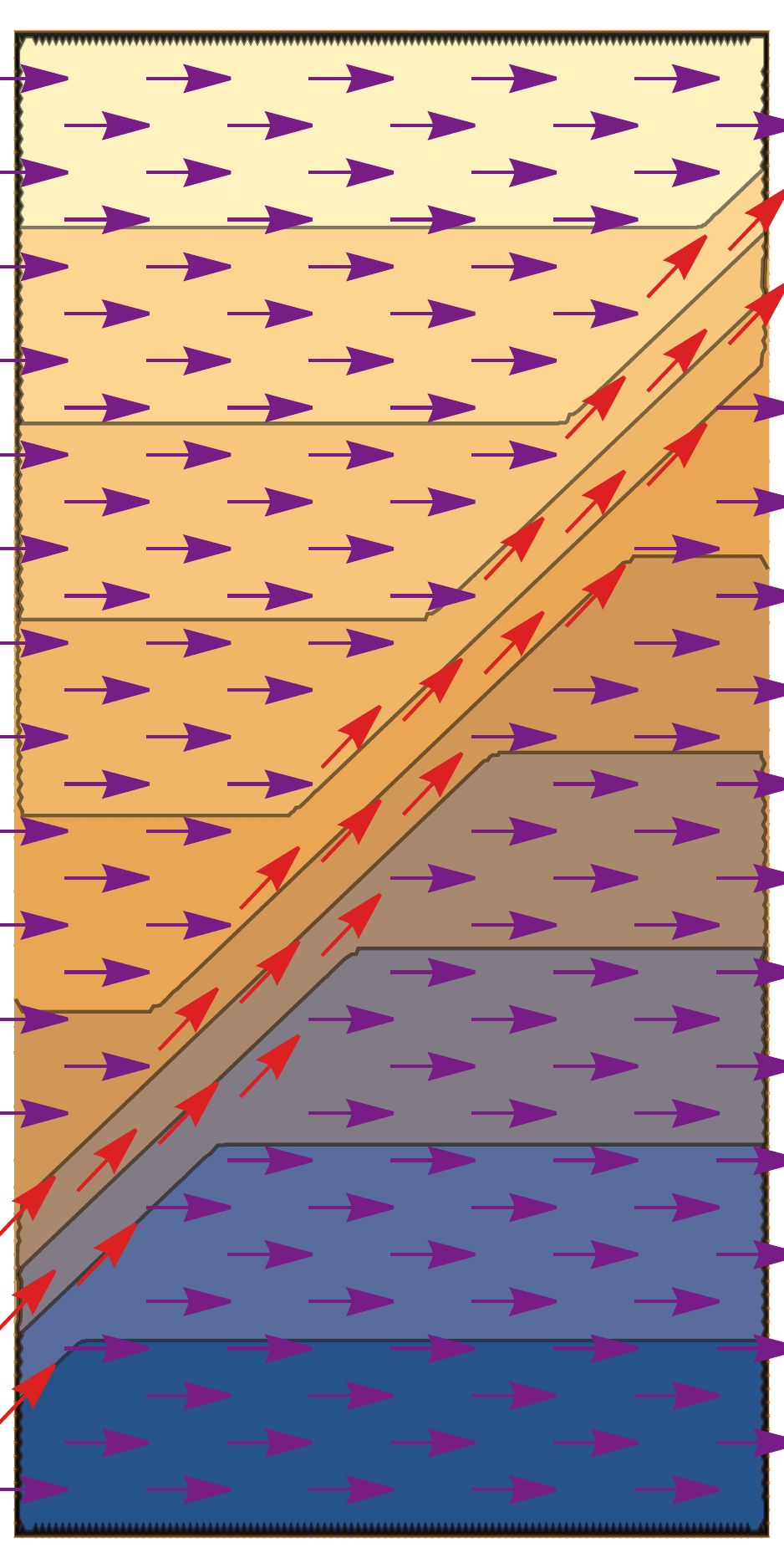}\hspace{1in}
\includegraphics[width=0.2\textwidth]{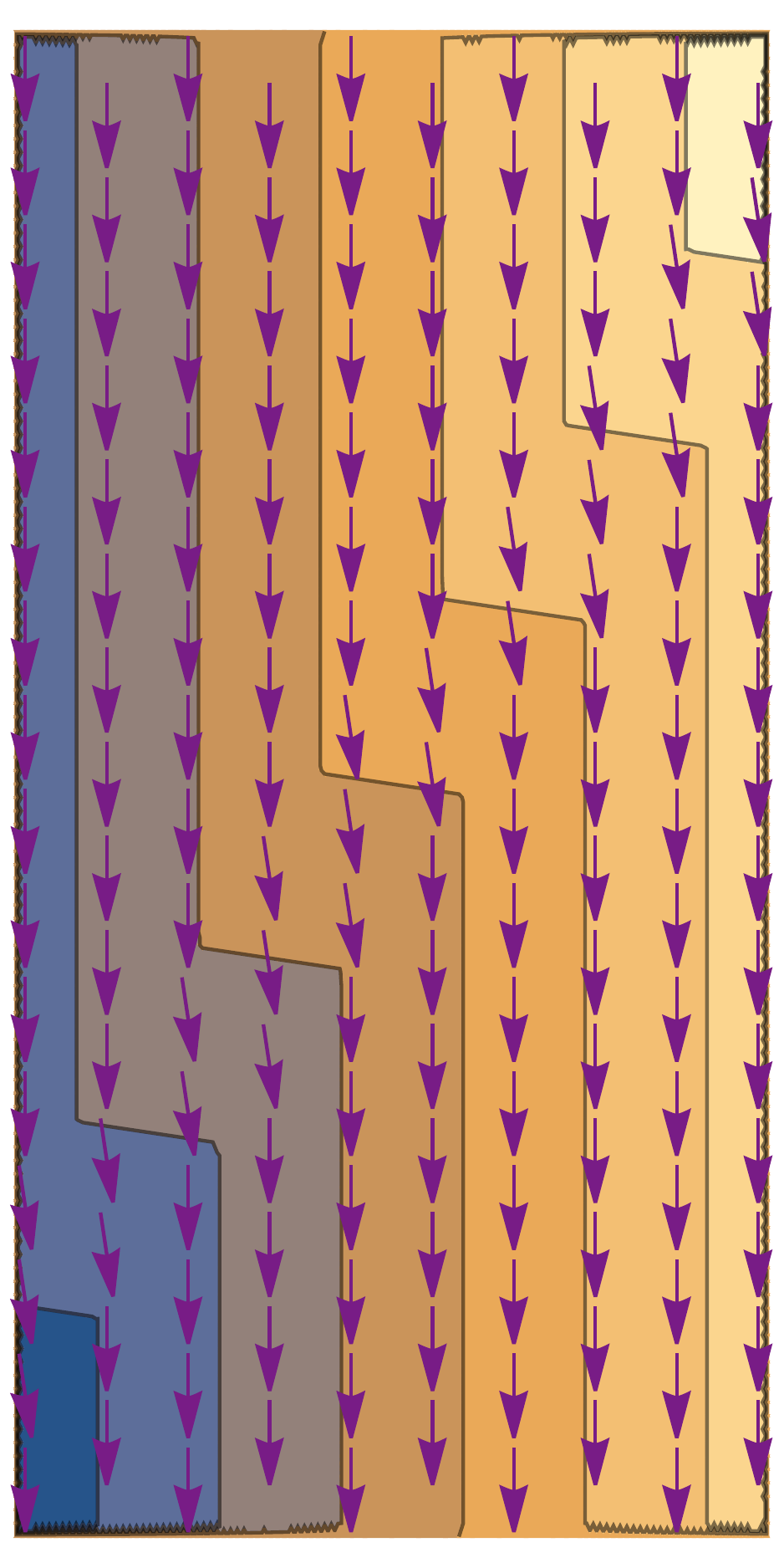}
\caption{\label{fig:STex}
Specific example of spacetime with parameters given in \eqref{eq:STexvals}, with piecewise-linear scalars and corresponding optimized flows.  Left:  level sets of $\phi$ and V-flow vector field, Right:  level sets of $\psi$ and U-flow vector field.
}
\end{figure}
Figure \ref{fig:STex} shows the corresponding spacetime, along with level sets of the scalars and the optimized flows.  The explicit values for the quantities in \eqref{eq:fluxnesting} (rounded to 3 decimal places) are respectively:
$\{1,\  1.019,\  1.084, \ 1.189, \ 1.21, \ 1.3 \}$.  We can see that they uphold the nesting of \eqref{eq:fluxnesting}, and that the gap between maximin and minimax has shrunk by almost a factor of 2; in particular,
$a_+-a_- = 0.3$,
\begin{equation}\label{eq:STexgap}
% a_+-a_- = 0.3, \qquad
\left( \int_{\I^+}*U - \int_{D(A)}*V \right)\arrowvert_{\zV=\zU=0} = 0.191, \qquad
\left( \int_{\I^+}*U - \int_{D(A)}*V \right)\arrowvert_{\zV=\zVm,\, \zU=\zUm}
= 0.105
\end{equation}
%

%---------------------------------------------------

\bibliographystyle{JHEP}
% \bibliography{refs} 

\providecommand{\href}[2]{#2}\begingroup\raggedright\endgroup

\end{document}